\documentclass{amsart}
\usepackage{epsf,amscd,amsmath,amssymb,amsfonts}
\usepackage{fullpage}
\newcommand{\be}{\begin{equation}}
\newcommand{\ee}{\end{equation}}
\newcommand{\bea}{\begin{eqnarray}}
\newcommand{\eea}{\end{eqnarray}}
\newcommand{\beas}{\begin{eqnarray*}}
\newcommand{\eeas}{\end{eqnarray*}}

\theoremstyle{plain}
\newtheorem{thm}{Theorem}
\newtheorem{lem}[thm]{Lemma}
\newtheorem{rem}[thm]{Remark}
\newtheorem{cor}[thm]{Corollary}
\newtheorem{prop}[thm]{Proposition}

\theoremstyle{definition}
\newtheorem{defn}[thm]{Definition}
\newtheorem{rmk}[thm]{Remark}
\newtheorem{rmks}[thm]{Remarks}
\newtheorem{ex}[thm]{Example}

\numberwithin{thm}{section}
\numberwithin{equation}{section}

\newcommand{\ve}{\varepsilon}
\newcommand{\eq}[2]{\begin{equation}\label{#1}#2 \end{equation}}
\newcommand{\ml}[2]{\begin{multline}\label{#1}#2 \end{multline}}
\newcommand{\ga}[2]{\begin{gather}\label{#1}#2 \end{gather}}

\newcommand{\surj}{\twoheadrightarrow}
\newcommand{\inj}{\hookrightarrow}

\newcommand{\Spec}{{\rm Spec \,}}

\newcommand{\sA}{{\mathcal A}}

\newcommand{\sC}{{\mathcal C}}

\newcommand{\sF}{{\mathcal F}}

\newcommand{\sH}{{\mathcal H}}

\newcommand{\sJ}{{\mathcal J}}
\newcommand{\sK}{{\mathcal K}}
\newcommand{\sL}{{\mathcal L}}

\newcommand{\sO}{{\mathcal O}}
\newcommand{\sP}{{\mathcal P}}

\newcommand{\sR}{{\mathcal R}}

\newcommand{\sT}{{\mathcal T}}

\newcommand{\sX}{{\mathcal X}}

\newcommand{\A}{{\mathbb A}}

\newcommand{\C}{{\mathbb C}}

\newcommand{\G}{{\mathbb G}}

\renewcommand{\P}{{\mathbb P}}
\newcommand{\Q}{{\mathbb Q}}
\newcommand{\R}{{\mathbb R}}

\newcommand{\Z}{{\mathbb Z}}

\def\One{\mathbb{I}}

\def\overl{\;\raisebox{-10mm}{\epsfxsize=30mm\epsfbox{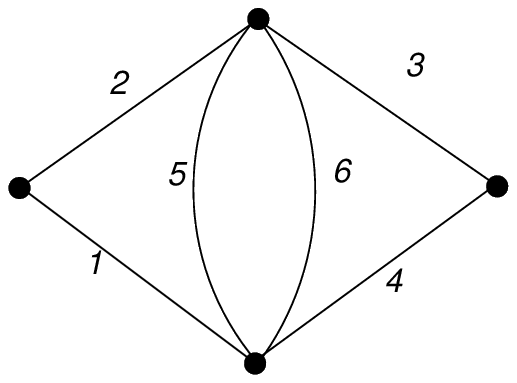}}\;}
\def\dunce{\;\raisebox{-16mm}{\epsfxsize=66mm\epsfbox{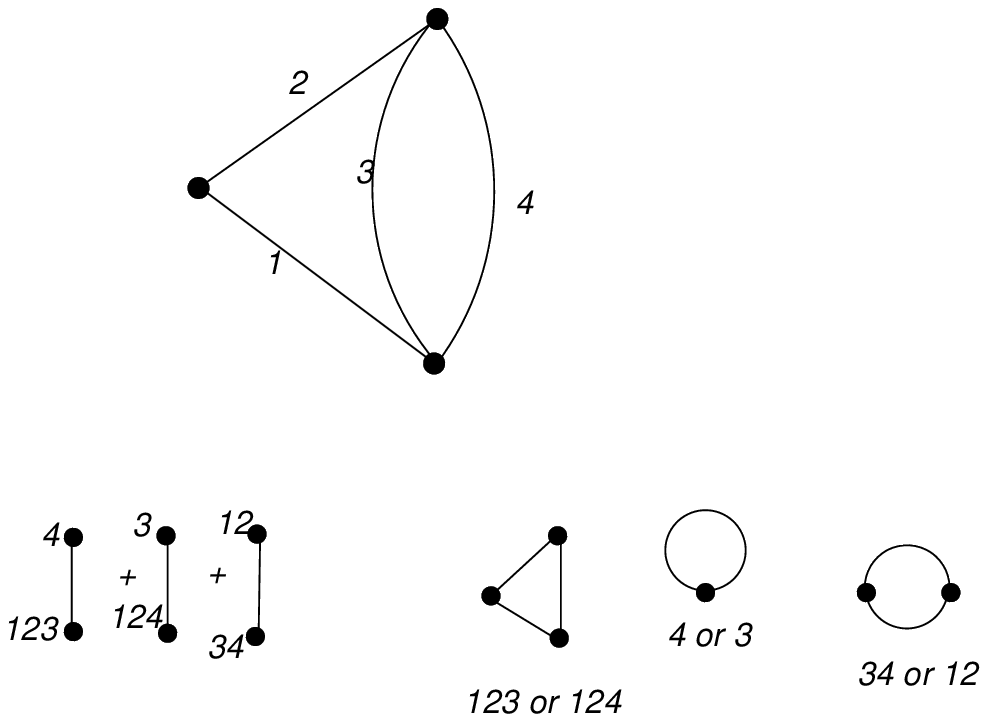}}\;}
\def\coreco{\;\raisebox{-16mm}{\epsfxsize=66mm\epsfbox{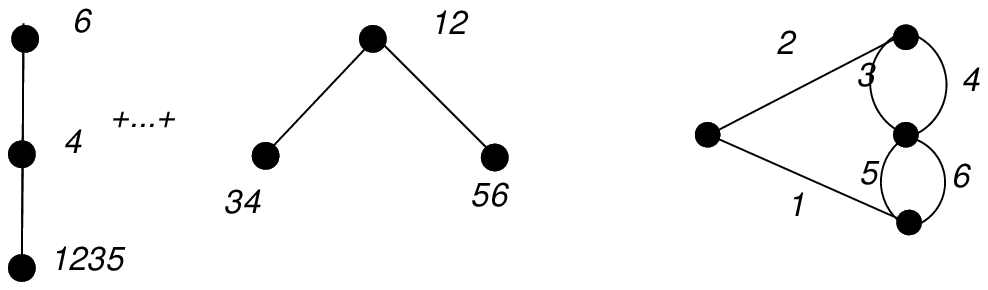}}\;}

\def\reconstr{\;\raisebox{-24mm}{\epsfxsize=100mm\epsfbox{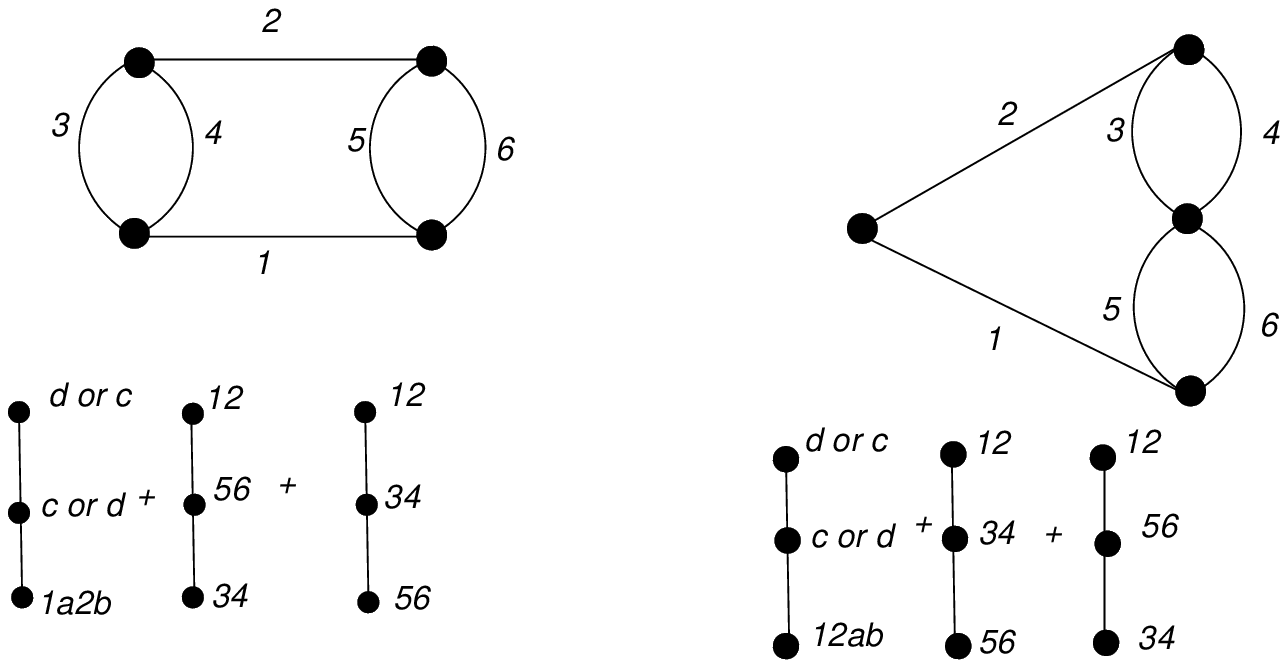}}\;}
\def\selfu{\;\raisebox{-20mm}{\epsfxsize=80mm\epsfbox{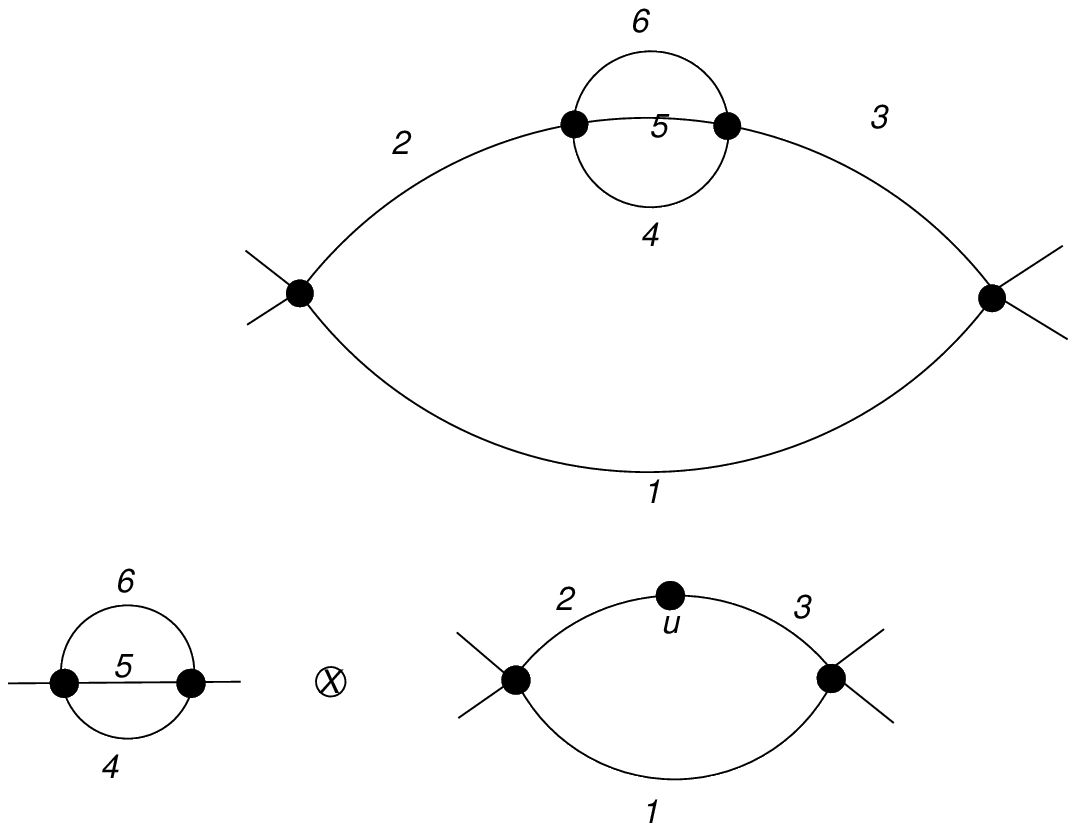}}\;}
\def\wthree{\;\raisebox{-8mm}{\epsfxsize=30mm\epsfbox{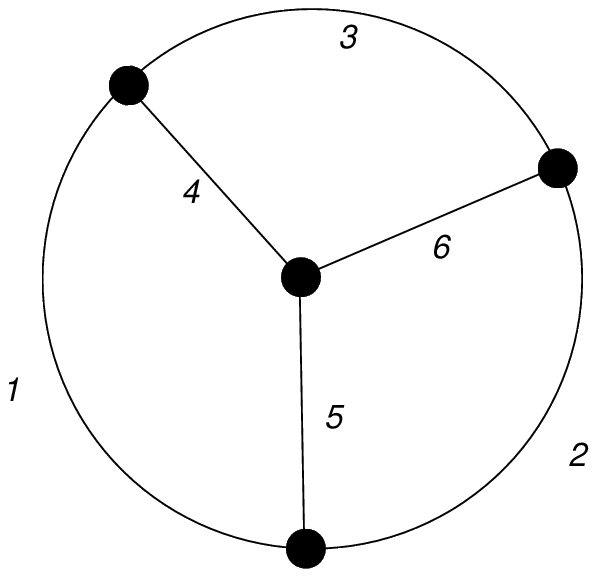}}\;}
\def\wsix{\;\raisebox{-4mm}{\epsfysize=12mm\epsfbox{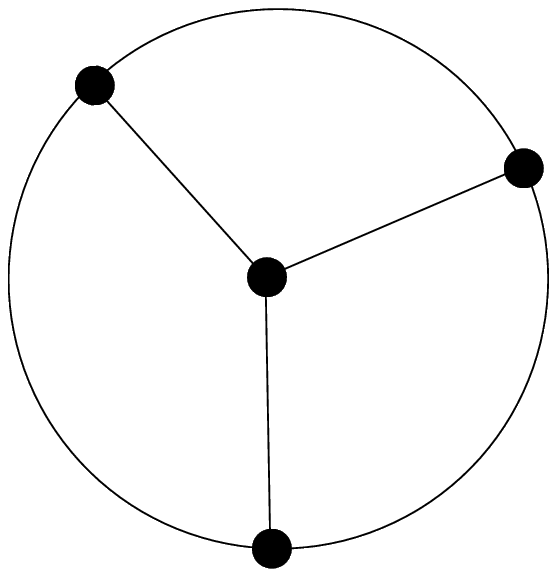}}\;}
\def\wthr{\;\raisebox{-4mm}{\epsfysize=12mm\epsfbox{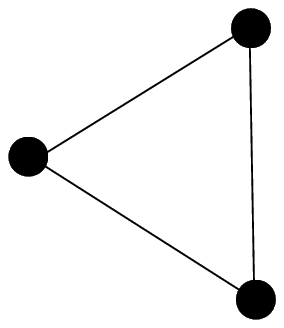}}\;}
\def\wtwo{\;\raisebox{-4mm}{\epsfysize=12mm\epsfbox{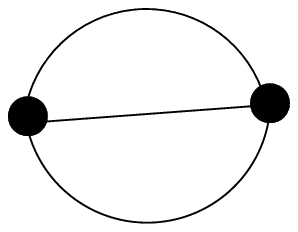}}\;}
\def\wtadtad{\;\raisebox{-4mm}{\epsfysize=12mm\epsfbox{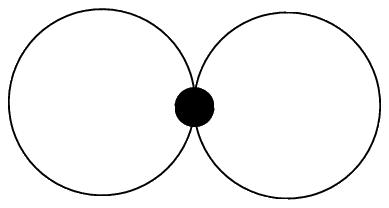}}\;}
\def\wfour{\;\raisebox{-4mm}{\epsfysize=12mm\epsfbox{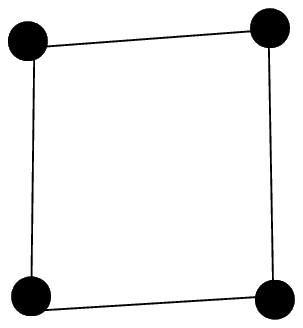}}\;}
\def\wfive{\;\raisebox{-4mm}{\epsfysize=12mm\epsfbox{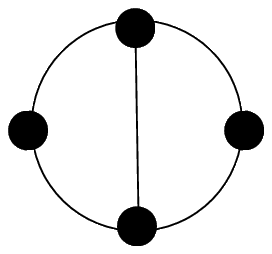}}\;}
\def\wtad{\;\raisebox{-4mm}{\epsfysize=12mm\epsfbox{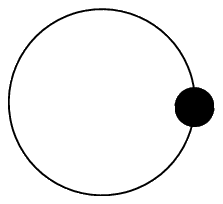}}\;}
\def\abcdef{\;\raisebox{-4mm}{\epsfysize=12mm\epsfbox{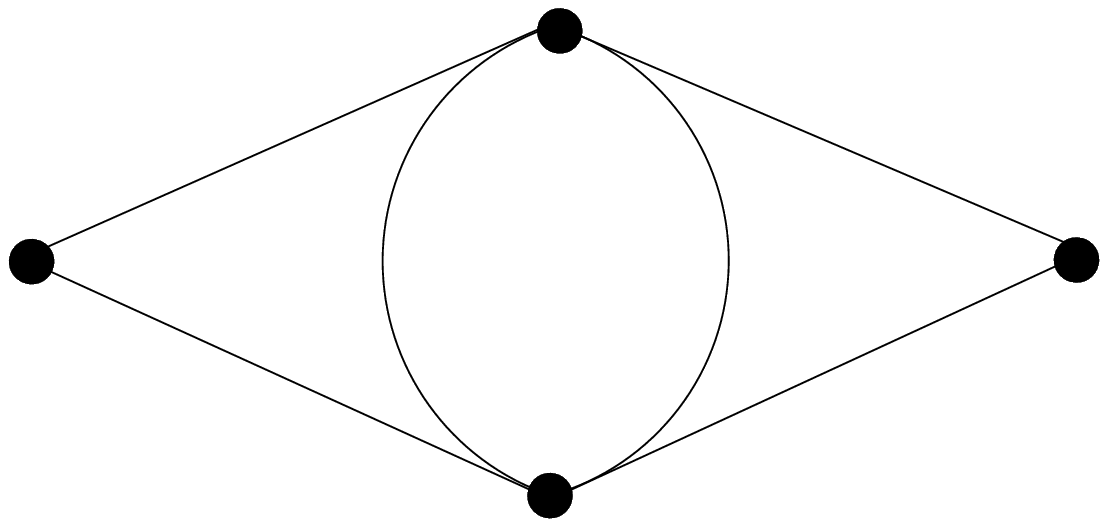}}\;}
\def\ef{\;\raisebox{-4mm}{\epsfysize=12mm\epsfbox{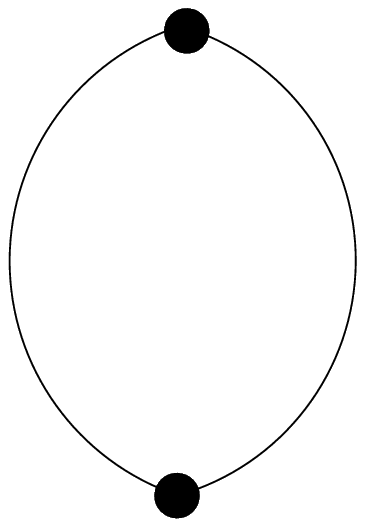}}\;}
\def\abcd{\;\raisebox{-4mm}{\epsfysize=12mm\epsfbox{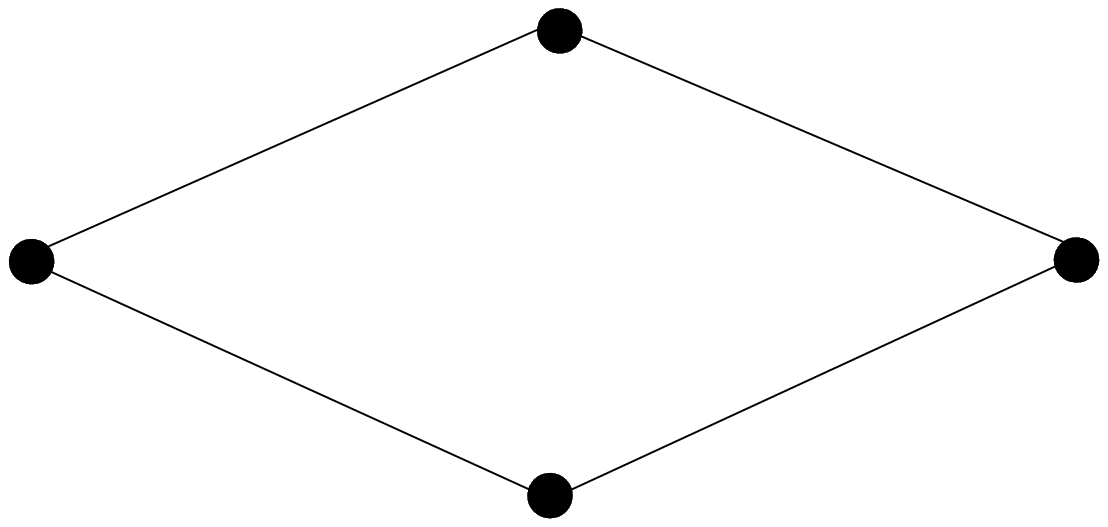}}\;}
\def\abe{\;\raisebox{-4mm}{\epsfysize=12mm\epsfbox{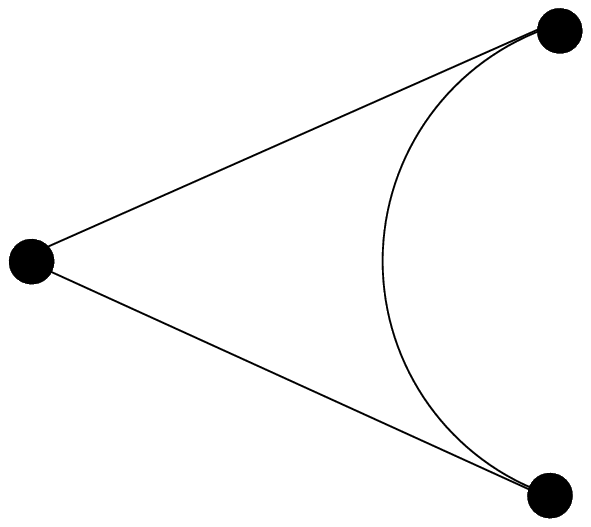}}\;}
\def\abf{\;\raisebox{-4mm}{\epsfysize=12mm\epsfbox{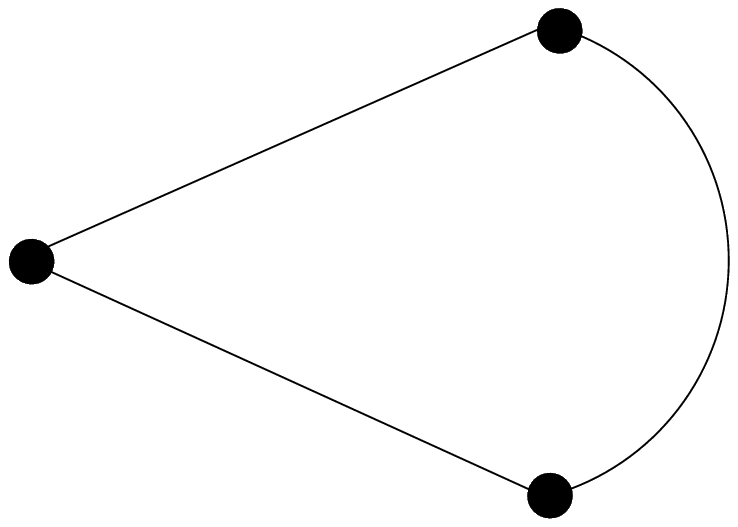}}\;}
\def\cdf{\;\raisebox{-4mm}{\epsfysize=12mm\epsfbox{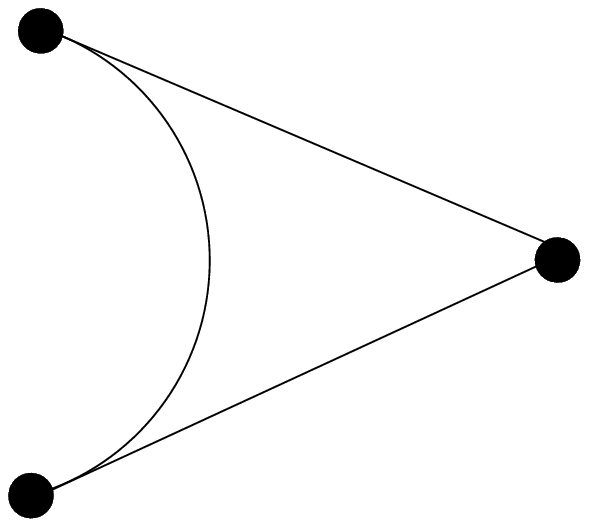}}\;}
\def\cde{\;\raisebox{-4mm}{\epsfysize=12mm\epsfbox{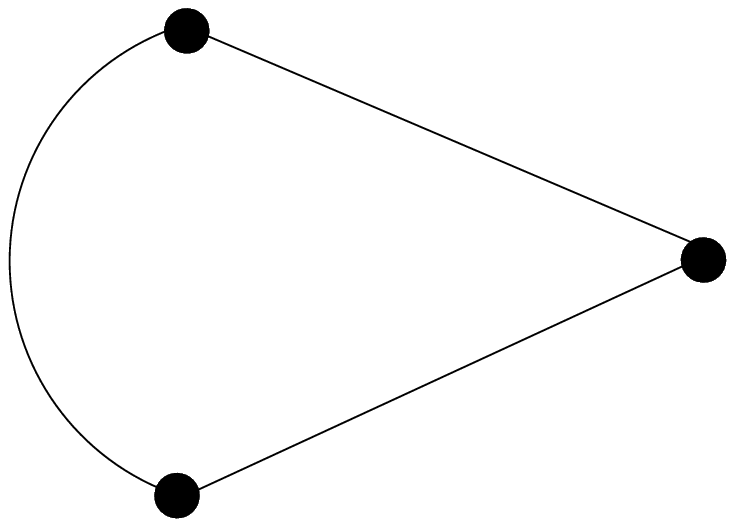}}\;}
\def\abef{\;\raisebox{-4mm}{\epsfysize=12mm\epsfbox{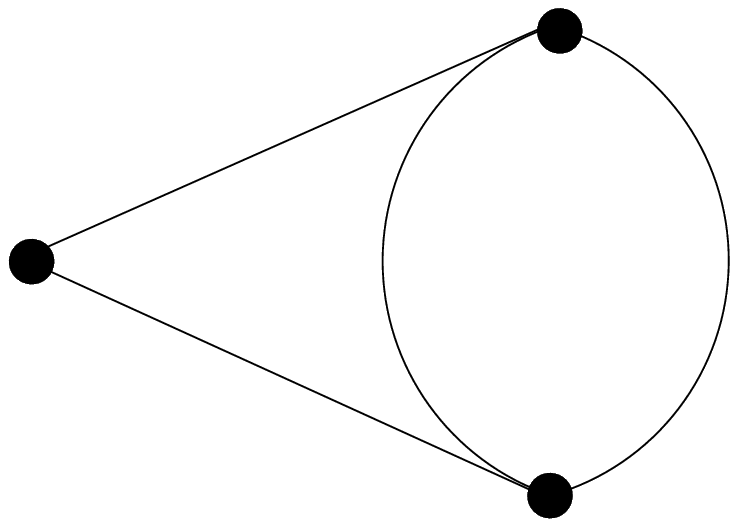}}\;}
\def\cdef{\;\raisebox{-4mm}{\epsfysize=12mm\epsfbox{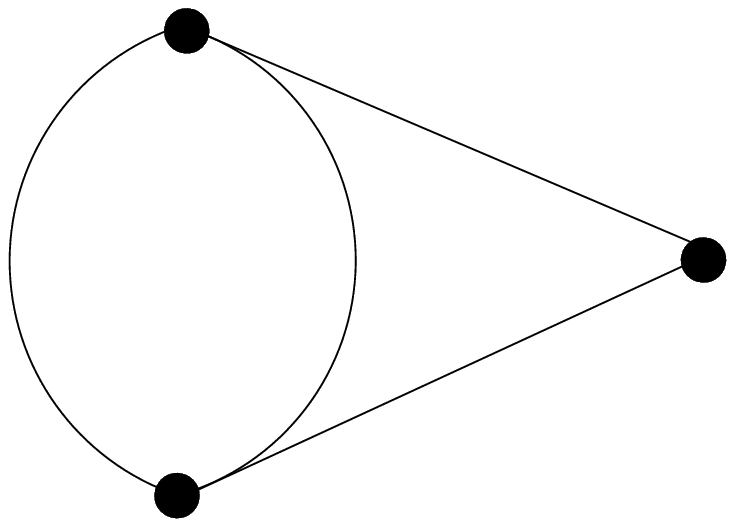}}\;}
\def\abcde{\;\raisebox{-4mm}{\epsfysize=12mm\epsfbox{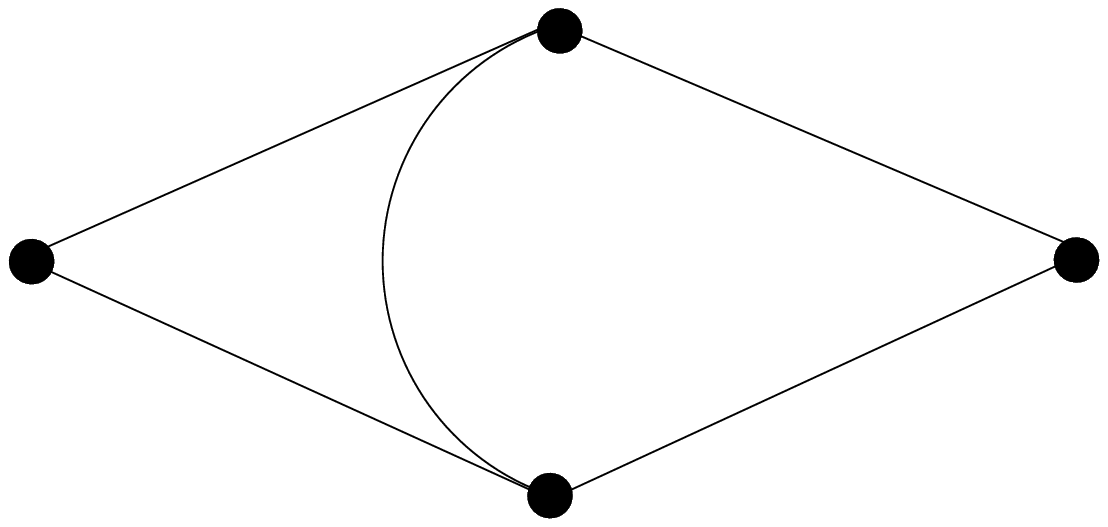}}\;}
\def\abcdf{\;\raisebox{-4mm}{\epsfysize=12mm\epsfbox{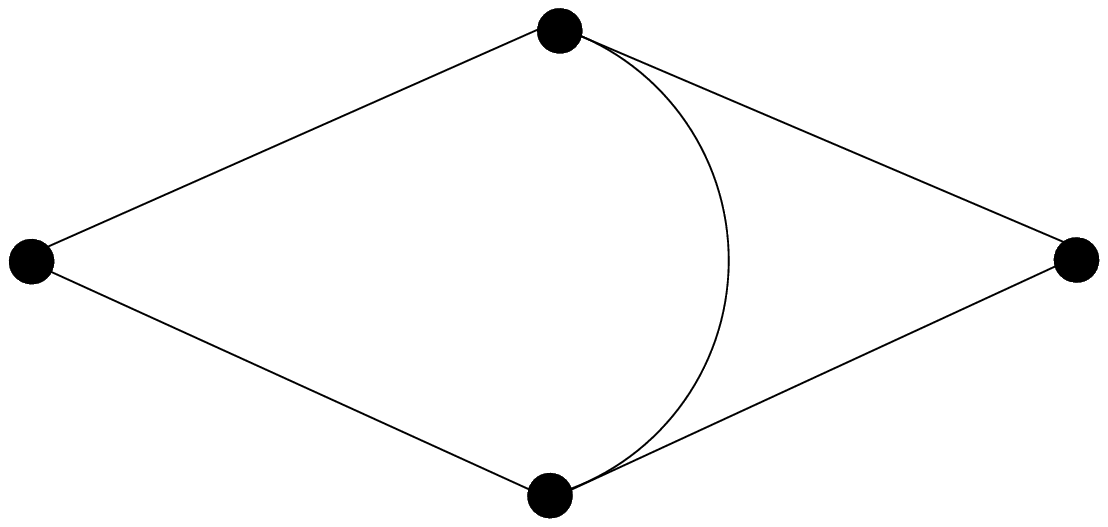}}\;}
\def\grapha{\;\raisebox{-3mm}{\epsfysize=8mm\epsfbox{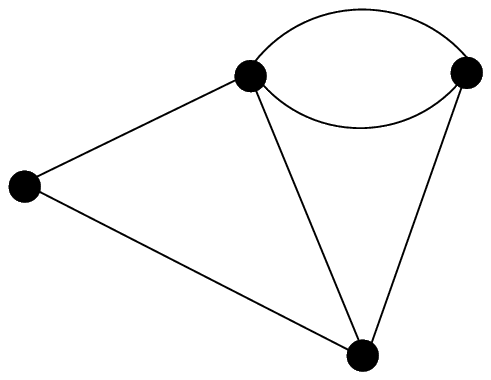}}\;}
\def\graphb{\;\raisebox{-3mm}{\epsfysize=8mm\epsfbox{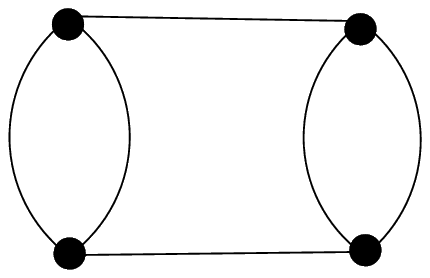}}\;}
\def\graphc{\;\raisebox{-3mm}{\epsfysize=8mm\epsfbox{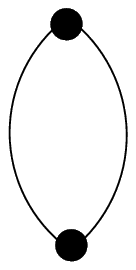}}\;}
\def\graphd{\;\raisebox{-3mm}{\epsfysize=8mm\epsfbox{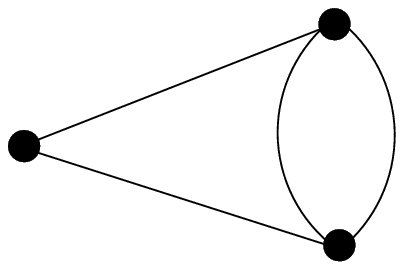}}\;}
\def\figa{\;\raisebox{-24mm}{\epsfxsize=80mm\epsfbox{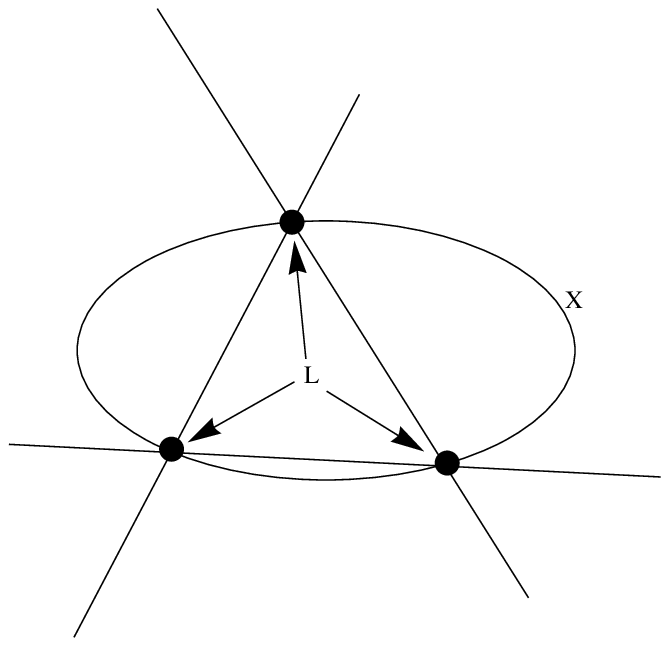}}\;}
\def\figb{\;\raisebox{-24mm}{\epsfxsize=80mm\epsfbox{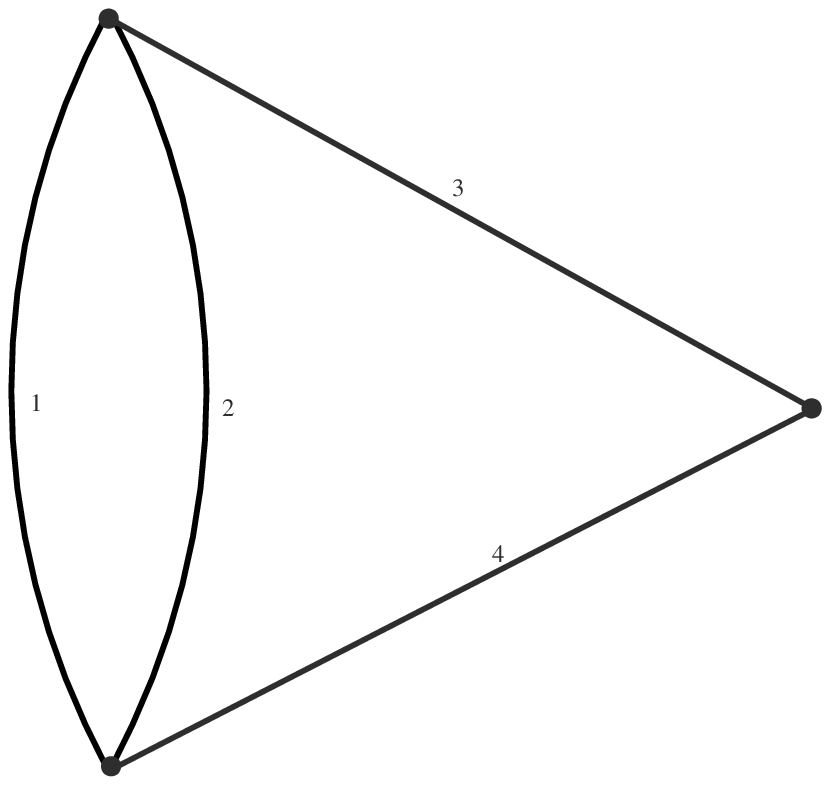}}\;}
\def\figc{\;\raisebox{-24mm}{\epsfxsize=80mm\epsfbox{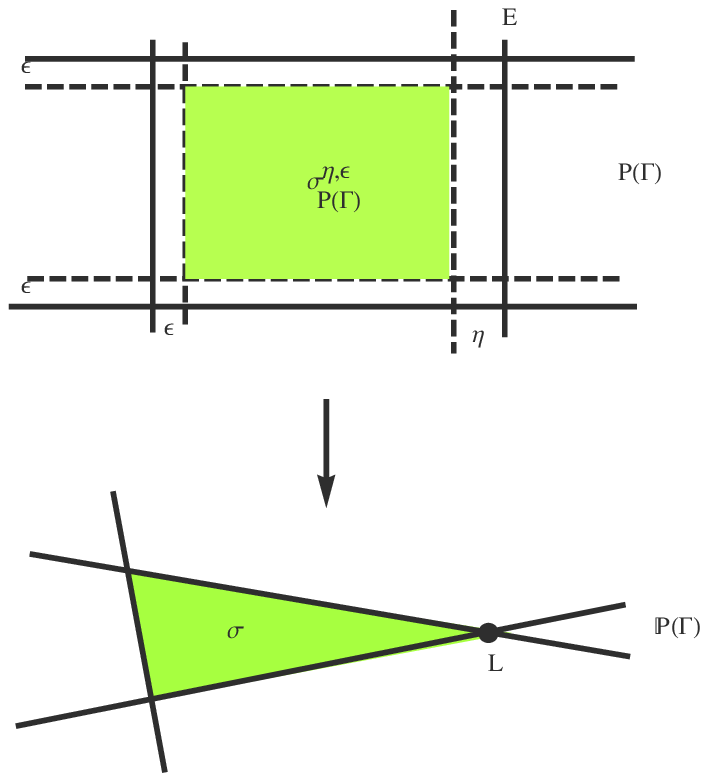}}\;}
\def\figd{\;\raisebox{-24mm}{\epsfxsize=80mm\epsfbox{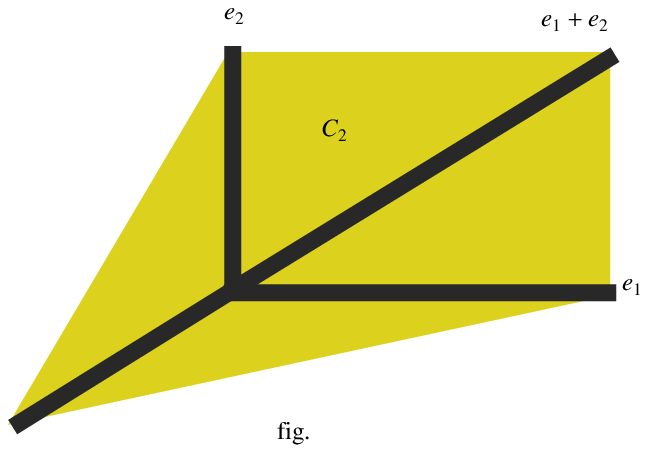}}\;}
\def\fige{\;\raisebox{-24mm}{\epsfxsize=80mm\epsfbox{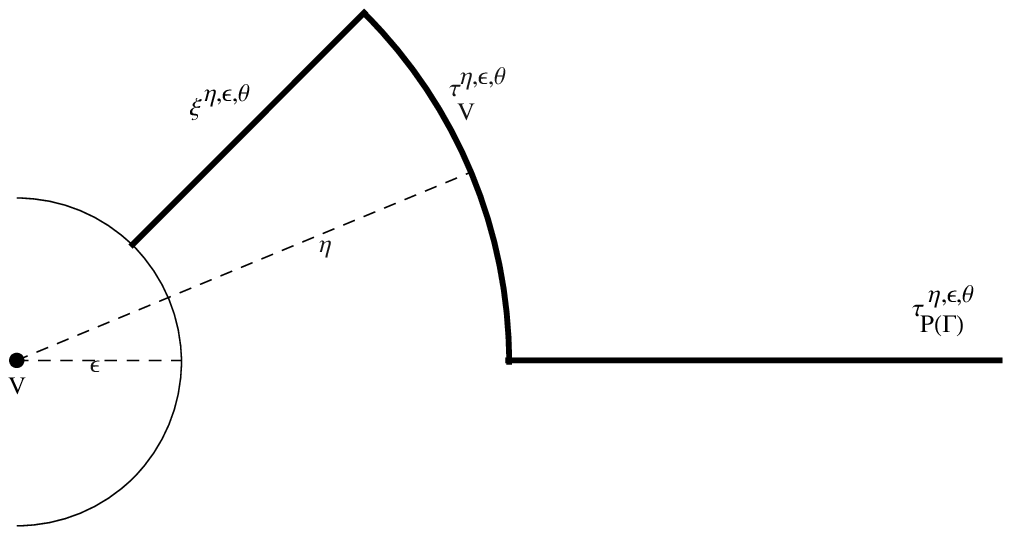}}\;}
\def\figf{\;\raisebox{-24mm}{\epsfxsize=80mm\epsfbox{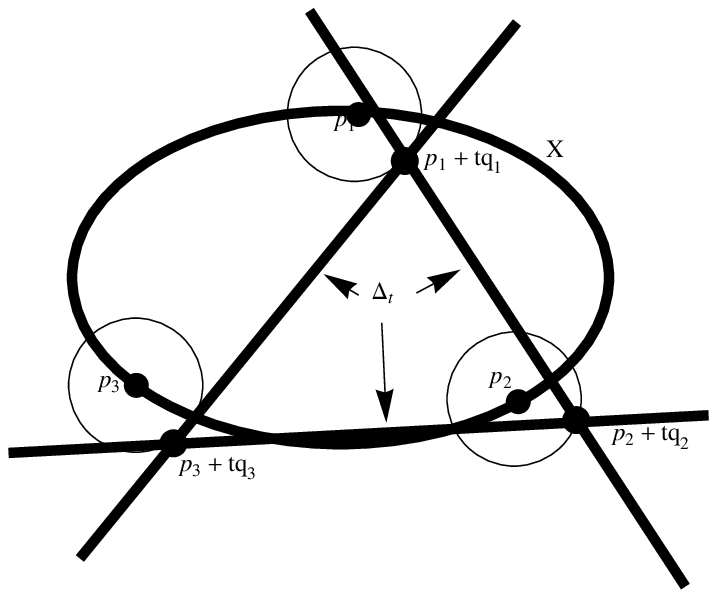}}\;}
\def\figg{\;\raisebox{-24mm}{\epsfxsize=80mm\epsfbox{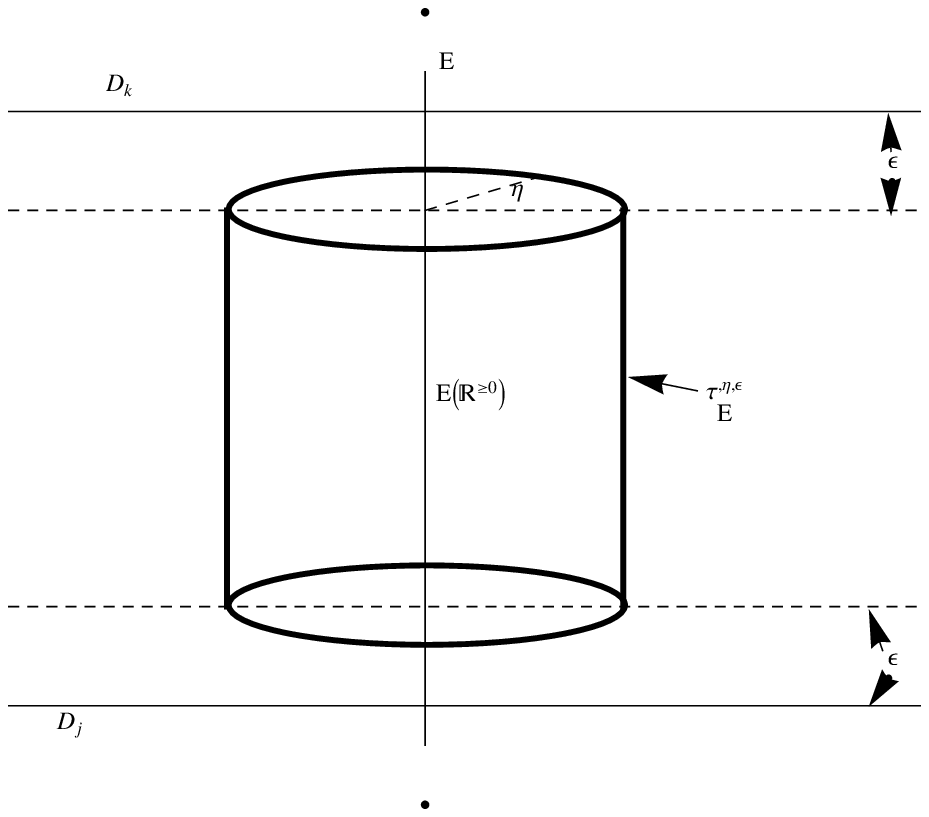}}\;}
\def\figh{\;\raisebox{-24mm}{\epsfxsize=80mm\epsfbox{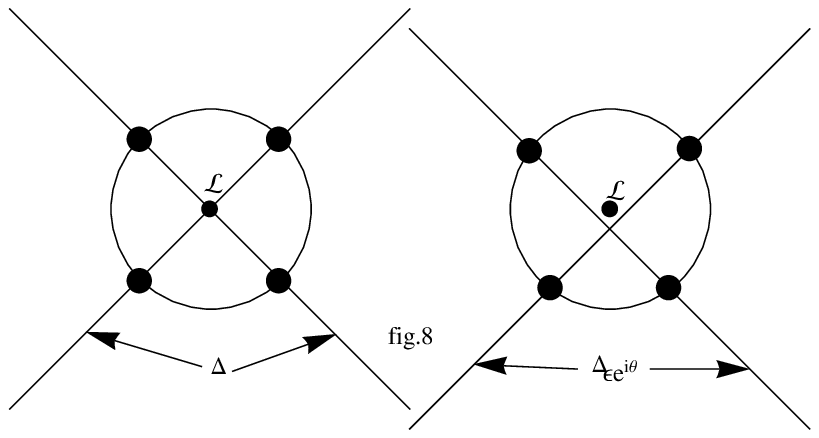}}\;}
\title{Mixed Hodge Structures and Renormalization in Physics}
\author{Spencer Bloch}\address{Dept. of Mathematics, University of Chicago, Chicago, IL 60637,
USA\\
E-mail address: bloch@math.uchicago.edu}
\author{Dirk Kreimer}
\address{CNRS-IHES, 91440 Bures sur Yvette, France and Center for Math.\ Phys.\, Boston U., Boston, MA 02215\\ E-mail address: kreimer@ihes.fr}
\begin{document}
\maketitle
\section{Introduction}
\subsection{} This paper is a collaboration between a mathematician and a physicist. It is based on the observation that renormalization of Feynman amplitudes in physics is closely related to the theory of limiting mixed Hodge structures in mathematics. Whereas classical physical renormalization methods involve manipulations with the integrand of a divergent integral, limiting Hodge theory involves moving the chain of integration so the integral becomes convergent and studying the {\it monodromy} as the chain varies.

Even methods like minimal subtraction in the context of dimensional or analytic regularization implicitly modify the integrand through the definition of a measure $\int d^Dk$ via analytic continuation. Still, as a regulator dimensional regularization is close to our approach in so far as it leaves the rational integrand assigned to a graph unchanged. Minimal subtraction as a renormalization scheme differs though from the renormalization schemes which we consider -momentum subtractions essentially- by a finite renormalization. Many of the nice algebro-geometric structures developed below are not transparent in that scheme.

The advantages of the limiting Hodge method are firstly that it is linked to a very central and powerful program in mathematics: the study of Hodge structures and their variations. As a consequence, one gains a number of tools, like weight, Hodge, and monodromy filtrations to study and classify the Feynman amplitudes. Secondly, the method depends on the integration chain, and hence on the graph, but it is in some sense independent of the integrand. For this reason it should adapt naturally e.g.\ to gauge theories where the numerator of the integrand is complicated.

An important point is to analyse the nature of the poles. Limiting mixed Hodge structures demand that the divergent subintegrals have at worst log poles. This does not imply that we can not apply our approach to perturbative amplitudes which have worse than logarithmic degree of divergence. It only means that we have to correctly isolate the polynomials in masses and external momenta which accompany those divergences such that the corresponding
integrands have singularities provided by log-poles. This is essentially automatic from the notion of a residue available by our very methods.
As a very pleasant byproduct, we learn that physical renormalization schemes -on-shell subtractions, momentum subtractions, Weinberg's scheme,- belong to a class of schemes for which this is indeed automatic.

Moreover, for technical reasons, it is convenient to work with projective rather than affine integrals. One of the central physics results in this paper is that the renormalization problem can be reduced to the study of logarithmically divergent, projective integrals.
This is again familiar from analytic regulators. The fact that it can be achieved here by leaving the integrand completely intact will hopefully
some fine day allow to understand the nature of the periods assigned to renormalized values in quantum field theory.

   A remark for Mathematicians: our focus in this paper has been renormalization, which is a problem arising in physics. We suspect, however, that similar methods will apply more generally for example to period integrals whenever the domain of integration is contained in $\R^{+n}$ and the integrand is a rational function with polar locus defined by a polynomial with non-negative real coefficients. The toric methods and the monodromy computations should go through in that generality.
\subsection*{Acknowledgments} Both authors thank Francis Brown, H\'el\`ene Esnault and Karen Yeats for helpful discussions. This work was partially supported by NSF grants DMS-0603781 and DMS-0653004. S.B.\ thanks the IHES for hospitality January-March 2006 and January-March 2008. D.K.\ thanks Chicago University for hospitality in February 2007.

\subsection{Physics Introduction}
This paper studies the renormalization problem in the context of parametric representations, with an emphasis on algebro-geometric properties.
We will not study the nature of the periods one obtains from renormalizable quantum field theories in an even dimension of space-time. Instead, we
provide the combinatorics of renormalization such that a future motivic analysis of renormalized amplitudes is feasible along the lines of \cite{BEK}.
Our result will in particular put renormalization in the framework of a limiting mixed Hodge structure, which hopefully provides a good starting point for an analysis of the periods in renormalized amplitudes. That these amplitudes are provided by numbers which are periods (in the sense of \cite{KZ})
is an immediate consequence of the properties of parametric representations, and will also emerge naturally below (see Thm.(\ref{cructhm})).

The main result of this paper is a careful study of the singularities of the first Kirchhoff--Symanzik polynomial, which carries all the short-distance singularities of the theory. The study of this polynomial can proceed via an analysis with the help of projective integrals.
Along the way, we will also give useful formulas for parametric representations involving affine integrals, and clarify the role of the second
Kirchhoff--Symanzik polynomial for affine and projective integrals.

Our methods are general, but in concrete examples we restrict ourselves to $\phi^4_4$ theory. Parametric representations are used which result from free-field propagators for propagation in flat space-time. In such circumstances, the advantages of analytic regularizations are also available
in our study of parametric representations as we will see. In particular, our use of projective integrals below combines such advantages with
the possibility to discuss    renormalization on the level of the pairing between integration chains and de Rham classes.

In examples, special emphasis is given to the study of  particular renormalization schemes, the momentum scheme (MOM-scheme, Weinberg's scheme, on-shell subtractions).

Also,  we often consider Green functions as functions of a single kinematical scale $q^2>0$. Green functions are defined throughout as the scalar coefficient functions (structure functions) for the radiative corrections to tree-level amplitudes $r$. They are to be regarded as scalar quantities of the form $1+{\mathcal O}(\hbar)$. Renormalized amplitudes are then, in finite order in perturbation theory, polynomial corrections  in $L=\ln q^2/\mu^2$
($\mu^2>0$) without constant term, providing the quantum corrections to the tree-level amplitudes appearing as monomials in a renormalizable Lagrangian \cite{Tor}:
\be \phi_R(\Gamma)=\sum_{j=1}^{{\textrm{aug}}(\Gamma)}p_j(\Gamma)L^j.\ee
 Correspondingly, Green functions become triangular series in two variables
\be G^r(\alpha,L)=1+\sum_{j=1}^\infty \gamma_j^r(\alpha)L^j=1+\sum_{j=1}^\infty c_j^r(L)\alpha^j.\ee
The series $\gamma_j^r(\alpha)$ are related by the renormalization group which leaves only the $\gamma_1^r(\alpha)$ undetermined,
while the polynomials $c_j^r(L)$ are bounded in degree by $j$. The series $\gamma_1^r$ fulfill ordinary differential equations driven by the primitive
graphs of the theory \cite{KrY}.

The limiting Hodge structure $ A(\Gamma)$ which we consider for each Feynman graph $\Gamma$ provides contribution of a graph $\Gamma$ to the coefficients of $\gamma_1^r$
in the limit. This limit is a period matrix (a column vector here) which has, from top to bottom, the periods provided by a renormalized graph $\Gamma$ as entries. The first entry is the contribution to $\gamma_1^r$ of a graph with $\textrm{res}(\Gamma)=r$ and the $k$-th is a rational multiple
of the contribution to $\gamma_k^r$. In section \ref{seclmhsvsren}
we determine the rational weights which connect these periods to the coefficients $p_j(\Gamma)$ attributed to the renormalization of a graph $\Gamma$.

We include a discussion of  the structure of renormalization which comes from an analysis of the second Kirchhoff--Symanzik polynomial. While this polynomial does not provide short-distance singularities in its own right, it leads to integrals of the form
\be \label{1.3a} \int \omega\ln(f)\ee
for a renormalized Feynman amplitude, with $\omega$ a de Rham class determined by the first Kirchhoff--Symanzik polynomial, and
$f$  -congruent to one along any remaining exceptional divisor-  determined by the second.  We do not actually do the monodromy calculation for integrals \eqref{1.3a} involving a logarithm, but it will be similar to the calculation for \eqref{0.2} which we do.
 A full discussion of the Hodge structure of a Green function seems feasible but will be postponed to future work.
\subsection{Math Introduction}\label{ssecmi}
Let $\P^{n-1}$ be the projective space of lines in $\C^n$ which we view as an algebraic variety with homogeneous coordinates $A_1,\dotsc,A_n$. Let $\psi(A_1,\dotsc,A_n)$ be a homogeneous polynomial of some degree $d$, and let $X \subset \P^{n-1}$ be the hypersurface defined by $\psi=0$. We assume the coefficients of $\psi$ are all real and $\ge 0$. Let $\sigma = \{[a_1,\dotsc,a_n]\ |\ a_i \ge 0, \forall i\}$ be the topological $(n-1)$-chain (simplex) in $\P^{n-1}$, where $[\ldots]$ refers to homogeneous coordinates. We will also use the notation $\sigma = \P^{n-1}(\R^{\ge 0})$. Our assumption about coefficients implies
\eq{0.1}{\sigma \cap X = \bigcup_{L\subset X} L(\R^{\ge 0}),
}
where $L$ runs through all coordinate    coordinate linear spaces $L:A_{i_1}=\cdots=A_{i_p}=0$ contained in $X$ (see (see Fig.\ref{figa})).
\begin{figure}[t]
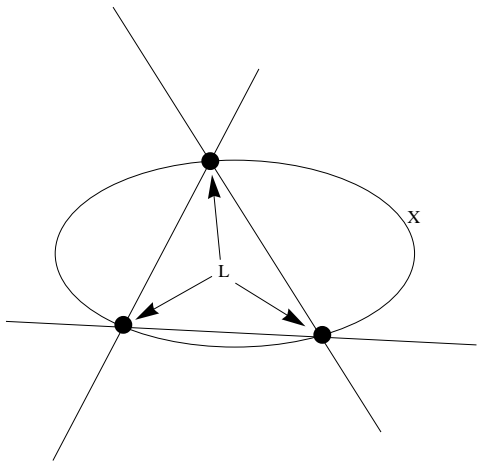
\centering
\figa\caption{Picture of $X$ and $L$}\label{figa}
\end{figure}
The genesis of the renormalization problem in physics is the need to assign values to integrals
\eq{0.2}{\int_\sigma \omega
}
where $\omega$ is an algebraic $(n-1)$-form on $\P^{n-1}$ with poles along $X$. The problem is an important one for physical applications, and there is an extensive literature (see, for example, \cite{ItzZ,Tod,Smirnov}) focusing on practical formulae to reinterpret \eqref{0.2} in some consistent way as a polynomial in $\log t$. (Here $t$ parametrizes a deformation of the integration chain. As a first approximation, one can think of $\int_t^\infty \omega$ when $\omega$ has a logarithmic pole at $t=0$.)

A similar problem arises in pure mathematics in the study of degenerating varieties, e.g. a family of elliptic curves degenerating to a rational curve with a node. In the classical setup, one is given a family $f:\sX \to D$, where $D$ is a disk with parameter $t$. The map $f$ is proper (so the fibres $X_t$ are compact). $\sX$ is assumed to be non-singular, as are the fibres $X_t,\ t\neq 0$. $X_0$ may be singular, though one commonly invokes resolution of singularities to assume $X_0\subset \sX$ is a normal crossing divisor.
Choose a basis $\sigma_{1,t},\dotsc,\sigma_{r,t}$ for the homology of the fibre $H_p(X_t,\Q)$ in some fixed degree $p$. By standard results in differential topology, the fibre space is locally topologically trivial over $D^* = D-\{0\}$, and we may choose the classes $\sigma_{i,t}$ to be locally constant. If we fix a smooth fibre $t_0 \neq 0$, the monodromy transformation $m: H_p(X_{t_0}) \to H_p(X_{t_0})$ is obtained by winding around $t=0$.
An important theorem (\cite{D}, III,2) says this transformation is quasi-unipotent, i.e. after possibly introducing a root $t' = t^{1/n}$ (which has the effect of replacing $m$ by $m^n$), $m-id$ is nilpotent. The matrix
\eq{0.3}{N:= \log m = -\Big[(id -m) + (id-m)^2/2 + \ldots\Big]
}
is thus also nilpotent. This is the mathematical equivalent of locality in physics. It insures that our renormalization of \eqref{0.2} will be a polynomial in $\log t$ rather than an infinite series.
We take a cohomology class $[\omega_t] \in H^p(X_t,\C)$ which varies algebraically. For example, in a family of elliptic curves $y^2=x(x-1)(x-t)$, the holomorphic $1$-form $\omega_t = dx/y$ is such a class. Note $\omega_t$ is single-valued over all of $D^*$. It is not locally constant. The expression
\eq{0.4}{\exp\Big(-(N\log t)/2\pi i\Big) \begin{pmatrix}\int_{\sigma_{1,t}}\omega_t \\  \vdots \\
\int_{\sigma_{r,t}}\omega_t \end{pmatrix}.
}
is then single-valued and analytic on $D^*$.
Suppose $\omega_t$ chosen such that the entries of the column vector in \eqref{0.4} grow at worst like powers of $|\log |t||$ as $|t| \to 0$. A standard result in complex analysis then implies that \eqref{0.4} is analytic at $t=0$. We can write this
\eq{0.5}{\begin{pmatrix}\int_{\sigma_{1,t}}\omega_t \\  \vdots \\
\int_{\sigma_{r,t}}\omega_t \end{pmatrix} \sim \exp\Big((N\log t)/2\pi i\Big)\begin{pmatrix}a_1 \\ \vdots \\a_r \end{pmatrix}.
}
Here the $a_j$ are constants which are periods of a {\it limiting Hodge structure}. The exponential on the right expands as a matrix whose entries are polynomials in $\log t$, and the equivalence relation $\sim$ means that the difference between the two sides is a column vector of (multi-valued) analytic functions vanishing at $t=0$.

We would like to apply this program to the integral \eqref{0.2}. Let $\Delta : \prod_1^n A_j = 0$ be the coordinate divisor in $\P^{n-1}$. Note that the chain $\sigma$ has boundary in $\Delta$, so as a first attempt to interpret \eqref{0.2} as a {\it period}, we might consider the pairing
\eq{0.6}{H^{n-1}(\P^{n-1}-X,\Delta-X\cap \Delta)\times H_{n-1}(\P^{n-1}-X,\Delta-X\cap \Delta) \to \C
}
The form $\omega$ is an algebraic $(n-1)$-form and it vanishes on $\Delta$ for    degree reasons, so it does give a class in the    relative cohomology group appearing in \eqref{0.6} (see the discussion \eqref{10.8}-\eqref{10.10}). On the other hand, the chain $\sigma$ meets $X$ \eqref{0.1}, so we do not get a class in homology. Instead we consider a family of coordinate divisors $\Delta_t: \prod_1^n A_{j,t}=0$ with $\Delta_0 = \Delta$. (For details, see section \ref{secmono}.) For $t=\ve>0$ there is a natural chain $\sigma_\ve$ which is what the physicists would call a cutoff. We have $\partial \sigma_\ve \subset \Delta_\ve$ and $\sigma_\ve \cap X = \emptyset$, so $\int_{\sigma_\ve} \omega$ is defined.
One knows on abstract grounds that the monodromy of
$$H_{n-1}(\P^{n-1}-X,\Delta_t-X\cap \Delta_t)
$$
is quasi-unipotent as above (\cite{D}, III,\S 2). The main mathematical work in this paper will be to compute the monodromy of $\sigma_\ve$ in the specific case of Feynman amplitudes in physics. More precisely, $X$ will be a graph hypersurface $X_\Gamma$ associated to a graph $\Gamma$ (section \ref{secgrhyp}).
We will write down chains $\tau_{\gamma}^\ve$, one for each flag of {\it core} ({\it one particle irreducible} in physics) subgraphs $\gamma= \{\Gamma_1 \subsetneq \cdots \Gamma_{p(\gamma)} \subsetneq \Gamma\}$, representing linearly independent homology classes in $H_{n-1}(\P^{n-1}-X, \Delta_\ve - X\cap \Delta_\ve)$. (The combinatorics here is similar to that found in \cite{BF}, \cite{L}.) We will show the monodromy in our case is given by
\eq{}{m(\sigma_\ve) = \sigma_\ve + \sum_\gamma (-1)^{p(\gamma)}\tau_\gamma^\ve.
}
We will then exhibit a nilpotent matrix $N$ such that
\eq{}{\begin{pmatrix} m(\sigma_\ve) \\ \vdots \\ m(\tau_\gamma^\ve)  \\ \vdots\end{pmatrix} = \exp(N)\begin{pmatrix} \sigma_\ve \\ \vdots \\ \tau_\gamma^\ve  \\ \vdots\end{pmatrix}.
}
With this in hand, renormalization is automatic for any physical theory for which $\Gamma$ and its subgraphs are at worst {\it logarithmically divergent} after taking out suitable polynomials in masses and momenta. Namely, such a physical theory gives a differential form $\omega_\Gamma$ as in \eqref{0.2} and we may repeat the above argument:
\eq{0.9}{\exp(-(N\log t)/2\pi i)\begin{pmatrix} \int_{\sigma_t}\omega \\ \vdots \\ \int_{\tau_\gamma^t}\omega  \\ \vdots\end{pmatrix}
}
is single-valued on the punctured disk. The hypothesis of log
divergence at worst for subgraphs of $\Gamma$ will imply that the
integrals will grow at worst like a power of $\log$ as $|t| \to
0$,(lemma \ref{lem5.2}). Precisely as in \eqref{0.5}, one gets the renormalization
\eq{0.10}{\int_{\sigma_t}\omega_\Gamma = \sum_{k=0}^r b_k(\log t)^k + O(t),
}
where $O(t)$ denotes a (multi-valued) analytic function vanishing at $t=0$. The renormalization schemes considered here can be characterized
by the condition $b_0=0$.

Of course, the requirement that a physical theory have at worst log divergences is a very strong constraint. The difficult computations in section 7 show how general divergences encountered in physics can be reduced to log divergences.

\begin{rmks}The renormalization scheme outlined above, and worked out
  in detail in the following sections, has a number of properties,
  some of which may seem strange to the physicist. \newline\noindent
(i) It does not work in renormalization schemes which demand counter-terms which are not defined by subtractions at fixed values of masses and momenta of the theory. So conditions on the regulator for example, as in minimal subtraction where one defines the counterterm by projection onto a pole part, are not considered. In such schemes, and for graphs which are worse than log divergent, a topological procedure of
the sort given here can not work. It is necessary instead to modify
the integrand $\omega_\Gamma$ in a non-canonical way. \newline\noindent
(ii) On the other hand, our approach is very canonical. It depends on
the choice of a parameter $t$, as any renormalization scheme
must. Somewhat more subtle is the dependence on the monodromy
associated to the choice of a family
$\Delta_t$ of coordinate divisors deforming the given
$\Delta=\Delta_0$. We have taken the most evident such monodromy,
moving all the vertices of the simplex.  Note that this choice is stable in the sense that a small deformation leaves the monodromy unchanged.\newline\noindent
(iii) It would seem that our answer is much more complicated than need
be, because $\Gamma$ will in general contain far more core subgraphs
than divergent subgraphs. For example, in $\varphi^4$-theory, the
``dunce's cap'' (see Fig.\ref{figb})
\begin{figure}[t]
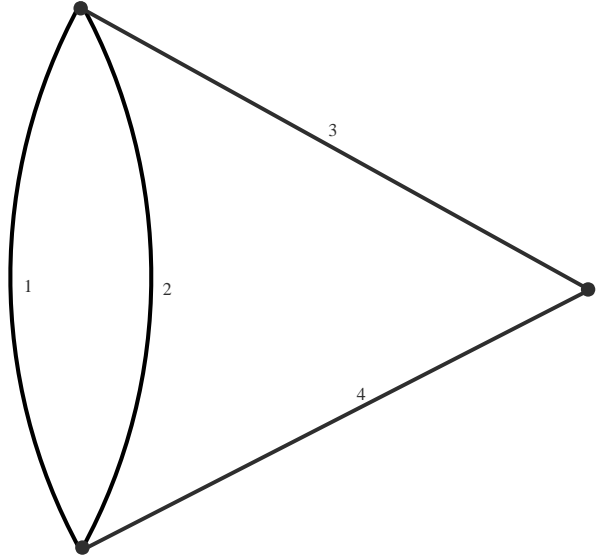
\centering
\figb\caption{Dunce's cap. Here and in following figures, external half-edges are often not drawn and are determined by the requirement that all vertices are four-valent.}\label{figb}
\end{figure}
has only one divergent subgraph, given in
the picture by edges $1,2$. It has $3$ core subgraphs
$(3,4,1),(3,4,2),(1,2)$. From the point of view of renormalization,
this problem disappears. The $\tau^\ve_\gamma$ are tubes, and the
integral $\int_{\tau^\ve_\gamma} \omega_\Gamma$ is basically a residue
which will vanish unless $\gamma\subset \Gamma$ is a divergent
subgraph. In \eqref{0.9}, the column vector of integrals will consist
mostly of $0$'s and the final regularization \eqref{0.10} will involve
only divergent subgraphs. \newline\noindent
(iv) An important property of the theory is the presence of a {\it
  limiting mixed Hodge structure}. The constants on the right hand
side of \eqref{0.5} are periods of a mixed Hodge structure called
the limiting MHS for the degeneration. One may hope that the
tendency for Feynman amplitudes to be multi-zeta numbers \cite{BK}
will some day be understood in terms of this Hodge structure. From
the point of view of this paper, the vector space $W \subset
H_{n-1}(\P^{n-1}-X_\Gamma, \Delta_t-X_\Gamma\cap \Delta_t)$
spanned by $\sigma_t$ and the $\tau^t_\gamma$ is invariant under
the monodromy. One may ask whether the image of $W$ in the
limiting MHS spans a sub-Hodge structure. If so, we would expect
that this HS would be linked to the multi-zeta numbers. Note that
$W$ is highly non-trivial even when $\Gamma$ has no
subdivergences. This $W$ is an essentially new invariant which
comes out of the monodromy. See section (\ref{seclmhs}) for a
final discussion of our viewpoint.
\newline\noindent (v) There are a number of renormalization
schemes in physics, some of which are not compatible with our
approach. One general test is that our scheme depends only on the
graph polynomials of $\Gamma$. For example, suppose $\Gamma =
\Gamma_1\cup \Gamma_2$ where the $\Gamma_i$ meet at a single
vertex. Then the renormalization polynomial in $\log t$ our theory
yields for $\Gamma$ will be the product of the renormalizations
for the $\Gamma_i$.
\end{rmks}
Most of the mathematical work involved concerns the calculation of monodromy for a particular topological chain. It is perhaps worth taking a minute to discuss a toy model. Suppose one wants to calculate $\int_0^\infty \omega$, where $\omega = \frac{dz}{(z-i)z}$. The integral diverges, so instead we consider $\int_t^\infty \omega$ as a function of $t=\ve e^{i\theta}$ for $0 \le \theta\le 2\pi$. If we take the path $[t,\infty]$ to be a great circle, then as $t$ winds around $0$, the path will get tangled in the singularity of $\omega$ at $z=i$. Assuming we do not understand the singularities of our integral far from $0$, this could be a problem. Instead we chose our path to follow the small circle from $\ve e^{i\theta}$ to $\ve$ and then the positive real axis from $\ve$ to $\infty$. The {\it variation} of monodromy is the difference in the paths for $\theta=0$ and $\theta = 2\pi$. In this case, it is the circle $\{|t|=\ve\}$. If we assume something (at worst superficial log divergence for the given graph and all subgraphs in the given physical theory) about the behavior of $\omega$ near the pole at $0$, then the behavior of our integral for $|t|<<1$ is determined by this monodromy, which is a topological invariant. This is quite different from the usual approach in physics involving complicated algebraic manipulations with $\omega$.
A glance at fig.(\ref{fige})  suggests that our toy model is too simple. We have to work with two scales, $\ve<<\eta<<1$. This is because in the more complicated situation, we have to deal with cylinders of small radius $\eta$, but then we have further to slightly deform the boundaries of the cylinder (cf.\ fig.(\ref{figg})).
\subsection{Leitfaden} Section \ref{secha} is devoted to Hopf algebras of graphs and of trees. These have played a central role in the combinatorics of renormalization. In particular, the insight afforded by passing from graphs to trees is important. Since the combinatorics of core subgraphs is even more complicated than that of divergent subgraphs, it seemed worth going carefully through the construction. Section \ref{seccomb} studies the toric variety we obtain from a graph $\Gamma$  by blowing up certain coordinate linear spaces in the projective space with homogeneous coordinates labeled by the edges of $\Gamma$. The orbits of the torus action are related to flags of core subgraphs of the given graph. In section \ref{sectopch}, we use the $\R$-structure on our toric variety to construct certain topological chains which will be used to explicit the monodromy. Section \ref{secgrhyp} recalls the basic properties of the graph polynomial $\psi_\Gamma\equiv\psi(\Gamma)$ and the graph hypersurface $X_\Gamma: \psi_\Gamma=0$. The crucial point is corollary \ref{cor3.3} which says that the strict transform of $X_\Gamma$ on our toric blowup avoids points with coordinates $\ge 0$. Any chain we construct which stays close to the locus of such points necessarily is away from $X_\Gamma$ and hence also away from the polar locus of our integrand. Section \ref{secmono} computes the monodromy of our chain. Section \ref{parrep} considers how to reduce Feynman amplitude calculations as they arise in physics, including masses and momenta as well as divergences which are worse than logarithmic, to the basic situation where limiting methods can apply. In section \ref{secn} we calculate the nilpotent matrix $N$ which is the log of the monodromy transformation, and in section \ref{secld} we prove the main renormalization theorem in the log divergent case, to which we have reduced the theory.
\section{Hopf algebras of trees and graphs}\label{secha}
\subsection{Graphs}\label{subsecgr}     In this section we bring together material on graphs and the graph Hopf algebra which will be used in the sequel. We also discuss Hopf algebras related to rooted trees and prove a result (proposition \ref{proptree}) relating the Hopf algebra of core graphs to a suitable Hopf algebra of labeled trees. Strictly speaking this is not used in the paper, but it provides the best way we know to understand flags of core subgraphs, and these play a central role in the monodromy computations. Trees labeled by divergent subgraphs have a long history in renormalization theory \cite{overl}, \cite{Kreimer}.

A graph $\Gamma$  is determined by giving a finite set $HE(\Gamma)$ of half-edges, together with two further sets $E(\Gamma)$ (edges) and $V(\Gamma)$ (vertices) and surjective maps
\eq{}{p_V: HE(\Gamma) \to V;\quad p_E: HE(\Gamma) \to E.
}
(Note we do not allow isolated vertices.) In combinatorics, one typically assumes all fibres $p_E^{-1}(e)$ consist of exactly two half-edges ($e$ an {\it internal edge}), while in physics the calculus of path integrals and correlation functions dictates that one admit {\it external edges} $e\in E$ with $\# p_E^{-1}(e) = 1$. If all internal edges of $\Gamma$ are shrunk to $0$, the resulting graph (with no internal edges) is called the {\it residue} $\text{res}(\Gamma)$. In certain theories, the vertices are decomposed into different types $V=\amalg V_i$, and the valence of the vertices in $V_i$, $\# p_{V}^{-1}(v)$, is fixed independent of $v\in V_i$.

We will typically work with labeled graphs which are triples $(\Gamma, A, \phi: A \cong E(\Gamma))$. We refer to $A$ as the set of edges.

A graph is a topological space with Betti numbers $|\Gamma| = h_1(\Gamma) = \dim H_1(\Gamma,\Q)$ and $h_0(\Gamma)$. We say $\Gamma$ is {\it connected} if $h_0=1$. Sometimes $h_1$ is referred to as the {\it loop number}.

A subgraph $\gamma \subset \Gamma$ is determined (for us) by a subset $E(\gamma) \subset E(\Gamma)$. We write $\Gamma/\!/\gamma$ for the quotient graph obtained by contracting all edges of $\gamma$ to points. If $\gamma$ is not connected, $\Gamma/\!/\gamma$ is different from the naive quotient $\Gamma/\gamma$. If $\gamma=\Gamma$, we take $\Gamma/\!/\Gamma = \emptyset$ to be the empty set. It will be convenient when we discuss Hopf algebras below to have the empty set as a graph.

Also, for $\gamma=e$ a single edge, we have the contraction $\Gamma/\!/e = \Gamma/e$. In this case we also consider the {\it cut} graph $\Gamma-e$ obtained by removing $e$ and also any remaining isolated vertex.

A graph $\Gamma$ is said to be {\it core} ($1PI$ in physics terminology) if for any edge $e$ we have $|\Gamma-e|<|\Gamma|$.

A {\it cycle} $\gamma \subset \Gamma$ is a core subgraph such that $|\gamma|=1$. If $\Gamma$ is core, it can be written as a union of cycles (see e.g. the proof of lemma 7.4 in \cite{BEK}).

\subsubsection{Self-energy graphs}
Special care has to be taken when the residue $\textrm{res}(\gamma)$ of a connected component  $\gamma$ of some subgraph consists of two half-edges connected to a vertex, $|\textrm{res}(\gamma)|=2$. Such graphs are called self-energy graphs in physics.
In such a situation, if the internal edges of $\gamma$ contract to a point, we are left with two edges in $\Gamma/\!/\gamma$, which are connected at this point $u$ . It might happen that the theory  provides more than one two-point vertex. In fact, for a
massive theory, there are two two-point vertices provided by the theory corresponding to the two monomials in the Lagrangian quadratic in the fields,
we call them of mass and kinetic type. $\Gamma/\!/\gamma$ represents then a sum over two graphs by summing over the two types of vertices for that point $u$. (see Fig.(\ref{selfu}) for an example).

\begin{figure}[t]
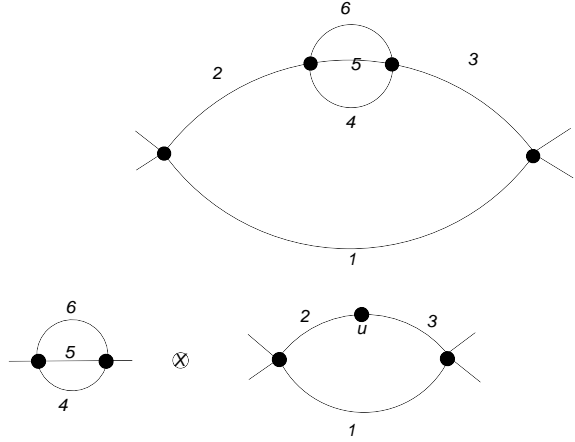
\centering\selfu
\caption{This vertex graph has a propagator correction given by edges $4,5,6$. The non-trivial part of the coproduct then delivers on the left
the subgraph with internal edges $4,5,6$ amongst other terms. The coproduct on the right has a co-graph on edges $1,2,3$. There is a two-point vertex $u$ between edges $2,3$. Choosing two labels $u=m^2$ or $u=\Box$ allows to distinguish between mass and wave-function renormalization.
We remind the reader that the corresponding monomials in the Lagrangian are $m^2\phi^2/2$ and $\phi\Box\phi/2$.}\label{selfu}
\end{figure}

The edges and vertices of various types have weights. We set the weight of an edge to be two, the weight of a vertex with valence greater than two is zero,
the weight of a vertex of mass type is zero, the weight of the kinetic type is +2.

Then, the superficial degree of divergence $\textrm{sdd}(\Gamma)$ for a connected core graph $\Gamma$ is
\be \textrm{sdd}(\Gamma)=4|\Gamma|-2|\Gamma^{[1]}|+2|\Gamma^{[0],\textrm{kin}}|,\label{sdd}\ee
where $\Gamma^{[0],\textrm{kin}}$ is the set of vertices of kinetic type,
and $\Gamma^{[1]}$ the set of internal edges. $\Gamma^{[0]}$, the set of interaction vertices (for which we assume we have only one type) does not show up as they have weight zero, nor does $\Gamma^{[0],\textrm{mass}}$.
By $|\cdots|$ we denote the cardinality of these sets.

Note that a graph $\Gamma/\!/\gamma$ which has one two-point vertex labeled $m^2$ (of mass type) which appears after contracting a self-energy subgraph $\gamma$
has an improved power-counting as its edge number is $2h_1(\Gamma/\!/\gamma)+1$. If the two-point vertex is labeled by $\Box$ (kinetic type), it has not changed though: $\textrm{sdd}(\Gamma/\!/\gamma)=\textrm{sdd}(\Gamma)$, as the weight of the two-point vertex compensates for the weight of the extra propagator. Quite often, in massless theories, one then omits the use of these two-point vertices altogether.

\subsection{Hopf algebras of graphs}\label{subsechag} Let $\sP$ be a class of graphs. We assume $\emptyset \in \sP$ and that $\Gamma \in \sP$ and $\Gamma'\cong \Gamma$ implies $\Gamma' \in \sP$. We say $\sP$ is {\it closed under extension} if given $\gamma \subset \Gamma$ we have
\eq{2.2b}{\gamma, \Gamma \in \sP \Leftrightarrow \gamma, \Gamma/\!/\gamma \in \sP.
}

Easy examples of such classes of graphs are $\sP = \text{core graphs}$, and $\sP = \text{log divergent graphs}$, where  $\Gamma$ is  log divergent (in $\phi^4_4$ theory) if it is core and if further $\# E(\Gamma_i)= 2|\Gamma_i|$ for every connected component $\Gamma_i\subset \Gamma$. (Both examples are closed under extension by virtue of the identity $|\gamma|+|\Gamma/\!/\gamma| = |\Gamma|$.) Examples which arise in physical theories are more subtle. Verification of \eqref{2.2b} requires an analysis of which graphs can arise from a given Lagrangian.
To verify $\sP = \{\Gamma\ |\ \textrm{sdd}(\Gamma) \ge 0\}$ satisfies \eqref{2.2b} one must consider self-energy graphs and the role of vertices of kinetic type as discussed above.

In particular, in  massless $\phi_4^4$ theory divergent graphs are closed under extension, and so is the class of graphs for which
$4|\Gamma|-2|\Gamma^{[1]}|+2|\Gamma^{[0],\textrm{kin}}|+2|\Gamma^{[0],\textrm{mass}}|\geq 0$. Note that this may contain superficially convergent graphs if there are sufficiently many two-point vertices of mass type.
It pays to include them in the class of graphs to be considered, which enables one to discuss the effect of mass in the renormalization group flow.

Associated to a class $\sP$ which is closed under extension as above, we define a (commutative, but not cocommutative) Hopf algebra $H = H_\sP$ as follows. As a vector space, $H$ is freely spanned by isomorphism classes of graphs in $\sP$. (A number of variants are possible. One may work with oriented graphs, for example. In this case, the theory of graph homology yields a (graded commutative) differential graded Hopf algebra. One may also rigidify by working with disjoint unions of subgraphs of a given labeled graph.) $H$ becomes a commutative algebra with $1=[\emptyset]$ and product given by disjoint union. Define a comultiplication $\Delta: H \to H\otimes H$:
\eq{2.3b}{\Delta(\Gamma) = \sum_{\substack{\gamma\subset\Gamma \\ \gamma \in \sP}} \gamma\otimes \Gamma/\!/\gamma.
}
One checks that \eqref{2.2b} implies that \eqref{2.3b} is coassociative. Since $H$ is graded by loop numbers and each $H_n$ is finite dimensional, the theory of Hopf algebras guarantees the existence of an antipode, so $H$ is a Hopf algebra.

If $\sP'\subset \sP$ with Hopf algebras $H', H$ (   e.g.\ take $\sP$ to be core graphs, and $\sP'\subset \sP$ divergent core graphs) then the map $H \surj H'$ obtained by sending $\Gamma\mapsto 0$ if $\Gamma \not\in \sP'$ is a homomorphism of Hopf algebras. For example, the divergent Hopf algebra carries the information needed for renormalization \cite{Kreimer}, while the core Hopf algebra $H_\sC$ determines the monodromy. In terms of groupschemes, one has $\Spec(H_{\text{log. div.}}) \inj \Spec(H_{\sC})$ is a closed subgroupscheme, and renormalization can be viewed as a morphism from the affine line with coordinate $L$ to $\Spec(H_{\text{log. div.}})$. Already here we use that for divergent graphs with $\textrm{sdd}(\Gamma)>0$, we can evaluate them as polynomials in masses and external momenta with coefficients determined from log divergent graphs, see below.

Let $\Gamma_i, i=1,2$ be core graphs (a similar discussion will be valid for other classes of graphs) and let $v_i \in \Gamma_i$ be vertices. Let $\Gamma = \Gamma_1\cup \Gamma_2$ where the two graphs are joined by identifying $v_1\sim v_2$. Then $\Gamma$ is core (cf. proposition \ref{prop1.4}). Further, core subgraphs $\Gamma' \subset \Gamma$ all arise as the image of $\Gamma_1'\amalg\Gamma_2' \to \Gamma$ for $\Gamma_i'\subset \Gamma_i$ core. Thus
\ml{}{\Delta(\Gamma) = \sum \Gamma'\otimes \Gamma/\!/\Gamma' = \Big(\sum \Gamma_1'\otimes \Gamma_1/\!/\Gamma_1' \Big)\Big(\sum \Gamma_2'\otimes \Gamma_2/\!/\Gamma_2' \Big)+\\
\sum(\Gamma'-\Gamma_1'\cdot\Gamma_2')\otimes(\Gamma_1/\!/\Gamma_1' \cdot \Gamma_2/\!/\Gamma_2') + \sum \Gamma'\otimes \Big(\Gamma/\!/\Gamma' - \Gamma_1/\!/\Gamma_1' \cdot \Gamma_2/\!/\Gamma_2' \Big).
}
It follows that the vector space $I\subset H_\sC$ spanned by elements $\Gamma - \Gamma_1\cdot\Gamma_2$ as above satisfies $\Delta(I) \subset I\otimes H_\sC + H_\sC\otimes I$. Since $I$ is an ideal, we see that $\overline H_\sC := H_\sC/I$ is a commutative Hopf algebra. Roughly speaking, $\overline H_\sC$ is obtained from $H_\sC$ by identifying {\it one vertex reducible} graphs with products of the component pieces.

Generalization to theories with more vertex and edge types are straightforward.

\begin{figure}[t]
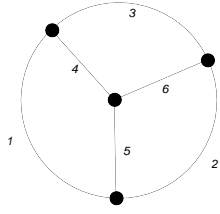

\wthree\caption{In Eq.(\ref{copw3}), we give the coproduct for this wheel with three spokes in the core Hopf algebra.}\label{wthree}
\end{figure}
Fig.(\ref{wthree}) gives the wheel with three spokes. This graph, which in $\phi^4$ theory (external edges to be added such that each vertex is four-valent) has a residue $6\zeta(3)$ for conceptual reasons \cite{BEK}, has a coproduct (we omit edge labels and identify terms which are identical under this omission, which gives the indicated multiplicities)
\bea \Delta\left(\wsix\right)  =  \wsix & \otimes & \One\label{copw3}\\+\One & \otimes &\wsix\nonumber\\
  + 4 \wthr & \otimes & \wtwo\nonumber\\
   +3\wfour & \otimes & \wtadtad\nonumber\\
    6 \wfive & \otimes & \wtad\nonumber.
\eea
For example, the three possible labelings for the four-edge cycle in the third line are $4523$, $5631$ and $6412$.
While the graph has a non-trivial coproduct in the core Hopf algebra, it is a primitive element in the renormalization Hopf algebra.
It is tempting to hope that the core coproduct relates to the Hodge structure underlying the period which appears in the residue of this graph.

\subsection{Rooted tree Hopf algebras \cite{overl}, \cite{BerK2}}\label{ssecrth} We introduce the Hopf
algebra of decorated non-planar rooted trees $H_{\mathcal T}$ using
non-empty finite sets as decorations (decorations will be sets of edge labels of Feynman graphs below)
to label the vertices of the rooted tree Hopf algebra $H_{\mathcal T}(\emptyset)$. Products in $H_\sT$ are disjoint unions of trees (forests). We write the coproduct as
\be \Delta(T)=T\otimes \One+\One\otimes T+\sum_{\textrm{admissible cuts $C$}}P^C(T)\otimes R^C(T).\ee
Edges are oriented away from the root and a vertex which has no outgoing edge we call a foot.
An admissible cut is a subset of edges of a tree such that no path
from the root to any vertex of $T$ traverses more than one element of
that subset. Such a cut $C$ separates $T$ into at least $2$ components. The component containing the root is denoted $R^C(T)$, and the product of the other components is $P^C(T)$.

A {\it ladder} is a tree without side branching. Decorated ladders
generate a sub-Hopf algebra $L_{\sT} \subset H_{\sT}$. A general element in $L_{\sT}$ is a sum of {\it bamboo forests}, that is disjoint unions of ladders. Decorated ladders have an associative shuffle product
\eq{2.7b}{L_1\star L_2 := \sum_{k\in\text{shuffle}(\ell_1,\ell_2)} L(k)
}
where $\ell_i$ denotes the ordered set of decorations for $L_i$ and $\text{shuffle}(\ell_1,\ell_2)$ is the set of all ordered sets obtained by shuffling together $\ell_1$ and $\ell_2$.
\begin{lem}\label{lem2.1b} Let $\sK \subset L_\sT$ be the ideal generated by elements of the form $L_1\cdot L_2 - L_1\star L_2$. Then $\Delta(\sK) \subset \sK\otimes L_\sT + L_\sT \otimes \sK$.
\end{lem}
\begin{proof}Write $\Delta(L_i) = \sum_{j=0}^{d_i} L_{ij}\otimes L_i^{d_i-j}$ where $d_i$ is the length of $L_i$ and $L_{ij}$ (resp. $L_i^j$) is the bottom (resp. top) subladder of length $j$. Then
\ga{2.8b}{\Delta(L_1)\Delta(L_2) = \sum_{j,\mu}L_{1j}L_{2\mu}\otimes L_1^{d_1-j}L_2^{d_2-\mu} \\
\Delta(L_1\star L_2) = \sum_k \Delta(L(k)) = \sum_{k,\nu} L(k)_\nu\otimes L(k)^{d_1+d_2-\nu}. \notag
}
Consider pairs $(j,\mu)$ of indices in \eqref{2.8b} and write $j+\mu=\nu$. Among the pairs $k,\nu$ we consider the subset $K(j,\mu)$ for which the first $\nu=j+\mu$ elements of the ordered set consist of a shuffle of the decorations on the ladders $L_{1j}, L_{2\mu}$. It is clear that the remaining $d_1+d_2-\nu$ elements of $k$ will then run through shuffles of the decorations of $L_1^{d_1-j}, L_2^{d_2-\mu}$, so
\ml{}{\Delta(L_1)\Delta(L_2) - \Delta(L_1\star L_2) =  \sum_{j,\mu}\Big((L_{1j}L_{2\mu}-
\sum_{k\in K(j,\mu)}L(k)_{j+\mu})\otimes L_1^{d_1-j}L_2^{d_2-\mu}\Big) + \\
\sum_{j,\mu}\Big(\sum_{k\in K(j,\mu)}L(k)_{j+\mu}\otimes(L_1^{d_1-j}L_2^{d_2-\mu} - L(k)^{d_1+d_2-j-\mu})\Big) \in \sK\otimes L_\sT + L_\sT \otimes \sK.
}
\end{proof}

\begin{rmk}\label{rmk2.2b} Any bamboo forest is equivalent mod $\sK$ to a sum of stalks. Indeed, one has e.g.
\eq{}{L_1\cdot L_2\cdot L_3\equiv (L_1\star L_2)\cdot L_3 \equiv (L_1\star L_2)\star L_3 \equiv L_1\star L_2\star L_3.
}
\end{rmk}

For any decoration $\ell$, one has an operator  \cite{BerK2}
\eq{}{B_+^\ell: H_\sT \to H_\sT
}
which carries any forest to the tree obtained by connecting a single root vertex with decoration $\ell$ to all the roots of the forest. This operator is a Hochschild 1-cocycle, i.e.
\eq{2.11b}{\Delta B_+^\ell = B_+^\ell \otimes \One + (\text{id}\otimes B_+^\ell)\Delta.
}
Let $\sJ \subset H_\sT$ be the smallest ideal containing the ideal $\sK$ as in lemma \ref{lem2.1b} and stable under all the operators $B_+^\ell$. Generators of $\sJ$ as an abelian group are obtained by starting with elements of $\sK$ and successively applying $B_+^\ell$ for various $\ell$ and multiplying by elements of $H_\sT$. It follows from \eqref{2.11b} that $\Delta\sJ \subset \sJ\otimes H_\sT + H_\sT\otimes \sJ$. Define
\eq{}{\overline H_\sT := H_\sT/\sJ.
}

 A {\it flag}
in a core graph $\Gamma$ is a chain
\eq{2.6b}{f:=\emptyset \subsetneq \Gamma_1\subsetneq\cdots\subsetneq \Gamma_n=\Gamma
}
of core subgraphs.  Write
$F(\Gamma)$ for the collection of all maximal flags of $\Gamma$. One checks easily that for a maximal flag, $n=|\Gamma|$.
Let us consider an example.
\bea \ef & \subsetneq & \abef\subsetneq \abcdef,\label{flag1}\\
\ef & \subsetneq & \cdef\subsetneq \abcdef,\label{flag2}\\
\abe & \subsetneq & \abef\subsetneq \abcdef,\\
\abf & \subsetneq & \abef\subsetneq \abcdef,\\
\abe & \subsetneq & \abcde\subsetneq \abcdef,\\
\abf & \subsetneq & \abcdf\subsetneq \abcdef,\\
\cde & \subsetneq & \cdef\subsetneq \abcdef,\\
\cdf & \subsetneq & \cdef\subsetneq \abcdef,\\
\cde & \subsetneq & \abcde\subsetneq \abcdef,\\
\cdf & \subsetneq & \abcdf\subsetneq \abcdef,\\
\abcd & \subsetneq & \abcde\subsetneq \abcdef,\\
\abcd & \subsetneq & \abcdf\subsetneq \abcdef,
\eea
are the twelve flags for the graph given in Fig.(\ref{overl}).
We omitted the edge labels in the above flags. Note that only the first two , (\ref{flag1},\ref{flag2})
are relevant for the renormalization Hopf algebra to be introduced below.
\begin{figure}[t]
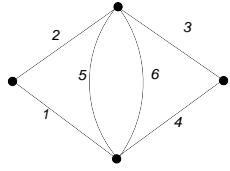
\overl
\caption{A graph with overlapping subdivergences. The renormalization Hopf algebra gives $\Delta^\prime{123456}=56\otimes 1234+1256\otimes 34+3456\otimes 12$. Note that each edge belongs to some subgraph with $\textrm{sdd}\geq 0$.}\label{overl}
\end{figure}

To the flag
$f$ we associate the ladder $L(f)$ with $n$ vertices decorated by
$\Gamma_i-\Gamma_{i-1}$. (More precisely, the foot is decorated by
$\Gamma_1$ and the root by $\Gamma-\Gamma_{n-1}$.). Define
\eq{}{\rho_L: H_\sC \to L_\sT;\quad  \rho_L(\Gamma) := \sum_{f\in F(\Gamma)} L(f)
}
Here the set of labels $D$ will be the set of subsets of graph
labels.
\begin{lem}The map $\rho_L$ is a homomorphism of Hopf algebras.
\end{lem}
\begin{proof}For a flag $f$ let $f^{(p)}$ be the bottom $p$ vertices
  with the given labeling, and let $f_{(p)}$ be the top $n-p$ vertices
  with the quotient labeling gotten by contracting the core graph
  associated to the bottom $p$ vertices. For $\gamma \subset \Gamma$ a
  core subgraph, define $F(\Gamma,\gamma) := \{f\in F(\Gamma)\ |\
  \gamma \in f\}$. There is a natural identification
\eq{}{F(\Gamma,\gamma) = F(\gamma) \times F(\Gamma/\!/\gamma).
}
We have
\eq{}{(\rho_L\otimes \rho_L)\circ\Delta_{\sC}(\Gamma) = \sum_\gamma
  \rho_L(\gamma)\otimes \rho_L(\Gamma/\!/\gamma) =
\sum_\gamma\sum_{f\in F(\Gamma,\gamma)}L(f^{|\gamma|})\otimes L(f_{|\gamma|}).
}
On the other hand
\eq{}{\Delta_L \circ \rho_L(\Gamma) = \sum_{f\in F(\Gamma)}
  \sum_{i=1}^n L(f^{(i)})\otimes L(f_{(i)}).
}
The assertion of the lemma is that there is a $1-1$ correspondence
\eq{}{\{\gamma, \text{max. flag of $\Gamma$ containing }\gamma\}
  \leftrightarrow \{\text{max. flag of $\Gamma$ }, i\le n\}.
}
This is clear.
\end{proof}

In fact, the tree structure associated to a maximal flag $f$ of $\Gamma$ is
rather more intricate than just a ladder. Though we do not use this tree structure in the sequel, we present the construction in some detail to help in understanding the difference between the core
and renormalization Hopf algebra.

We want to associate a
forest $T(f)$ to the flag $f$, and we proceed by induction on
$n=|\Gamma|$. We can write $\Gamma = \bigcup \Gamma^{(j)}$ in such a
way that all the $\Gamma^{(j)}$ are core and one vertex irreducible,
and such that $|\Gamma| = \sum |\Gamma^{(j)}|$. This decomposition is unique. If it is nontrivial, we define $T(f) = \prod T(f^{(j)})$ where $f^{(j)}$ is
the induced flag from $f$ on $\Gamma^{(j)}$. We now may assume
$\Gamma$ is one vertex irreducible. If the $\Gamma_i$ in our flag are
all one vertex irreducible, we take $T(f) = L(f)$ to be a ladder as
above. Otherwise, let $m<n$ be maximal such that $\Gamma_m \subsetneq
\Gamma$ is one vertex reducible. By induction, we have a forest
$T(f|\Gamma_m)$. To define $T(f)$, we glue the foot of the ladder with
decorations
$\Gamma_{m+1}-\Gamma_m,\dotsc,\Gamma-\Gamma_{n-1}$ to all the roots of
$T(f|\Gamma_m)$. (For an example, see figs.(\ref{coreco}) and (\ref{reconstr}).)
\begin{figure}[t]
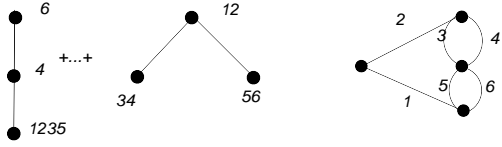
\coreco
\caption{The core Hopf algebra on rooted trees. We indicate subgraphs by edge labels on the vertices of rooted trees. The dots indicate seven more such trees, corresponding to flags $\Gamma_i\subsetneq\Gamma_j\subsetneq\Gamma_k$ with $\Gamma_i$ a cycle on four edges. The last tree represents a sum of two flags, $34\subsetneq 3456\subsetneq 123456+56\subsetneq 3456\subsetneq 123456$, again indicating graphs by edge labels. Hence that tree corresponds to a sum of two ladders, as it should.}\label{coreco}
\end{figure}
\begin{lem}\label{lem2.3b} Let $\Gamma = \bigcup \Gamma^{(j)}$ where $\Gamma$ and the $\Gamma^{(j)}$ are core. Assume $|\Gamma| = \sum_j |\Gamma^{(j)}|$. Then, viewing flags $f\in F(\Gamma)$ as sets of core subgraphs, ordered by inclusion, there is a $1-1$ correspondence between $F(\Gamma)$ and shuffles of the $F(\Gamma^{(j)})$.
\end{lem}
\begin{proof}One checks easily that the $\Gamma^{(j)}$ can have no edges in common. Further, there is a $1-1$ correspondence between core subgraphs $\Gamma' \subset \Gamma$ and collections of core subgraphs $\Gamma^{(j)}{}' \subset \Gamma^{(j)}$. Here, the dictionary is given by $\Gamma' \mapsto \{\Gamma'\cap \Gamma^{(j)}\}$ and $\{\Gamma^{(j)}{}'\} \mapsto \bigcup \Gamma^{(j)}{}'$. The assertion of the lemma follows.
\end{proof}
\begin{figure}[t]
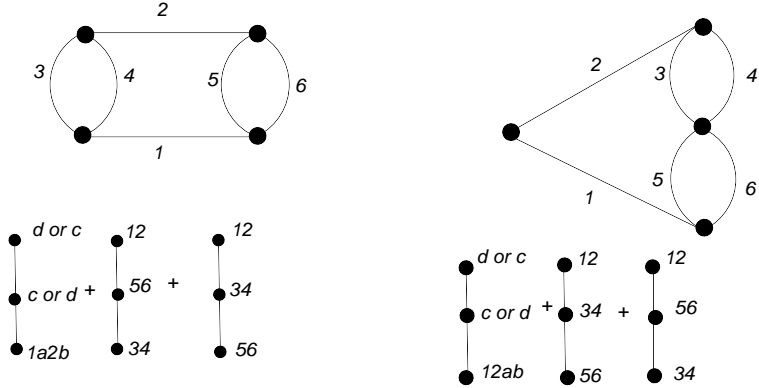
\reconstr\caption{The two graphs differ in how the subdivergences are inserted. $a,c\in{3,4}$, and $b,d\in{5,6}$, $c\not=a,b\not=d$. So there are eight such legal trees, plus the two which are identical between the two graphs. Note the permutation of labels at the feet of the trees in $\rho(\Gamma)$: $1a2b\leftrightarrow 12ab$. Keeping that order, we can uniquely reconstruct each graph from the knowledge of the labels at the feet: $1a2b,34,56$ and $12ab,34,56 $, which are the cycles in each graph. Note that in the difference of the two graphs, only the difference of those eight trees remains, corresponding to a primitive element in the renormalization Hopf algebra. The core Hopf algebra hence stores much more information than the renormalization Hopf algebra, which we hope to use in the future to understand the periods assigned to Feynman graphs by the Feynman rules.}\label{reconstr}
\end{figure}

As a consequence of lemma \ref{lem2.3b} we may partition the flags $F(\Gamma)$ associated to a core $\Gamma$ as follows. Given $f\in F(\Gamma)$, Let $\Gamma_m \subset \Gamma$ be maximal in the flag $f$ such that $\Gamma_m$ is $1$-vertex reducible. The flag $f$ induces a flag $f_m$ on $\Gamma_m$, and we know that it is a shuffle of flags $f_m^{(j)}$ on $\Gamma_m^{(j)}$ where $\Gamma_m = \bigcup \Gamma_m^{(j)}$ as in the lemma. We say two flags are equivalent, $f\sim f'$, if $f$ and $f'$ agree at $\Gamma_m$ and above, and if they simply correspond to two different shuffles of the flags $f_m^{(j)}$. We now have
\eq{2.20b}{T(f) \equiv \sum_{f'\sim f} L(f') \mod \sJ.
}
Indeed, $T(f)$ is obtained by successive $B_+^\ell$ operations applied to the forest $T(f|\Gamma_m)$. The latter, by remark \ref{rmk2.2b}, coincides with the righthand side of \eqref{2.20b}. We conclude
\begin{prop}\label{proptree} With notation as above, there exist homomorphisms of Hopf algebras
\eq{homomorph}{\begin{CD} H_\sC @> \rho_L >> L_\sT \\
@VVV @VVV \\
\overline H_\sC @>\rho_T >> \overline H_\sT
\end{CD}
}
\end{prop}
Here $\rho_T(\Gamma)$ is the sum $T(f)$ over equivalence classes of flags $f$ as above.
We will barely use $H_\sT$ in the following, and introduced it for completeness and the benefit of the reader used to it.

\subsection{Renormalization Hopf algebras}\label{renhopf}
In a similar manner, one may define homomorphisms
\eq{2.35b}{\rho_\sR: H_\sR \to \overline H_\sT
}
for any one of the renormalization Hopf algebras obtained by imposing
restrictions on external leg structure. For a graph $\Gamma$, let, as before, the
{\it residue} of $\Gamma$, $\textrm{res}(\Gamma)$, be the graph with
no loops obtained by shrinking all its internal edges to a point. What remains are the external half edges connected to that point (cf. section \ref{subsecgr}). Note that "doubling" an edge by putting a two-point vertex in it does not change the residue.

In $\phi_4^4$ theory for example,
graphs have $2m$ external legs, with $m\geq 0$. For a renormalizable theory, there is a finite set of external leg structures $\mathcal R$ such that we obtain
a renormalization Hopf algebra for that set.

For example, for massive $\phi_4^4$ theory, there are three such structures: the four-point vertex, and two two-point vertices, of kinetic type and mass type.

Let us now consider flags associated to core graphs. Such chains $\cdots \Gamma_i\subsetneq\Gamma_{i+1}\subsetneq\cdots\subsetneq\Gamma$ correspond to decorated ladders,
and the coproduct on the level of such ladders is a sum over all possibilities to cut an edge in such a ladder, splitting the chain
\be [\cdots\subsetneq\Gamma_i]\otimes [\Gamma_{i+1}/\!/\Gamma_i\subsetneq\cdots\subsetneq\Gamma/\!/\Gamma_i].\ee

So let us call such an admissible cut renormalization-admissible, if all core graphs $\Gamma_i$, $\Gamma/\!/\Gamma_i$ obtained by the cut have residues in ${\mathcal R}$.

The set of renormalization-admissible cuts is  a subset of the admissible cuts of a core graph, and the coproduct respects this. Hence the renormalization Hopf algebra $H_\sR$ is a quotient
Hopf algebra of the core Hopf algebra.

If we enlarge the set
${\mathcal R}$ to include other local field operators appearing for
example in an operator product expansion we get quotient Hopf algebras
between the core and the renormalization Hopf algebra.

\subsection{External leg structures}\label{extleg}
External edges are usually labeled by data which characterize the amplitude under consideration. Let $\sigma$ be such data. For graphs $\Gamma$ with a given residue ${\textrm{res}}(\Gamma)$, there is a finite set $\tau\in \{\sigma\}_{{\textrm{res}}(\Gamma)}$ of possible data $\tau$. A choice of such data determines a labeling of the corresponding vertex to which a subgraph shrinks. Let $\Gamma/\!/\gamma_\tau$ be that co-graph with the corresponding vertex labeling.

One gets a Hopf algebra structure on pairs
$(\Gamma,\sigma)$ by using the renormalization coproduct $\Delta(\Gamma)=\Gamma^\prime\otimes\Gamma^{\prime\prime}$ by setting $\Delta(\Gamma,\sigma)=\sum_{\tau\in \{\sigma\}_{\textrm{res}(\Gamma^\prime)}}(\Gamma^\prime,\tau)\otimes (\Gamma^{\prime\prime}_\tau,\sigma)$.
We regard the decomposition into external leg structures as a partition of unity and write \be \sum_{\tau\in \{\sigma\}_{\textrm{res}(\Gamma)}}(\Gamma,\tau)=(\Gamma,\One).\ee
In our applications we only need this for (sub)graphs $\gamma$ with $|\textrm{res}(\gamma)|=2$, and the use of these notions will become clear in the applications below.

\section{Combinatorics of blow-ups}\label{seccomb}
We consider $\P^{n-1}$ with fixed homogeneous coordinates
$\sA:=\{A_1,\dotsc,A_n\}$. Suppose given a subset $S\subset
2^{\sA}$. Assume $\sA \not\in S$ and that $S$ has the
property that whenever $\mu_1, \mu_2 \in S$ with $\mu_1\cup \mu_2 \neq
\sA$, then $\mu_1\cup \mu_2 \in S$. For $\mu \in S$ we
write $L_\mu \subset \P^{n-1}$ for the coordinate linear space defined
by $A_i=0,\ i\in \mu$. Write $L(S) := \{L_\mu\ |\ \mu \in S\}$. We see
that
\eq{1.1}{L_{\mu_i} \in L(S);\ L_{\mu_1}\cap  L_{\mu_2} \neq \emptyset
  \Rightarrow L_{\mu_1}\cap L_{\mu_2} \in L(S).
}
We can stratify the set $L(S)$ taking $L(S)_1$ to be the set of all
minimal elements (under inclusion) of $L(S)$. More generally, $L(S)_i$
will be the set of minimal elements in $L(S) -
\coprod_{j=1}^{i-1}L(S)_j$.
\begin{prop}\label{prop1.1}(i) Elements in $L(S)_1$ are all disjoint, so
  we may define $P(S)_1$ to be the variety defined by blowing up
  elements in $L(S)_1$ on $\P^{n-1}$. We do not need to specify an order
  in which to
  perform the blowups. \newline\noindent
(ii) More generally, the strict transforms of elements in $L(S)_{i+1}$ to
the space $P(S)_{i}$ obtained by successively blowing the strict
transform of $L(S)_j,\ j=1,\dotsc,i$ are disjoint, so we may
inductively define $P(S)$ to be the successive blowup of the
$L(S)_i$. \newline\noindent
(iii) Let $E_i \subset P(S)$ correspond to the blowup of $L_{\mu_i},\
i=1,\dotsc,r$.    ($E_i$ is the unique exceptional divisor with image $L_{\mu_i}$ in $P(S)$.)Then $E_1\cap \cdots \cap E_r \neq \emptyset$ if and
only if after possibly reordering, we have inclusions
$L_{\mu_1}\subset\cdots\subset L_{\mu_r}$. \newline\noindent
(iv) The total exceptional divisor $E \subset P(S)$ is a normal
crossings divisor. \newline\noindent
(v) Let $M \subset \P^{n-1}$ be a coordinate linear space. Assume $M
\not\subset L$ for any $L\in L(S)$. Then $M\cap L(S):= \{M\cap L\ |\ L \in
L(S)\}$ satisfies \eqref{1.1}. The strict transform of $M$ in $P(S)$
is obtained by blowing up elements of $M\cap L(S)$ on $M$ as in (i) and (ii)
above.
\end{prop}
\begin{proof}If $L_1\neq L_2 \in L(S)_i$ and $L_1\cap L_2 \neq
  \emptyset$, then $L_1\cap L_2 \in L(S)_j$ for some $j<i$. This means
  that when we get to the $i$-th step, $L_1\cap L_2$ has already been
  blown up, so the strict transforms of the $L_i$ are disjoint,
  proving (ii). For $(iii)$, $\bigcap E_i \neq \emptyset \Leftarrow
  L_{\mu_1}\subset \cdots \subset L_{\mu_r}$ follows from the above
  argument. Conversely, if we have strict inclusions among the
  $L_{\mu_i}$, we may write (abusively) $L_{\mu_i}/L_{\mu_{i-1}}$ for
  the projective space with homogeneous coordinates the    homogeneous
  coordinates on $L_{\mu_i}$ vanishing on $L_{\mu_{i-1}}$. The
  exceptional divisor on the blowup of $L_{\mu_{i-1}} \subset
  L_{\mu_i}$ is identified with $L_{\mu_{i-1}}\times
  (L_{\mu_i}/L_{\mu_{i-1}})$. A straightforward calculation identifies
  nonempty open sets (open toric orbits in the sense to be discussed
  below) in $\bigcap E_i$ and
\eq{1.2a}{L_{\mu_1}\times (L_{\mu_2}/L_{\mu_{1}})\times\cdots\times
(L_{\mu_r}/L_{\mu_{r-1}})
}
The remaining parts of the proposition follow from the algorithm in
  \cite{ESV}.
\end{proof}
For us, sets $S$ as above will arise in the context of graphs. Recall in \ref{subsecgr} we defined the notion of core graph.
\begin{prop}\label{prop1.4}Let $\Gamma$ be a graph, and let
  $\Gamma_1, \Gamma_2 \subset \Gamma$ be core subgraphs. Then the union
  $\Gamma_1\cup \Gamma_2$ is a core subgraph.
\end{prop}
\begin{proof}Removing an edge increases the Euler-Poincar\'e
  characteristic by $1$. If $h_1$ doesn't drop, then either $h_0$
  increases (the graph disconnects when $e$ is removed) or $e$ has a
  unary vertex so removing $e$ drops the number of vertices. Suppose $e$ is
  an edge of $\Gamma_1$ (assumed core). Then $e$ cannot have a unary
  vertex. If, on the other hand, removing $e$ disconnects $\Gamma_1\cup
  \Gamma_2$, then since the $\Gamma_i$ are core what must happen is
  that each $\Gamma_i$ has precisely one vertex of $e$. But this would
  imply that $\Gamma_1$ is not core, a contradiction.
\end{proof}
To a graph $\Gamma$ we may associate the projective space $\P(\Gamma)$
with homogeneous coordinates $A_e,\ e\in E(\Gamma)$ labeled by the
edges of $\Gamma$.
Let $\Gamma$ be a core graph. A {\it coordinate linear space}
$L\subset \P(\Gamma)$ is a non-empty linear space defined by some
subset of the
homogeneous coordinate functions, $L:A_{e_1}=\cdots =
A_{e_p}=0$. Define $L(\Gamma)$ to be the set of coordinate linear
spaces in $\P(\Gamma)$ such that the corresponding set of edges
$e_{i_1},\dotsc,e_{i_p}$ is the edge set of a core subgraph $\Gamma'
\subset \Gamma$. It follows from proposition \ref{prop1.4} that
$L(\Gamma)$ satisfies condition \eqref{1.1}, so the iterated blowup
\eq{3.3b}{\pi: P(\Gamma) \to \P(\Gamma)
}
as in proposition \ref{prop1.1} is defined. Define
\eq{1.4}{\sL = \bigcup_{L\in L(\Gamma)}L \subset \P(\Gamma);\quad E =
  \bigcup E_L = \pi^{-1}\sL.
}
\begin{lem}\label{lem1.5} Suppose $\P(\Gamma) = \P^{n-1}$ with
  coordinates $A_1,\dotsc,A_n$. Let $L \subset \P(\Gamma)$ be defined
  by $A_1=\cdots=A_p=0$. Let $\pi_L:P_L \to \P(\Gamma)$ be the blowup
  of $L$. Then the exceptional divisor $E\subset P_L$ is identified
  with $\P^{p-1}\times L$. Further $A_1,\dotsc,A_p$ induce coordinates
  on the vertical fibres $\P^{p-1}$ and $A_{p+1},\dotsc,A_n$ give
  homogeneous coordinates on $L$.
\end{lem}
\begin{proof} This is standard. One way to see it is to use the map $\P^{n-1}-L \to \P^{p-1},\ [a_1,\dotsc,a_n] \mapsto [a_1,\dotsc,a_p]$. (Here, and in the sequel, $[\cdots]$ denotes a point in homogeneous coordinates.) This extends to a map $f$ on $P_L$:
\eq{}{\begin{CD} E @> \inj >> P_L @> f >> \P^{p-1} \\
@VV\pi_L|_E V  @VV\pi_L V \\
L @> \inj >> \P^{n-1}.
\end{CD}
}
The resulting map $\pi_L|_E \times f: E \cong L\times \P^{p-1}$.
\end{proof}
It will be helpful to better understand the geometry of
$P(\Gamma)$. Let $\G_m = \Spec \Q[t,t^{-1}]$ be the standard one
dimensional algebraic torus. Define $T=\G_m^n/\G_m$ where the quotient
is taken with respect to the diagonal embedding. For all practical
purposes, it suffices to consider complex points
\eq{}{T(\C) = \C^{\times n}/\C^\times \cong \C^{\times n-1}.
}
A {\it toric variety} $P$ is an equivariant (partial) compactification of
$T$. In other words, $T\subset P$ is an open set, and we have an
extension of the natural group map $m$
\eq{}{\begin{CD}T\times T @>\subset >> T\times P \\
@V m VV @V\bar m VV \\
T @>\subset >> P.
\end{CD}
}
For example, $\P(\Gamma)$ is a toric variety for a torus $T(\Gamma)$. Canonically, we may write $T(\Gamma) = (\prod_{e\in \text{Edge}(\Gamma)}\G_m)\big/ \G_m.$
More important
for us:
\begin{prop}\label{prop1.6} (i) $P(\Gamma)$ is a toric variety for
  $T=T(\Gamma)$. \newline\noindent
(ii) The orbits of $T$ on $P(\Gamma)$ are in $1-1$ correspondence with
pairs $(F,\ \Gamma_p \subsetneq \cdots \subsetneq \Gamma_1\subsetneq
\Gamma)$. Here $F \subset \Gamma$ is a (possibly empty) subforest
(subgraph with
$h_1(F)=0$) and the $\Gamma_i$ are core subgraphs of $\Gamma$. We
require that the image of $F_i:=F\cap \Gamma_i$ in
$\Gamma_i/\!/\Gamma_{i+1}$ be a subforest for each $i$.
(cf. \eqref{1.2a}).   The orbit associated to such a pair is
canonically identified with the open orbit in the toric variety
$\P(\Gamma_p/\!/F_p)\times \P((\Gamma_{p-1}/\!/\Gamma_p)/\!/F_{p-1})\times \cdots
\times \P((\Gamma/\!/\Gamma_1)/\!/F)$.
\end{prop}
\begin{proof}A general reference for toric varieties is \cite{F}. The fact (i) that $P(\Gamma)$ is a toric variety follows inductively from the fact that the blowup of an invariant ideal $I$ in a toric variety is toric. Indeed, the torus acts on $I$ and hence on the blowup $\text{Proj}(I)$.

We recall some toric constructions. Let $N= \Z^{\text{Edge}(\Gamma)}/\Z$, and let $M=\hom(N,\Z)$. We have canonically $T = \Spec \Q[M]$ where $\Q[M]$ is the group ring of the lattice $M$. A {\it fan} (op. cit., 1.4, p. 20) $\sF$ is a finite set of convex cones in $N_\R = N\otimes \R$ satisfying certain simple axioms. To a cone $C\subset N_\R$ one associates the dual cone (op. cit. p. 4)
\eq{}{C^\vee = \{m \in M_\R\ |\ \langle m,c\rangle \ge 0, \forall c \in C\}
}
(resp. the semigroup $C^\vee_\Z = C^\vee \cap M$).
The toric variety $V(\sF)$ associated to the fan $\sF$ is then a union of the affine sets $U(C):= \Spec \Q[C^\vee_\Z]$.
For example, our $N$ has rank $n-1$. There are $n$ evident elements
$e$ determined by the $n$ edges of $\Gamma$. Let $C_e = \{\sum_{e'\neq
  e} r_{e'}e'\ |\ r_{e'} \ge 0\}$ be the cone spanned by all edges
except $e$. The spanning edges for $C_e$ form a basis for $N$ which
implies that $U(C_e) \cong \A^{n-1}$. Since all the coordinate rings
lie in $\Q[M]$ (i.e. $T(\Gamma) \subset U(C_e)$), one is able to glue
together the $U(C_e)$. The resulting toric variety associated to the
fan $\{C_e\ |\ e\in \text{Edge}(E)\}$ is canonically identified with
$\P(\Gamma)$.
\begin{rmk}\label{rmk1.7} Our toric varieties will all be smooth
  (closures of orbits in smooth
toric varieties are smooth), which is equivalent (\cite{F},
\S 2) to the condition that cones in the fan are all generated by
subsets of bases for the lattice $N$. Faces of these cones are in
$1-1$ correspondence with subsets of the generating set.
\end{rmk}
In general, the orbits of the torus action are in $1-1$ correspondence with the cones $C$ in the fan (op. cit. 3.1, p.51). The subgroup of $N$ generated by $C\cap N$ corresponds to the subgroup of $T$ which acts trivially on the orbit. For example, in the case of projective space $\P^{n-1}$, there are $n$ cones $C_e$ of dimension $n-1$ corresponding to the $n$ fixed points $(0,\dotsc,1,\dotsc,0) \in \P^{n-1}$. For any $S\subsetneq \text{Edge}(\Gamma)$, the cone $C(S)$ spanned by the edges of $S$ corresponds to the orbit $
\{(\dotsc,x_e,\ldots)\ |\ x_e=0 \Leftrightarrow e\in S\}\subset \P^{n-1}$.
Let $L:A_e=0, e\in \Gamma'\subset \Gamma$ be a coordinate linear space in $\P(\Gamma)$ associated to a subgraph $\Gamma'\subset \Gamma$. It follows from lemma \ref{lem1.5} that the exceptional divisor $E_L \subset P_L$ in the blowup of $L$ can be identified with
\eq{1.11}{E_L= \P(\Gamma')\times \P(\Gamma/\!/\Gamma').
}
Let $e(\Gamma') = \sum_{e\in \Gamma'} e \subset N_\R$, and write
$\tau(\Gamma') = \R^{\ge 0}\cdot e(\Gamma')$.  The subgroup $\Z\cdot e(\Gamma') \subset N$ determines a $1$-parameter subgroup $G(\Gamma') \subset T = \Spec \Q[M]$. It follows from \eqref{1.11} that $G(\Gamma')$ acts trivially on $E_L$. One has $\tau(\Gamma') \subset C'\subset C_e$ for all $e\not\in \Gamma'$, where $C'$ is the cone generated by the edges of $\Gamma'$.  For all $e' \in \Gamma'$ we define a subcone $C_{e,e'} \subset C_e$ to be spanned by $\tau(\Gamma')$ together with all edges of $\Gamma$ except $e, e'$. The fan for $P_L$ is then
\eq{}{\{C_e,\ e\in \Gamma'\}\cup \{ \ C_{e,e'},\ e\not\in \Gamma', e' \in \Gamma'\}.
}
Note that $C_e,\ e\not\in \Gamma'$ is not a cone in the fan for
$P_L$. More generally, let $\sF$ be the fan for
$P(\Gamma)$. Certainly, $\sF$ will contain as cones the half-lines
$\tau(\Gamma')$ for all core subgraphs $\Gamma'\subset \Gamma$ as well
as the $\R^{\ge 0}e,\ e\in \Gamma$.  but we must make precise which
subsets of this set of half-lines span higher dimensional cones in
$\sF$. By general theory, the cones correspond to the nonempty
orbits. In other words,
\eq{1.13}{\R^{\ge 0}e_1,\dotsc,\R^{\ge 0}e_p,\R^{\ge
  0}e(\Gamma_1),\dotsc,\R^{\ge 0}e(\Gamma_q)
}
span a cone in $\sF$ if
and only if the intersection
\eq{1.14}{E_1\cap\cdots \cap E_q\cap
D_1\cap\cdots\cap D_p \neq \emptyset,
}
where $E_i \subset P(\Gamma)$
is the exceptional divisor corresponding to $L(\Gamma_i)$ and $D_j
\subset P(\Gamma)$ is the strict transform of the coordinate divisor
$A_{e_i}=0$ in $\P(\Gamma)$. To understand \eqref{1.14}, consider the
simple case $E_1\cap D_1$. We have a core subgraph $\Gamma_1\subset
\Gamma$, and an edge $e_1$ of $\Gamma$. We know by lemma \ref{lem1.5}
that $E_1 \cong \P(\Gamma_1)\times \P(\Gamma/\!/\Gamma_1)$. If $e_1$ is
an edge of $\Gamma_1$, then $D_1\cap E_1=\P(\Gamma_1/\!/e_1)\times
\P(\Gamma/\!/\Gamma_1)$. Otherwise
$$D_1\cap E_1=\P(\Gamma_1)\times
\P((\Gamma/\!/\Gamma_1)/\!/e_1).
$$
One (degenerate) possibility is that $e_1$ is an edge of $\Gamma_1$
which forms a loop (tadpole). In this case, $e_1$ is itself a core
subgraph of $\Gamma$, and the divisor $D_1$ should be treated as one
of the exceptional divisors $E_i$. Thus, we omit this
possibility. Another possibility is that $e_1 \not\in \Gamma_1$, but
that the image of $e_1$ in $\Gamma/\!/\Gamma_1$ forms a loop. In this
case, $\Gamma_2:=\Gamma_1\cup e_1$ is a core subgraph, so the linear
space $L_2:A_e=0,\ e\in \Gamma_2$ gets blown up in the process of
constructing $P(\Gamma)$. But blowing $L_2$ separates $E_1$ and $D_1$,
so the intersection of the strict transforms of $D_1$ and $E_1$ in
$P(\Gamma)$ is empty. The general argument to show that \eqref{1.14}
is empty if and only if the conditions of (ii) in the proposition are
fulfilled is similar and is left for the reader. Note that the case
where there are no divisors $D_i$ follows from proposition
\ref{prop1.1}(iii).
\end{proof}
We are particularly interested in orbits corresponding to filtrations
by core subgraphs $\Gamma_p\subsetneq \cdots \subsetneq
\Gamma_1\subsetneq \Gamma$. Let $V \subset P(\Gamma)$ be the closure
of this orbit. We want to exhibit a toric neighborhood of $V$ which
retracts onto $V$ as a vector bundle of rank $p$. As in the proof of
proposition \ref{prop1.6} we have $e(\Gamma_i):=\sum_{e\in \Gamma_i}
e$. The cone $C$ spanned by the $e(\Gamma_i)$ lies in the fan
$\sF$. For cones $C' \in \sF$ we write $C'>C$ if $C$ is a subcone of
$C'$. By the general theory, this will happen if and only if $C\subset
C'$ is a subcone which appears
on the boundary of $C'$. The orbit corresponding to $C'$ will then
appear in the closure of the orbit for $C$.
\begin{prop}\label{prop1.8} With notation as above, Let $\sF_C \subset
  \sF$ be the
  subset of cones $C'$ such that we have $C'\le C''\ge C$ for some
  $C'' \in \sF$. Write $P^0 \subset P(\Gamma)$ for the open toric subvariety
  corresponding to the subfan $\sF_C\subset \sF$. We have $V\inj P^0
  \subset P(\Gamma)$. Further there is a retraction $\pi: P^0 \to V$
  realizing $P^0$ as a rank $p$ vector bundle over $V$ which is
  equivariant for the action of the torus $T$.
\end{prop}
\begin{proof}One has the following functoriality for toric
  varieties \cite{F}, \S 1.4. Suppose $\phi: N' \to N''$ is a
  homomorphism of lattices (finitely generated free abelian
  groups). Let $\sF', \sF''$ be fans in $N_\R', N_\R''$. Suppose for each cone
  $\sigma' \in \sF'$ there exists a cone $\sigma'' \in \sF''$ such that
  $\phi(\sigma') \subset \sigma''$. Then there is an induced map on
  toric varieties $V(\sF') \to V(\sF'')$.
Let
  $N'=N=\Z^n/\Z$ as above, and $N'' =
  N'/(\Z e(\Gamma_1)+\cdots+\Z e(\Gamma_p))$. One has the evident
  surjection  $\phi: N' \surj N''$. We take as fan $\sF'=\sF_C \subset
  \sF$. The
  closure $V$ of the orbit corresponds to the fan $\sF''$ in $N_\R''$
  given by the images of all cones $C''\ge C$ (op. cit. \S 3.1). Such
  a $C''$ is generated by $e(\Gamma_1),\dotsc,e(\Gamma_p),
  f_1,\dotsc,f_q$, and there are no linear relations among these
  elements (remark \ref{rmk1.7}). A subcone $C'\le C''$ is generated
  by a subset
  $e(\Gamma_{i_1}),\dotsc,e(\Gamma_{i_a}),f_{1},\dotsc,f_{b}$. The
  image is simply the cone in $N_\R''$ generated by the images of the
  $f$'s. If we have another cone $C_1'\le C_1''\ge C$ in $\sF'$ with
  the same image in
  $\sF''$, it will have generators say $g_{1},\dotsc,g_b$ together
  with some of the $e(\Gamma_i)$'s. Reordering the $g$'s, we find that
  there are relations
\eq{}{f_i + \sum a_{ij}e(\Gamma_j) = g_i + \sum b_{ij}e(\Gamma_j)
}
with $a_{ij}, b_{ij} \ge 0$. It follows that the cones in $\sF$ spanned by
$f_i, e(\Gamma_1),\dotsc,e(\Gamma_p)$ and
$g_i,e(\Gamma_1),\dotsc,e(\Gamma_p)$ meet in a subset strictly larger
that the cone spanned by the $e(\Gamma_j)$. By the fan axioms, the
intersection of two cones in a fan is a common face of both, so these
two cones coincide, which implies $f_i=g_i$. In particular, for each
cone in $\sF''$, there is a unique minimal cone in $\sF'$ lying over
it. This is the hypothesis for \cite{O}, p. 58, proposition 1.33. One
concludes that the map $\pi: P^0 \to V$ induced by the map $\sF' \to \sF''$
is an equivariant fibration, with fibre the toric bundle associated to
the fan generated by the $e(\Gamma_i),\ 1\le i\le p$. This toric
variety is just affine $p$-space, so we get an equivariant
$\A^p$-fibration over $V$. Any such fibration is necessarily a vector
bundle with structure group $\G_m^p$. Indeed, this amounts to saying
that any automorphism of the polynomial ring $k[x_1,\dotsc,x_p]$ which
intertwines the diagonal action of $\G_m^p$ is necessarily of the form
$x_i \mapsto c_ix_i$ with $c_i \in k^\times$.
\end{proof}
\begin{rmk}\label{rmk1.9} We will need to understand how these
  constructions are
  compatible. Let $V$ be a closed orbit corresponding to a cone $C$ as
  above, and let $V_1 \subset
  V$ be a smaller closed orbit corresponding to a larger cone
  $C_1>C$. (The correspondence between cones and orbits is
  inclusion-reversing.) As above we have a toric variety $V_1\subset P_1^0
  \subset P(\Gamma)$ and a retraction $\pi_1:P_1^0 \to V_1$. The fan $\sF_1'$
  for $P_1^0$ is given by the set of cones $C_1'$ in $\sF$ such that
\eq{}{C_1' \le C'' \ge C_1\ (>C).
}
It follows that $\sF_1' \subset \sF' = \sF_C$, so $P_1^0 \subset P^0$
is an open subvariety. Let $V^0 \subset V$ be the image of the
composition $P_1^0 \subset P^0 \xrightarrow{\pi} V$. Then $V^0$ is the
open toric subvariety of $V$ corresponding as above to the closed
orbit $V_1 \subset V$, and we have a retraction $V^0
\xrightarrow{\pi_V} V_1$. One gets commutative diagrams
\eq{}{\begin{CD}P_1^0 @= P_1^0 @>\subset >> P^0 \\
@V\pi_1 VV @V \pi VV @V \pi VV \\
V_1 @<\pi_V << V^0 @>\subset >> V
\end{CD}
}
and
\eq{1.18}{\begin{CD}P^0|V_1 @>\subset >> P_1^0 \\
@V\pi VV @V\pi_1 VV \\
V_1 @= V_1.
\end{CD}
}
\end{rmk}
\begin{rmk}\label{rmk1.10} Using the toric structure, one can realize
  these vector bundles as direct sums of line bundles corresponding to
  characters of the tori acting on the fibres.The inclusion on the top
  line of \eqref{1.18} corresponds to characters which act trivially
  on all of $V$.
\end{rmk}
\begin{rmk}\label{rmk1.11}(compare proposition \ref{prop1.6}). Given a flag of core subgraphs
\eq{1.19}{\Gamma_p \subsetneq \Gamma_{p-1}\subsetneq \cdots\subsetneq
  \Gamma_1\subsetneq \Gamma,
}
let $L_i\subset \P(\Gamma)$ be defined by the edge variables for edges
in $\Gamma_i$, so we have $L_1\subsetneq\cdots\subsetneq L_p\subsetneq
\P(\Gamma)$. For $L \subset \P(\Gamma)$ a coordinate linear space, let
$T(L)\subset L$ be the subtorus where none of the coordinates
vanish. Then the orbit associated to \eqref{1.19} is
\eq{}{ T(L_1)\times T(L_2/L_1)\times \cdots \times T(L_p/L_{p-1})\times
  T(\P^{n-1}/L_p)
}
(Here the notation $L_{i+1}/L_i$ is as in \eqref{1.2a}.)
\end{rmk}
 \section{Topological Chains on Toric Varieties}\label{sectopch}
 One can define the notion of non-negative real points $V(\R^{\ge 0})$ and positive real points $V(\R^{>0})$. For a torus $T = \Spec \Q[N^\vee]$ for some $N \cong \Z^g$ we take
 $$T(\R^{> 0}) = \{\phi: \Q[N^\vee] \to \R \ |\
 \phi(n) > 0, \forall n\in N^\vee \}.$$
 A toric variety $V$ can be stratified as a disjoint union of tori $V =
 \coprod T_\alpha$. Define
 \ga{2.1a}{V(\R^{\ge 0}) = \coprod T_\alpha(\R^{>0}); \\
 V(\R^{>0}) = T(\R^{>0}), \notag
 }
 where $T \subset V$ is the open orbit.
 Let $V\subset P(\Gamma)$ be the closure of the orbit associated to a flag
 \eqref{1.19}, and let $T(V)\subset T=\Spec \Q[N^\vee]$ be the subtorus
 acting trivially on $V$. Let $\pi_V: P_V \to V$ be the vector bundle as in
 proposition \ref{prop1.8}. We write $P_V = \sL_1\oplus \cdots \oplus
 \sL_p$ as a direct sum of line bundles, where each $\sL_i$ is
 equivariant for $T(V)$. Let $K(V) \cong (S^1)^p \subset T(V)(\C)$ be
 the maximal compact subgroup. Note that one has a canonical identification $T(V) = \G_m^p$ associated to the $1$-parameter subgroups of $T(V)$ generated by $e(\Gamma_i) \in N$. In particular, the identification $K(V) = (S^1)^p$ is canonical as well.
 For all closed orbits $V$ we may fix metrics on the $\sL_i$ which    are compatible under inclusions \eqref{1.18} and are
 (necessarily) invariant under the action of $K(V)$. We fix also a constant
 $\eta>0$. We can then define $S_V^\eta \subset P_V$ to be the product
 of the circle bundles of radius $\eta$ embedded in the
 $\sL_i$. $S_V^\eta$ becomes a principal bundle over $V$ with structure
 group $K(V)$.   Note that $S^\eta_V\cap P_V(\R^{\ge 0})$ contains a unique point in every fibre of $S_V^\eta$ over a point of $V(\R)$.
 Let $0<\ve<<\eta$ be another constant. We need to define a chain
 $\sigma_V^{\eta,\ve} \subset V(\R^{> 0})$. We consider closures $V_1
 \subset V$ of codimension $1$ orbits in $V$. For each such $V_1$ we
 have an open $P(V)_1 \subset V$ and a retraction $P(V)_1 \to V_1$
 which is a line bundle with a metric. The fibres of $P(V)_1(\R^{>0})$
 have a canonical coordinate $r>0$. If $V_1$ corresponds to an
 intersection of $V = E_1\cap\cdots\cap E_p$ with another exceptional
 divisor $E_{p+1}$, then we remove from each fibre of $P(V)_1(\R^{>0})$
 over $V_1(\R^{>0})$ the locus where $r<\eta$. If, on the other hand
 $V_1$ corresponds to an intersection of V with one of the $D_i$
 (i.e. with a strict transform of one of the coordinate divisors), then
 we remove the locus $r<\ve$. Repeating this process for each $V_1$
 (i.e. for each irreducible toric divisor in $V$), we obtain a compact
 $\sigma_V^{\eta,\ve}\subset V(\R^{>0})$ which stays away from the
 boundary components.    (Here "boundary components" are exceptional divisors together with strict transforms of coordinate divisors.)
 \begin{ex}\label{ex2.1a} Consider the case $V=P(\Gamma)$. Let $\pi: P(\Gamma) \to \P(\Gamma)$, and let
 $$\sigma = \{(A_1,\dotsc,A_n)\ |\ A_i \ge 0\} \subset \P(\Gamma)(\R)$$
 be the original integration chain. We have $\sigma_{P(\Gamma)}^{\eta,\ve} \subset \pi^{-1}(\sigma)$ defined by excising away points within a distance of $\eta$ from an $E_i$ or $\ve$ from the strict transform $D_j$ of a coordinate divisor $A_j=0$. (cf.\ fig.(\ref{figc})).  It is a manifold with corners.
 \begin{figure}[t]
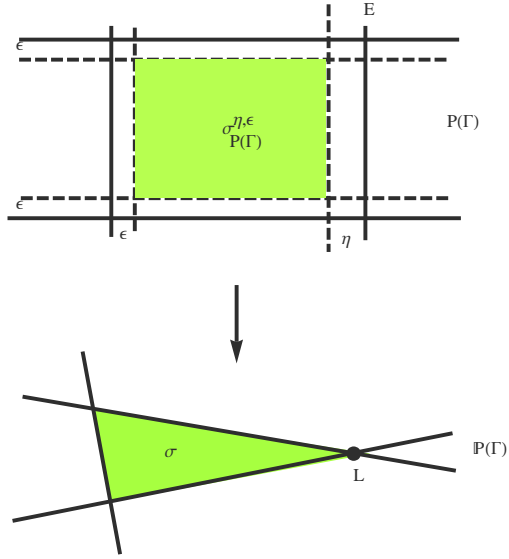
\centering
 \figc\caption{$P(\Gamma)$ and the real chain $\sigma_{P(\Gamma)}^{\eta,\ve}$.}\label{figc}
 \end{figure}
 \end{ex}
  Define $\tau_V^{\eta,\ve}$ to be the inverse
 image of $\sigma_V^{\eta,\ve}$ in $S_V^\eta$. The fibres of
 $\tau_V^{\eta,\ve}$ over $\sigma_V^{\eta,\ve}$ are products $(S^1)^p$
 with a canonical origin at the point where this fibre meets $
 P_V(\R^{\ge 0})$. For an angle $0\le \theta\le 2\pi$, we can thus
 define $\tau_V^{\eta,\ve,\theta} \subset \tau_V^{\eta,\ve}$ to be
 swept out by the origin in each fibre under the action of
 $[0,\theta]^p \subset K(V)$.  The chains $\tau_V^{\eta,\ve,\theta}$ have $\R$-dimension $n-1$ which is equal to the complex dimension of $\P(\Gamma)$ and $P(\Gamma)$.
 \begin{ex}\label{ex4.2} Here is an example which is too simple to correspond to any
   graph, but is sufficient to clarify the toric picture. Take
 \eq{}{L_1:A_1=A_2=0;\quad L_2:A_2=0
 }
 in $\P^2$ with coordinates $A_1, A_2, A_3$. Take $P\xrightarrow{\pi}
 \P^2$ to be the blowup of $L_1=(0,0,1)$. Let $E_1 \subset P$ be the
 exceptional divisor, and let $E_2\subset P$ be the strict transform of
 $L_2$. Note that $E_2$ is already a divisor so it is not necessary to
 blow up again. Take $V=E_1\cong \P^1$. The fan $\sF$ for $P$ is fig.(\ref{figd}).
   \begin{figure}[t]
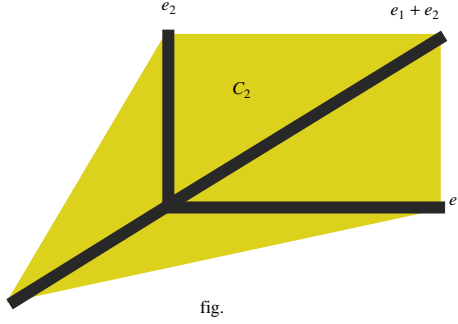
\centering
 \figd\caption{Fan for Example \ref{ex4.2}.}\label{figd}
 \end{figure}
  The cone $C=\R^{\ge 0}\cdot(e_1+e_2)$, so the fan $\sF' = \sF_C\subset
 \sF$ is the subset of cones lying in the first quadrant. The toric
 variety $P_V$ is $\A^2$ with $(0,0)$ blown up. It projects down onto
 $V$ as a line bundle. $S_V^\eta \subset P_V(\C)$ is then a circle
 bundle over $V(\C)$.
 $V$ has two suborbits $V_2 = E_1\cap E_2$ and $V_1 = E_1\cap D_1$,
 where $D_1$ is the strict transform of the divisor $A_1=0$ in
 $\P^2$. We may interpret $z:=A_1/A_2$ as a coordinate on $V$, so
 $V_1:z=0$ and $V_2:z=\infty$. We have $P(V)_1=V-\{z=\infty\}$ and
 $P(V)_2=V-\{z=0\}$. The real chain $\sigma_V^{\eta,\ve} = \{\eta\le
 z\le 1/\ve\}$, and $\tau_V^{\eta,\ve}$ is the $S^1$-bundle of radius
 $\eta$ over $\sigma_V^{\eta,\ve}$.
 On the other hand, $V_2$ corresponds to the cone labeled $C_2$ in
 fig.(\ref{figd}), and the fan $\sF_{C_2}$ is just $C_2$ itself. The toric
 variety $P_{V_2} \cong \A^2$ is a rank $2$ vector bundle over the
 point $V_2$. We have $P_{V_2} \subset P_V$. In
 this case $\sigma_{V_2}^{\eta,\ve}$ is simply the point $V_2$, and
 $\tau_{V_2}^{\eta,\ve}\cong S^1\times S^1\subset P_{V_1}(\C)$. In
 local coordinates around $V_1$ given by eigenfunctions for the torus
 action we have
 \ga{2.2a}{\tau_V^{\eta,\ve, \theta} = \{(\eta e^{i\mu},z)\ |\ \eta\le z\le
   1/\ve,\ 0 \le \mu \le \theta\}  \\
 \tau_{V_1}^{\eta,\ve, \theta} = \{(\eta e^{i\mu},\eta e^{i\nu})\ |\ 0
 \le \mu, \nu \le \theta\}  \notag \\
 \tau_V^{\eta,\ve, \theta}\cap \tau_{V_1}^{\eta,\ve, \theta} = \{(\eta
 e^{i\mu}, \eta)\ |\ 0 \le \mu \le \theta\}. \notag
 }
 \end{ex}
 We want now to establish a basic formula for the boundary of the
 chains $\tau_V^{\eta,\ve,\theta}$. Here $V$ runs through the closures
 of orbits in $P(\Gamma)$ associated to flags of core subgraphs
 \eqref{1.19}. We include the big orbit $V=P(\Gamma)$. We write $|V| :=
 \text{codim}(V/P(\Gamma))$. We may express the boundary chains
 $\partial \tau_V^{\eta,\ve,\theta}$ locally (in fact Zariski-locally) in
 coordinates which are eigenfunctions for the torus action. It is clear
 (cf. \eqref{2.2a}) that boundary terms are obtained by setting a
 suitable one of
 these coordinates to be constant: either $\eta e^{i\theta}$ or $\eta$
 or $\ve$. (The presence of $1/\ve$ in the first line of \eqref{2.2a}
 simply means that the appropriate coordinate near that point is
 $1/z$.)
 \begin{prop}For a suitable orientation, the boundary
 \eq{2.3a}{\partial \sum_V (-1)^{|V|}\tau_V^{\eta,\ve,\theta}
 }
 will contain no chains with one coordinate constant $=\eta$.
 \end{prop}
 \begin{proof}(Cf.\ fig.(\ref{fige}) ). For a given boundary term, we can choose local
   eigenfunction coordinates $x_1,\dotsc,x_{n-1}$ such that be boundary
   term is given by $x_1=\eta$. We take the chains to be oriented in
   some consistant way by this ordering of coordinates.    (Note that these coordinates are defined on a Zariski open set. The obstruction to choosing consistent orientations for various open sets is a class in the first Zariski cohomology of $P(\Gamma)$ with constant $\Z/2\Z$-coefficients. Since this cohomology group vanishes, we can choose such consistent orientations.) If $\partial
   \tau_V^{\eta,\ve,\theta}$ contains a term with $x_1=\eta$, there are
   two possibilities. Either $x_1$ is a real coordinate on
   $\tau_V^{\eta,\ve,\theta}$ or it is a circular coordinate. If $x_1$ is a real coordinate, then the fact that $x_1=\eta$ appears in the boundary means that locally $x_1=0$ defines a codimension $1$ orbit closure $V_1 \inj V$. In $\partial \tau_{V_1}^{\eta,\ve,\theta}$, $x_1$ will appear as a circular coordinate. Since $|V|=|V_1|+1$, the same chain $x_1=\eta$ will appear in $\partial\tau_V^{\eta,\ve,\theta}$ and in $\partial\tau_{V_1}^{\eta,\ve,\theta}$ and will cancel in \eqref{2.3a}.  If, on the other hand, $x_1 = \eta e^{i\theta}$ is a circular coordinate, then for suitable ordering of coordinates, the chain will be an $(S^1)^p$-bundle over a chain $\sigma$ contained in the locus where certain coordinates $\ge 0$. But then \eqref{2.3a} will contain another chain which is an $(S^1)^{p-1}$-bundle over $\{x_1\ge \eta\}\times \sigma$, and the boundary components involving $x_1=\eta$ will occur with opposite signs and will cancel.
 \end{proof}
 The boundary chain \eqref{2.3a}  is an $(n-2)$-chain involving two scales $0<\ve<\eta$. We want to construct an $(n-1)$-chain $\xi^{\eta,\ve,\theta}$ which amounts to a scaling $\eta \to \ve$. To do this, we construct a vector field $v$ on $P(\Gamma)$. Let $E = \sum E_i$ be the exceptional divisor. $v$ will be $0$ outside a neighborhood $N$ of $E$. Locally, at a point on $N$ which is close to divisors $E_1,\dotsc,E_p$ we have coordinates $x_1,\dotsc,x_p$ which are eigenfunctions for the  torus action such that locally $E_i:x_i=0$. Locally we will take $v$ to be radial and inward-pointing in each $x_i$. We glue these local $v$'s using a partition of unity. "Flowing" the $(n-2)$-chain \eqref{2.3a} along this vector field yields  an $(n-1)$-chain $\xi^{\eta,\ve,\theta}$. If this is done with care, we can arrange
 \eq{}{\partial \xi^{\eta,\ve,\theta} \equiv \partial \sum_V (-1)^{|V|}\tau_V^{\eta,\ve,\theta} -
 \partial \sum_V (-1)^{|V|}\tau_V^{\ve,\ve,\theta}.
 }
 Here $\equiv$ means that the two sides differ by a chain lying in an
 $\ve$-neighborhood of the strict transform $D$ of the coordinate
 divisor $\Delta$ in $P(\Gamma)$.
 Another important property of the chain $\xi^{\eta,\ve,\theta}$ is
 \begin{lem}$\xi^{\eta,\ve,2\pi} \equiv \xi^{\eta,\ve,0}$.
 \end{lem}
 \begin{proof}The point is that $\partial \tau_V^{\eta,\ve,2\pi} \equiv 0$
   except for the case $V=P(\Gamma)$, and
   $\tau_{P(\Gamma)}^{\eta,\ve,\theta}$ is independent of
   $\theta$. (See fig.(\ref{fige})).
  \begin{figure}[t]
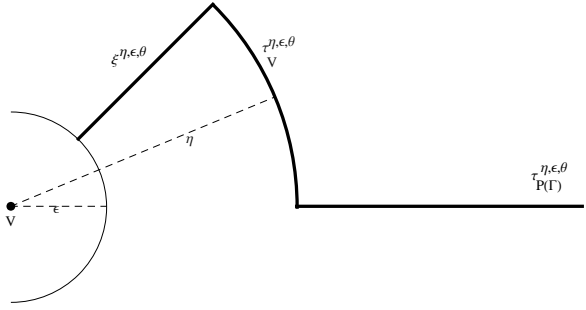
\centering
 \fige\caption{The monodromy chain, with angular variable $\theta$.}\label{fige}
 \end{figure}
 \end{proof}
 Define the chain $c^{\eta,\ve,\theta} = \sum_V (-1)^{|V|}\tau_V^{\eta,\ve,\theta} - \xi^{\eta,\ve,\theta}$. We have
 \eq{}{\partial c^{\eta,\ve,\theta} = \partial \sum_V (-1)^{|V|}\tau_V^{\ve,\ve,\theta}.
 }
 Note that $c^{\eta,\ve,0} = \sigma_{P(\Gamma)}^{\eta,\ve}$, i.e. all chains involving at least one circular variable die at $\theta=0$. We define the variation,
 \eq{2.6}{var(c^{\eta,\ve,\theta}) = c^{\eta,\ve,2\pi}-  c^{\eta,\ve,0} \equiv \sum_{V\subsetneq P(\Gamma)} (-1)^{|V|}\tau_V^{\ve,\ve}.
 }
 It is a sum of ``$(S^1)^p$-tubes'' over all $E_1\cap \cdots \cap E_p \subsetneq P(\Gamma)$.
 \section{The Graph Hypersurface}\label{secgrhyp}
 Associated to a graph $\Gamma$ with $n$ edges, one has the graph
 polynomial
 \eq{}{\psi_\Gamma(A_1,\dotsc,A_n) = \sum_T \prod_{e\not\in T}A_e
 }
 where $T$ runs through spanning trees of $\Gamma$. This polynomial has
 degree $h_1(\Gamma)$. For more detail, see \cite{BEK} and the
 references cited there. Let $X = X_\Gamma: \psi_\Gamma=0$ be the graph
 hypersurface in $\P^{n-1}$.
 For $\mu \subset \text{Edge}(\Gamma)$, let $L_\mu \subset \P(\Gamma)$
 be defined by $A_e=0,\ e\in \mu$. Let $\Gamma_\mu = \bigcup_{e\in \mu}e
 \subset \Gamma$ be the subgraph with edges in $\mu$. Note the
 dictionary $\Gamma_\mu \leftrightarrow L_\mu$ is inclusion reversing.
 \begin{lem}\label{lem2.1a} (i) $L_\mu \subset X_\Gamma \subset
   \P(\Gamma)$ if and only
   if $h_1(\Gamma_\mu)>0$. \newline\noindent
 (ii) If $h_1(\Gamma_\mu)>0$, there exists a unique $\nu \subseteq
 \mu$ such that $h_1(\Gamma_\nu) = h_1(\Gamma_\mu)$ and such that
 moreover $\Gamma_\nu$ is a core graph. \newline\noindent
 (iii) We have in (ii) that $\nu =\bigcup \xi$ where $\xi$ runs through all
 minimal subsets of $\mu$ such that $L_\xi \subset X$.
 \newline\noindent
 (iv) $L_\mu = L_\nu\cap M$, where $M$ is a coordinate linear space not
 contained in $X_\Gamma$.
 \end{lem}
 \begin{proof}These assertions are straightforward from the results in
   \cite{BEK}, section 3. Note that (iv) justifies our strategy of
   only blowing up core subgraphs.
 \end{proof}
 We have seen (remark \ref{prop1.6}) that our blowup $P(\Gamma)$ is
 stratified as a union of tori indexed by pairs
 \eq{2.2}{(F,
 \{\Gamma_p\subsetneq\cdots \subsetneq \Gamma_1\subsetneq
 \Gamma/\!/\gamma\})
 }
 where $F \subset \Gamma$ is a suitable subforest and
 the $\Gamma_i$ are core.
 \begin{prop}\label{prop2.2a} (i) As in proposition \ref{prop1.6}, the torus
   corresponding to \eqref{2.2} is
 \eq{2.3}{T(\Gamma_p/\!/F_p)\times T((\Gamma_{p-1}/\!/\Gamma_p)/\!/F_{p-1})\times \cdots
 \times T((\Gamma/\!/\Gamma_1)/\!/F).
 }
 Here $T(\Gamma):=\P(\Gamma)-\Delta$, where $\Delta:\prod_{e\in
   \text{Edge}(\Gamma)}  A_e = 0$. \newline\noindent
 (ii) The strict transform $Y$ of $X_\Gamma$ in $P(\Gamma)$ meets the
 stratum \eqref{2.3} in a union of pullbacks
 \eq{}{pr_1^{-1}(X^0_{\Gamma_p})\cup
   pr_2^{-1}(X^0_{\Gamma_{p-1}/\!/\Gamma_p})\cup\cdots \cup
   pr_p^{-1}(X^0_{(\Gamma/\!/\gamma)/\!/\Gamma_1}).
 }
 Here the $pr_i$ are the projections to the various subtori in
 \eqref{2.3}, and $X^0$ denotes the restriction of the corresponding
 graph hypersurface to the open torus in the projective space.
 \end{prop}
 \begin{proof}Let $\Gamma' \subset \Gamma$ be a subgraph and let
   $L:A_e=0,\ e\in \text{Edge}(\Gamma')$. Assume $h_1(\Gamma')>0$, so $L
   \subset X_\Gamma$. Let $P_L \to \P(\Gamma)$ be the blowup of
   $L$. Let $E_L \subset P_L$ be the exceptional divisor,
   and let $Y_L\subset P_L$ be the strict transform of $X_\Gamma$.
 The basic geometric
   result (op. cit. prop. 3.5) is that $E_L = \P(\Gamma')\times
   \P(\Gamma/\!/\Gamma')$ and
 \eq{2.5}{Y_L \cap E_L = \Big(X_{\Gamma'}\times
   \P(\Gamma/\!/\Gamma')\Big)\cup \Big(\P(\Gamma')\times
   X_{\Gamma/\!/\Gamma'}\Big) .
 }
 The assertions of the proposition follow by an induction argument.
 \end{proof}
 \begin{cor}\label{cor3.3} The strict transform $Y$ of $X_{\Gamma}$ in $P(\Gamma)$ does
 not meet the non-negative points $P(\Gamma)(\R^{\ge 0})$ \eqref{2.1a}.
 \end{cor}
 \begin{proof}It suffices by \eqref{2.1a} to show that $Y$ doesn't meet the positive
   points in any stratum. By proposition \ref{prop2.2a}, it suffices to
   show that for any graph $\Gamma$, the graph hypersurface $X_\Gamma$
   has no $\R$-points with coordinates all $>0$. This is immediate
   because $\psi_\Gamma$ is a sum of monomials with non-negative
   coefficients.
 \end{proof}
 \begin{rmk}\label{rmk3.4}The Feynman amplitude is obtained by calculating an integral over $\sigma = \P(\Gamma)(\R^{\ge 0})$ with an integrand which has a pole along $X_\Gamma$. Again using that $\psi_\Gamma$ is a sum of monomials with non-negative coefficients, one sees from lemma \ref{lem2.1a} that
 \eq{}{\sigma\cap X_\Gamma = \bigcup_\mu L_\mu(\R^{\ge 0})
 }
 where $L_\mu \leftrightarrow \Gamma_\mu$ with $\Gamma_\mu \subset \Gamma$ a core subgraph. The iterated blowup $P(\Gamma) \to \P(\Gamma)$ is exactly what is necessary to separate the non-negative real points from the strict transform of $X_\Gamma$.
 \end{rmk}
 \begin{rmk}The points where $\psi_\Gamma \neq 0$ have some remarkable
properties. It is shown in \cite{S} that for any angular sector $S$
with angle $< \pi$, $\psi_\Gamma(a_1,\dotsc,a_n) \neq 0$ at any
complex projective point $a$ such that the $a_i \neq 0$ and all the $\arg(a_i)$
lie in $S$.
 \end{rmk}
 \section{Monodromy}\label{secmono}
 Let $p_i = (0,\dotsc,1,0,\dotsc,0) \in \C^n$ be the $i$-th coordinate
 vector. Define
 $$\sigma^{aff} = \{\sum_{i=1}^n\tau_ip_i \ |\ \tau_i\ge
 0,\ \sum \tau_i = 1\}\subset \C^n-\{(0,\dotsc,0)\}\to \P^{n-1}.
 $$
 Fix a positive constant $\ve <<1$ and choose $q_k=(q_{k1},\dotsc,q_{kn}) \in
 \R^n,\ 1\le k\le n$ with $1-\ve<q_{kj} \le 1$ and $|q_{jk}-q_{\ell,m}|
 \le \ve^2$. We assume the $q_k$ are algebraically generic. Write $r_k(t) =
 p_k+tq_k \in \C^n$. Define (cf.\ fig.(\ref{figf}))
 \eq{6.1a}{\sigma^{aff}_t = \{\sum_{i=1}^n\tau_kr_k(t) \ |\ \tau_k\ge
 0,\ \sum \tau_k = 1\}
 }
  \begin{figure}[t]
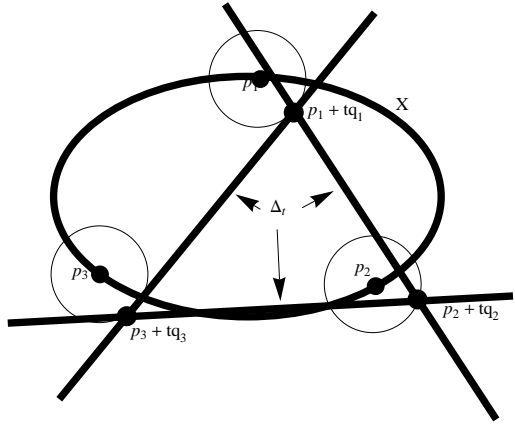
\centering
\figf\caption{Moving $\Delta_t$.}\label{figf}
 \end{figure}
 We write $\sigma$ and $\widetilde\sigma_t$ for the images of these chains
 in $\P^{n-1}$. Of course, $\sigma = \sigma_{\P^{n-1}}$ as above, and
 we know that $\sigma\cap X_\Gamma = \bigcup_{L\subset \sL}\sigma_L$. Here $\sL$ is as in \eqref{1.4}.
 \begin{lem}\label{lem3.4} Let $\sL \subset N_\sL$ be a neighborhood of $\sL$ in
   $\P^{n-1}$ and let $\sigma\subset N_\sigma$ be a neighborhood of $\sigma$. Then there exists $\ve_0>0$ such that $\ve\le \ve_0$
   implies that for all $0\le \theta\le 2\pi$, we have $\widetilde\sigma_{\ve e^{i\theta}} \subset N_\sigma$ and
   $\widetilde\sigma_{\ve e^{i\theta}}\cap X_\Gamma \subset N_\sL$.
 \end{lem}
 \begin{proof}We have $\sigma\cap X_\Gamma\subset \sL$. By compacity, $\widetilde\sigma_{\ve e^{i\theta}} \subset N_\sigma$ for $\ve <<1$. Again by compacity, if we shrink $N_\sigma$ we will have $N_\sigma \cap X_\Gamma \subset N_\sL$.
 \end{proof}
 \begin{rmk}\label{rmk3.5} Write $H_{k,t}$ for the projective span of the points
 $$r_1(t),\dotsc,\widehat{r_k(t)},\dotsc,r_n(t),$$
 and let $\Delta_t = \bigcup_{k=1}^n H_{k, t}$. Thus, $\Delta = \Delta_0$ and we may consider the {\it monodromy} for $\Delta_{\ve e^{i\theta}},\ 0 \le \theta\le 2\pi$. More precisely, renormalization in physics involves an integral over the chain $\sigma$. The integrand has poles along $X_\Gamma$. Since $\sigma\cap X_\Gamma \neq \emptyset$, the integral is possibly divergent. On the other hand, by corollary \ref{cor3.3}, the chain $\sigma_\ve$ does not meet $X_\Gamma$ and so represents a singular homology class
 \eq{3.6a}{[\sigma_\ve] \in H_{n-1}(\P^{n-1}-X_\Gamma, \Delta_\ve-\Delta_\ve\cap X_\Gamma,\Z).
 }
 (Since all $q_{kj}>0$, it follows that $\sigma_\ve \subset \sigma$, and
 points in $\sigma_\ve$ have all coordinates $>0$.)
 We consider the topological pairs $(\P^{n-1}-X_\Gamma, \Delta_{\ve e^{i\theta}}-\Delta_{\ve e^{i\theta}}\cap X_\Gamma)$ as a family over the circle and we continuously deform our chain $\sigma_\ve$ to a family of chains $\sigma_{\ve e^{i\theta}}$ on $\P^{n-1}-X_\Gamma$ with boundary on $ \Delta_{\ve e^{i\theta}}-\Delta_{\ve e^{i\theta}}\cap X_\Gamma$. (We will not be able to take $\sigma_{\ve e^{i\theta}} = \widetilde\sigma_{\ve e^{i\theta}}$ because this chain can meet $X_\Gamma$.) The monodromy map $m$ is an automorphism of \eqref{3.6a} obtained by winding around the circle: $m(\sigma_{\ve }) = \sigma_{\ve e^{2\pi i}}$. We will calculate $m(\sigma_{\ve}) $ and see that it determines in a natural way the renormalization expansion we want.
 \end{rmk}
 Recall we have $\pi: P(\Gamma) \to \P(\Gamma)$, and $\pi^{-1}(X_\Gamma) = Y_\Gamma\cup E$, where $Y=Y_\Gamma$ is the strict transform of $X_\Gamma$ and $E = \bigcup E_i$ is the exceptional divisor. (The $E_i$ are closures of orbits associated to core subgraphs of $\Gamma$.) We may transfer our monodromy problem to $P(\Gamma)$. $\Delta_{\ve e^{i\theta}}$ is in general position with respect to the blowups, so we obtain a family of divisors $\Delta_{\ve e^{i\theta}}' = \pi^*\Delta_{\ve e^{i\theta}}$ on $P(\Gamma)$. Since $\pi: P(\Gamma)-E-Y_\Gamma \cong \P(\Gamma)-X_\Gamma$, we have an isomorphism of topological pairs
 \eq{}{\Big(P(\Gamma)-E-Y_\Gamma, \Delta_{\ve e^{i\theta}}'-\Delta_{\ve e^{i\theta}}'\cap (E\cup Y_\Gamma)\Big) \cong
 \Big(\P(\Gamma) - X_\Gamma,\Delta_{\ve e^{i\theta}}-\Delta_{\ve e^{i\theta}}\cap X_\Gamma\Big).
 }
 In section \ref{sectopch} we have defined chains $\tau_V^{\eta,\ve,\theta}, \xi^{\eta,\ve,\theta}, c^{\eta,\ve,\theta}$ on $P(\Gamma)$. These chains sit on (or, in the case of $\xi$, within) various $(S^1)^p$-bundles over $P(\Gamma)(\R^{\ge 0})$ where the $S^1$ have radius $\eta$ with respect to a chosen metric. From corollary \ref{cor3.3} it follows that for $0< \eta<<1$, none of these chains meets $Y_\Gamma$. By construction, these chains do not meet $E$, so they may be identified with chains on $\P(\Gamma) - X_\Gamma$. We claim that a small modification of the chains $c^{\eta,\ve,\theta}$ will represent the monodromy chains $\sigma_{\ve e^{i\theta}}$.
 The monodromy chains  $\sigma_{\ve e^{i\theta}}$ should have boundary on $\Delta_{\ve e^{i\theta}}$. On the other hand, the chains $c^{\eta,\ve,\theta}$ were cut off so they had boundaries on tubes a distance $\ve$ from the toric divisors $D_j$ given by the strict transforms of the $A_j=0$ (see fig.(\ref{figg})).
  \begin{figure}[t]\centering
 \figg\caption{The chain $\tau_E^{\eta,\ve}$.}\label{figg}
 \end{figure}
  We must ``massage'' these brutal cutoffs to get them into $\Delta_{\ve e^{i\theta}}$.  Our chains $\tau$ sit on tubes or products of tubes or products of tubes of radius $\eta$ which we can think of as lying on $\P^{n-1}-\sL$. Since $\ve <<\eta$, when we deform $\Delta \to \Delta_{\ve e^{i\theta}}$ the homotopy type of the circles, or product of circles where these divisors intersect the tubes doesn't change. This may seem strange because $\sL \subset \Delta$ while $\Delta_{\ve e^{i\theta}}$ is in general position with respect to $\sL$, but the intersections with a hollow tubular neighborhood of $\sL$ are canonically homotopic. Indeed, we may take $\Delta_{\ve e^{i\theta}}$ to correspond to a point in a small contractible disk in the moduli space for coordinate simplices around the point corresponding to $\Delta$. The canonical path up to homotopy between the two points in moduli will induce the desired homotopy on the intersections. (See fig.(\ref{figh}). The two sets of four dots on the circles are canonically homotopic.).
 \begin{figure}[t]
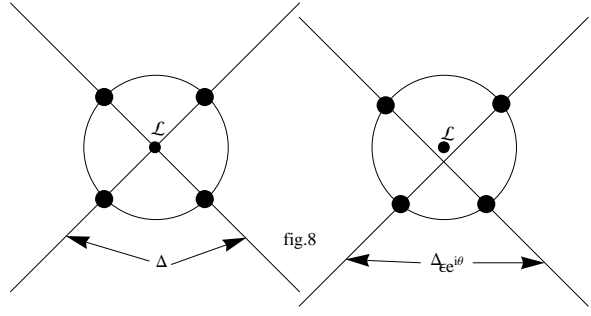
\centering
\figh\caption{Homotopy invariance of $\Delta_t\cap \text{tube over }\sL$.}\label{figh}
 \end{figure}
 In more detail,  by corollary \ref{cor3.3}, the chains $\tau^{\eta,\ve,\theta}$ are bounded away from $X_\Gamma$ by a bound which is independent of $\ve$ as $\ve \to 0$. Outside of some tubular neighborhood $N$ of $X_\Gamma$ we may
 find a space $M$ disjoint from $X_\Gamma$ such that $M$ contains open
 neighborhoods of both $\Delta-N\cap \Delta$ and
 $\Delta_{\ve e^{i\theta}}-N\cap \Delta_{\ve e^{i\theta}}$ and such that we
 have deformation retractions $M \to \Delta-N\cap \Delta$ and $M \to
 \Delta_{\ve e^{i\theta}}-N\cap \Delta_{\ve e^{i\theta}}$. Shrinking $\ve$, we
 may assume our $\ve$-cutoffs lie in $M$. We may then use the
 deformation retract to extend the chain slightly to a chain
 $\tilde\tau_V^{\eta,\ve,\theta}$ which bounds on $\Delta_{\ve e^{i\theta}}$.
 It remains to consider the chains $\xi^{\eta,\ve,\theta}$. Recall these
 were obtained by flowing the chain $\partial \sum_V (-1)^{|V|}
 \tau_V^{\eta,\ve,\theta}$ inward toward the exceptional divisor $E$,
 so $\eta \to \ve$ (cf.\ fig.(\ref{fige})). We are in a small neighborhood of
 $E(\R^{\ge 0})$ hence by corollary \ref{cor3.3} we are away from
 $X_\Gamma$. The point to be checked is that the term $\partial \sum_V
 (-1)^{|V|} \tau_V^{\ve,\ve,\theta}$ is very close to
 $\Delta_{\ve e^{i\theta}}$ so by the same deformation retraction argument
   as above we can extend the chain to bound on
   $\Delta_{\ve e^{i\theta}}$. The subtlety is that we are $\ve$-close to
   $E$ as well, so we need the distance from $\Delta_{\ve e^{i\theta}}$ to
   be $o(\ve)$. Recall \eqref{6.1a} we have the vertices $r_k(\ve e^{i\theta}) = [q_{k1}\ve e^{i\theta},\dotsc,1+q_{kk}\ve e^{i\theta},\dotsc,q_{kn}\ve e^{i\theta}] \in \P^{n-1}$. The coordinate divisor $\Delta_{\ve e^{i\theta}}$ is determined by these projective points. The projective point does not change if we scale the coordinates by $e^{i\theta}$, so the image in $\P^{n-1}$ of the affine simplex below, parametrized by $\tau_1,\dotsc,\tau_n \ge 0,\ \sum \tau_j = 1$, will have boundary in $\Delta_{\ve e^{i\theta}}$:
 \ml{4.4a}{ e^{i\theta}\tau_1(1+ \ve e^{i\theta}q_{11},\dotsc, \ve e^{i\theta}q_{1n})+\cdots \\
 +e^{i\theta}\tau_p( \ve e^{i\theta}q_{p1},\dotsc,1+ \ve e^{i\theta}q_{pp},\dotsc, \ve e^{i\theta}q_{pn}) \\
 + \tau_{p+1}( \ve e^{i\theta}q_{p+1,1},\dotsc, 1+\ve e^{i\theta}q_{p+1,p+1},\dotsc, \ve e^{i\theta}q_{p+1,n})+\cdots  \\
 +  \tau_n( \ve e^{i\theta}q_{n1},\dotsc, \ve e^{i\theta}q_{nn}+1).
 }
 Consider for example $\partial \tau_V^{\ve,\ve,\theta}$ where $V$ is
 the orbit closure corresponding to the blowup of $A_1=\cdots
 =A_p=0$. Take in \eqref{4.4a} $\tau_1,\dotsc,\tau_p\le \ve$ so terms
 in $\tau_j\ve$ may be neglected for $j\le p$. Take $u_j :=
 \tau_j/\tau_k$ where $k>p$ is chosen so that (say) $\tau_k \ge 1/n$. As
 a consequence, $u_1,\dotsc, u_p \le n\ve$. The corresponding
 projective point can then be written
 \ml{4.5a}{\Big[e^{i\theta}(u_1+\ve)+O(\ve^2),\dotsc, e^{i\theta}(u_p+\ve)+O(\ve^2),\\
 u_{p+1}+e^{i\theta}\ve+O(\ve^2),\dotsc,u_n+e^{i\theta}\ve+O(\ve^2)\Big].
 }
 The boundary is given by setting one or more of the $u_j=0$. Points in
 $\partial \tau_V^{\ve,\ve,\theta}$ can be approximated by points
 \eqref{4.5a} which then deform into $\Delta_{\ve e^{i\theta}}$. To see
 this, note that since $V$ is a codimension $1$ orbit closure, there
 will locally be one coordinate on $P(\Gamma)$ near $V$ which takes the
 constant value $\ve e^{i\theta}$ on $\partial \tau_V^{\ve,\ve,\theta}$
 (cf.\ fig.(\ref{fige})). On the other hand, \eqref{4.5a} is in homogeneous
 coordinates on $\P(\Gamma)$. To transform to $P(\Gamma)$ near a
 general point of $V$, one fixes $\ell\le p$ and looks at ratios
 \eq{}{\frac{e^{i\theta}(u_j+\ve)+O(\ve^2)}{e^{i\theta}(u_\ell+\ve)+O(\ve^2)}
 }
 for $1\le j\neq \ell\le p$. Clearly, at the boundary $u_\ell=0$ we will get
 $p-1$ coordinates $u_j/\ve + O(\ve)$ which are close to $\R^{\ge 0}$,
 and one coordinate (corresponding
 to the local defining equation for $V$) of the form $\ve
 e^{i\theta}+O(\ve^2)$. The remaining coordinates on $V$ are ratios of
 the $u_j+\ve e^{i\theta}+O(\ve^2),\ j\ge p+1$. Since $u_k=1$, these
 ratios are again close to $\R^{\ge 0}$. The calculation for orbit
 closures $V$ of codimension $\ge 2$ in $P(\Gamma)$ is similar and is
 left for the reader. We have proven
 \begin{prop}\label{prop6.3} With notation as above, the monodromy of the chain
   $\sigma_\ve \in H_{n-1}(\P^{n-1}-X_\Gamma,\Delta_\ve-X_\Gamma\cap
   \Delta_\ve)$ is represented by the chains $\tilde
   c^{\eta,\ve,\theta}$ given by modifying the chains
   $c^{\eta,\ve,\theta}$ to have boundary in $\Delta_{\ve e^{i\theta}}$. In
     particular, the monodromy $m(\sigma_\ve) =\tilde
   c^{\eta,\ve,2\pi}$ is given by
 \eq{}{m(\sigma_\ve) = \sum_V (-1)^{|V|}\tilde\tau_V^{\ve,\ve}
 }
 where $\tilde\tau_V^{\ve,\ve}$ is the chain $\tau_V^{\ve,\ve}$ defined
 in section \ref{sectopch} with boundary extended to $\Delta_\ve$ as above.
 \end{prop}
 It will be convenient to simplify the notation and write
 \eq{}{\tau_V^\ve :=\tilde\tau_V^{\ve,\ve}.
 }

\section{Parametric representations}\label{parrep}
In this section we list well-known representations of the Feynman rules and then prepare for a subsequent analysis of short-distance singularities in terms of
mixed Hodge structures.
\subsection{Kirchhoff--Symanzik  polynomials}
Let
\bea
\psi(\Gamma) & = & \sum_{T}\prod_{e\not\in T}A_e,\\
\phi(\Gamma) & = & \sum_{T_1\cup T_2=T} Q(T_1)\cdot Q(T_2)\prod_{e\not\in T_1\cup T_2}A_e,
\eea
be the two homogenous Kirchhoff--Symanzik polynomials \cite{ItzZ,Tod}. Here, $T$ is a spanning tree of the 1PI graph $\Gamma$ and $T_1,T_2$ are disjoint trees which together cover all vertices of $\Gamma$. Also,
$Q(T_i)$ is the sum of all external momenta attached to vertices covered by $T_i$. Note that $\phi(\Gamma)$ can be written as
\be \sum_\textrm{kinetic invariants $(q_i\cdot q_j)$} R_{q_i\cdot q_j}.\ee
Here, $q_i$ are external momenta attached to $T_1$ and $q_j$ to $T_2$, and $R_{q_i\cdot q_j}$ are rational functions of the edge variables only, and the sum is over independent such kinematical invariants where momentum conservation has been taken into account.
We extend the definition to the empty graph $\One$ by $\psi(\One)=1$, $\phi(\One)=0$.

Let $|\cdot |_\gamma$ denote the degree of a polynomial with regard to variables of the graph $\gamma$.
\begin{lem}\label{lem4} i) $\deg \phi = \deg \psi + 1$. \\
ii) \be \psi(\Gamma)=\psi(\Gamma/\!/\gamma)\psi(\gamma)+\psi_{\Gamma,\gamma} \label{psigg}\ee
with $|\psi_{\Gamma,\gamma}|_\gamma>|\psi(\gamma)|_\gamma$ for all core graphs $\Gamma$ and subgraphs $\gamma$.\\
iii)\be \phi(\Gamma)=\phi(\Gamma/\!/\gamma)\psi(\gamma)+\phi_{\Gamma,\gamma} \label{phigg}\ee
with $|\phi_{\Gamma,\gamma}|_\gamma>|\psi(\gamma)|_\gamma$ for all core graphs $\Gamma$ and subgraphs $\gamma$.
\end{lem}
\noindent Proof: i) by definition, ii) has been proved in \cite{BEK}, iii) follows similarly by noting that the two-trees of $\phi$ are obtained from the spanning trees of $\psi$ by removing an edge. If that edge belongs to $\Gamma/\!/\gamma$, we get $\phi(\Gamma/\!/\gamma)\psi(\gamma)$. If it belongs to $\gamma$, we get a monomial
$m$ with $|m|_\gamma>|\psi(\gamma)|_\gamma$.\hfill $\Box$\\
Note that it might happen that $\phi(\Gamma/\!/\gamma)=0$, if the external momenta flows through subgraphs $\gamma$ only. In such a case
(which can lead to infrared divergences) one easily shows $\phi_{\Gamma,\gamma}=\psi(\Gamma/\!/\gamma)\phi(\gamma)$.
\subsection{Feynman rules}
>From these polynomials one constructs the Feynman rules of a given theory. For example
we have in $\phi^4$ theory for a vertex graph $\Gamma$, $\textrm{sdd}(\Gamma)=0$,
\be \Phi(\Gamma)=\int_{\mathbb{R}_{>\epsilon}^k}\frac{e^{-\sum_{\textrm{edges $e$}} A_e m_e^2-\frac{\phi(\Gamma)}{\psi(\Gamma)}}}{\psi^2(\Gamma)}d\!A_1\cdots d\!A_{|\Gamma^{[1]}|}.\label{7.6a}\ee
We will write $\int_{>\epsilon} d\!A_{\Gamma}$ to abbreviate the affine chain of integration.

The integral is over the $k$-dimensional hypercube of positive real coordinates in $\mathbb{R}_{>\epsilon}$ with a small strip of width $1\gg\epsilon>0$ removed at each axis.
We regard the integrand
\be
\mathbb{C}\ni \iota(\Gamma):=\frac{e^{-\sum_e A_e m_e^2-\frac{\phi(\Gamma)}{\psi(\Gamma)}}}{\psi^2(\Gamma)},
\ee
$\iota(\Gamma)=\iota\left(\Gamma)(\{m^2\},\{q_i\cdot q_j\},\{A\}\right)$ as a function of the set of internal masses $\{m^2\}$, the set of external momenta $\{q_i\cdot q_j\}$ (which can be considered as labels on external half-edges) and the set of graph coordinates $\{A\}$, and $\iota$ takes values in $ \mathbb{C}$.
We often omit the $A$ dependence and abbreviate $P=\{m\},\{q_i\cdot q_j\}$ for all these external parameters of the integrand: $\iota=\iota(P)$.
The renormalization schemes we consider are determined by the condition that the Green function shall vanish at a particular renormalization point $R$,
so that renormalization becomes an iterated sequence of subtractions \be \iota_-(P,R):=\iota(P)-\iota(R).\ee
We let ${\textrm{sdd}}(\Gamma)$ be the superficial degree of divergence of a graph $\Gamma$ given as (see also Eq.(\ref{sdd}) for a refined version)\be {\textrm{sdd}}(\Gamma):=D|\Gamma|-\sum_{\textrm{edges $e$}}w_e-\sum_{\textrm{vertices $v$}}w_v,\ee
where $|\Gamma|$ is the rank of the first Betti homology, $D$ the dimension of spacetime which we keep as an integer, $w_e$ the weights of the propagator for edge $e$ as prescribed by free field theory and $w_v$ the weight of the vertex as prescribed by the interaction Lagrangian.
Note that we can set the width $\epsilon$ to zero, $\int_{>\epsilon} d\!A_\Gamma\to \int_{>0}d\!A_\Gamma$ if the integrand $\iota_-(\Gamma)$
is evaluated on a graph $\Gamma$ which has no divergent subgraphs.

Throughout, we assume that all all masses and external momenta are in general position so that there are no zeroes in the $\phi$-polynomial off the origin
for positive values of the $A$ variables.  In particular, we assume that the point $P$ is chosen appropriately away from all mass-shell and kinematical singularities. We remind the reader of the notation $(\Gamma,\sigma)$ (section (\ref{extleg})) where $\sigma$ stores all the necessary detail on how to evaluate the graph $\Gamma$.

A special role is played by the evaluations $(\Gamma,\sigma_{P=0})$. They set all internal masses and momenta to zero. Note that this leads immediately to infrared divergences: the Feynman integrands $\iota(\cdot)(P=0)$ are missing the exponential in the numerator, which provides a regulator at large values of the $A$ variables, and hence an infrared regulator.
The ultraviolet singularities at small values of the $A$ variables are taken into account by the renormalization procedure itself, and hence by our limiting
mixed Hodge structure. We will eliminate the case $P=0$ below using that $\iota_-$ evaluates to zero if there is no dependence on masses or external momenta.

\subsection{General remarks on renormalization and QFT}
We now consider the renormalization Hopf algebra $H_\sR$ of 1PI Feynman graphs in section (\ref{renhopf}).
We use the notation
\be \Delta(\Gamma)=\sum_\gamma \gamma\otimes \Gamma/\!/\gamma,\label{cop}\ee
for its coproduct. Also, $\Delta(\One)=\One\otimes \One$. Projection $P$ into the augmentation ideal on the rhs is written as
\be ({\textrm{id}}\otimes P)\Delta(\Gamma)=\sum_{\emptyset\not=\Gamma/\!/\gamma} \gamma\otimes \Gamma/\!/\gamma,\label{copP}\ee
so that for example the antipode $S$ is
\be S(\Gamma)=-\sum_{\emptyset\not=\Gamma/\!/\gamma}S(\gamma)\Gamma/\!/\gamma=:-\bar{\Gamma}.\ee
Furthermore, we introduce a forest notation for the antipode:
\be S(\Gamma)=\sum_{\textrm{[for]}}(-1)^{|\textrm{[for]}|}\Gamma/\!/\textrm{[for]}\prod_{j=1}^{|\textrm{[for]}|}\gamma_{\textrm{[for]},j},\ee
where the sum is over all forests $\textrm{[for]}$ and the product is over all subgraphs which make up the forest.
Here, a forest [for] is a possibly empty collection of proper superficially divergent 1PI subgraphs $\gamma_{\textrm{[for]},j}$ of $\Gamma$
which are mutually disjoint or nested. We call a forest [for] maximal if $\Gamma/\!/\textrm{[for]}$ is a primitive element of the Hopf algebra.
As edge sets \be \Gamma=\left(\Gamma/\!/\mathrm{[for]}\right)\cup \left(\cup_j\gamma_j\right).\label{fornot}\ee
This is in one-to-one correspondence with the representation of the antipode as a sum over all cuts on rooted trees $\rho_\sT(\Gamma)$  as detailed in section (\ref{ssecrth}) above.
The integer $|\textrm{[for]}|$ is the number of edges removed in this representation.

Let us first assume that the graph $\Gamma$ and all its core subgraphs have a non-positive superficial degree of divergence, so they are convergent or provide log-pole: ${\textrm{sdd}}\leq 0$ for all elements in (the complement of) the forests.

As the integrand $\iota(\Gamma)(P)$ depends on $P=\{m\},\{q_i\cdot q_j\}$ only through the argument of the exponential, we redefine the second Kirchhoff--Symanzik polynomial as follows:
\be \phi(\Gamma)(\{q_i\cdot q_j\})\to \varphi(\Gamma)(P):=\phi(\Gamma)(\{q_i\cdot q_j\})+\psi(\Gamma)\sum_e A_em_e^2.\ee
Then, the unrenormalized integrand is
\be \iota(\Gamma)(P)=\frac{\exp^{-\frac{\varphi(\Gamma)(P)}{\psi(\Gamma)}}}{\psi^2(\Gamma)}.\ee
With this notation,
the renormalized integrand is (in all sums and products over $j$ here and in the following, $j$ runs from $1$ to $|\textrm{[for]}|$)
\bea\label{7.16a} \iota_R(\Gamma)(P,R) & = & \sum_{\textrm{[for]}}(-1)^{\textrm{[for]}}
\frac{\exp{-\left(\frac{\varphi(\Gamma/\!/\textrm{[for]})(P)}{
\psi(\Gamma/\!/\textrm{[for]})}+\sum_j\frac{\varphi(\gamma_j)(R)}{
\psi(\gamma_j)}\right)}}{\psi^2(\gamma/\!/\textrm{[for]})\prod_j\psi^2(\gamma_j)}\nonumber\\
& & -\sum_{\textrm{[for]}}(-1)^{\textrm{[for]}}
\frac{\exp{-\left(\frac{\varphi(\Gamma/\!/\textrm{[for]})(R)}{
\psi(\Gamma/\!/\textrm{[for]})}+\sum_j\frac{\varphi(\gamma_j)(R)}{
\psi(\gamma_j)}\right)}}{\psi^2(\gamma/\!/\textrm{[for]})\prod_j\psi^2(\gamma_j)}\label{reno}\\
 & =: & \bar{\iota}(\Gamma)(P,R)+S^\iota(\Gamma)(R),\nonumber
\eea
where $+S^\iota(\Gamma)(R)=-\bar{\iota}(\Gamma)(R,R)$ is the integrand for the counterterm, and $\bar{\iota}(\Gamma)(P,R)$, the integrand in the first line,
delivers upon integrating Bogoliubov's $\Bar{R}$ operation.
Note that this formula (\ref{reno}) is just the evaluation \be m(S_R^\iota\otimes \iota)\Delta(\Gamma),\ee which guarantees that the corresponding Feynman integral
exists in the limit $\epsilon\to 0$ \cite{Kreimer},\cite{Tor}.

This Feynman integral is obtained by integrating from $\epsilon$ to $\infty$ each edge variable. For the renormalized Feynman integral $\Phi_R(\Gamma)(P)$ we can take the limit $\epsilon\to 0$, while for the $\bar{R}$-operation  \be \bar{\Phi}(\Gamma)(P,R;\epsilon)=\int_\epsilon \bar{\iota}(\Gamma)(P,R),\ee and the counterterm \be S_{R;\epsilon}^\Phi(\Gamma)=-\bar{\Phi}(\Gamma)(R,R;\epsilon),\ee the lower boundary remains as a dimension-full parameter in the integral.
Note that the result (\ref{reno}) above can also be written in the $P-R$ form, typical for renormalization schemes which  subtract by constraints on physical parameters:
\be \iota_R(\Gamma)(P,R)=\sum_{\emptyset\not=\Gamma/\!/\gamma}\left[\iota(\Gamma/\!/\gamma)(P)-\iota(\Gamma/\!/\gamma)(R)\right]S_{R;\epsilon}^\iota(\gamma),\label{diff}\ee
and as
\be \bar{\iota}(\Gamma)(P,R)=\sum_{\emptyset\not= \Gamma/\!/\gamma} S_{R;\epsilon}^\iota(\gamma)\iota(\Gamma/\!/\gamma)(P)\Rightarrow
\iota_R(\Gamma)(P,R)=\sum_{\gamma} S_{R;\epsilon}^\iota(\gamma)\iota(\Gamma/\!/\gamma)(P),\ee
using the notation (\ref{copP},\ref{cop}).
Similarly, for Feynman integrals,
\be \bar{\Phi}(\Gamma)(P,R;\epsilon)=\sum_{\emptyset\not= \Gamma/\!/\gamma} S_{R;\epsilon}^\Phi(\gamma)\Phi(\Gamma/\!/\gamma)(P),\;
\Phi_R(\Gamma)(P)=\lim_{\epsilon\to 0}\sum_{\gamma} S_{R;\epsilon}^\Phi(\gamma)\Phi(\Gamma/\!/\gamma)(P).\ee
When it comes to actually calculating the integral \eqref{7.6a} (or, in its renormalized form \eqref{7.16a}), something rather remarkable happens.  By lemma \ref{lem4}(i), the term in the exponential in these integrals is homogeneous of degree $1$ in the edge variables $A_i$. The assumption $\text{sdd}(\Gamma)=0$ means $dA/\psi^2$ is homogeneous of degree $0$. Making the change of variable $A_i=ta_i$, we find
\eq{}{dA/\psi(A)^2 = dt/t\wedge(\sum (-1)^{j-1}a_jda_1\wedge\cdots\wedge\widehat{da_j}\wedge\cdots)/\psi(a)^2 = dt/t\wedge\Omega/\psi^2.
}
Note that $\Omega/\psi^2$ is naturally a meromorphic form on the projective space $\P(\Gamma)$ with homogeneous coordinates the $a_i$. Writing $\sigma = \{a_i\ge 0\} \subset \P(\Gamma)(\R)$, we see that the renormalized integral can be rewritten up to a term which is $O(\ve)$ as a sum of terms of the form
\eq{7.25a}{\int_\sigma \Omega/\psi_j^2\int_\ve^\infty \Big(e^{(-tf_j(a))}-e^{(-tg_j(a))}\Big)dt/t = \int_\sigma \Omega/\psi_j^2\Big(E_1(\ve f_j(a)) - E_1(\ve g_j(a))\Big),
}
where
\be E_1(z):=\int_1^\infty e^{-tz}\frac{d\!t}{t}=-\gamma_E-\ln z+O(z);\quad z\to 0\ee
is the exponential integral. (Here $f_j(a), g_j(a)$ are defined by taking the locus $a_i\ge 0, \sum a_i=1$.) As long as $f_j(a), g_j(a)>0$, we may allow $\ve \to 0$ for fixed $a$. The Euler constant and $\log \ve$ terms cancel. When the dust settles, we are left with the projective representation for the renormalized Feynman integral
\be \Phi_R(\Gamma)(P)= \int_\sigma \Omega_\Gamma \sum_{\textrm{[for]}}(-1)^{\textrm{[for]}}
\frac{\ln{\left(
\frac{\varphi(\Gamma/\!/\textrm{[for]})(P)\prod_j\psi(\gamma_j)+\sum_j\varphi(\gamma_j)(R)\psi(\Gamma/\!/\textrm{[for]})
\prod_{h\not= j}\psi(\gamma_h)}{
\varphi(\Gamma/\!/\textrm{[for]})(R)\prod_j\psi(\gamma_j)+\sum_j\varphi(\gamma_j)(R)\psi(\Gamma/\!/\textrm{[for]})
\prod_{h\not= j}\psi(\gamma_h)}\right)}}{\psi^2(\Gamma/\!/\textrm{[for]})\prod_j\psi^2(\gamma_j)}.\ee
Note that the use of $\sigma$ is justified as long as the integrand has all subdivergences subtracted, so is in the $\bar{\iota}$ form, so that lower boundaries in the $a_i$ variables can be set to zero indeed.

By (\ref{diff}), this can be equivalently written as
\be \Phi_R(\Gamma)(P)=\lim_{\epsilon\to 0}\sum_{\gamma}S_{R;\epsilon}^\Phi(\gamma)
\int_{>\epsilon}d\!A_{\Gamma/\!/\gamma}\iota_-(\Gamma/\!/\gamma)(P,R),\ee
in any renormalization scheme which is described by kinematical subtractions $P\to R$.
\begin{rem}
It will be our goal to replace the affine
$\int d\! A$
by the projective
$\int d\!\Omega$
in the above. The presence of lower boundaries, which can not be ignored as the integrand has divergent subgraphs, allows this only upon introducing
suitable chains $\tau_\gamma^\epsilon$ as discussed in previous sections.
\end{rem}
Next, we relax the case of log-divergence.
\subsection{Reduction of graphs with ${\rm{ssd}}(\Gamma)>0$}\label{logpolereduct}
We start with an example. To keep things simple but not too simple, we consider the one-loop self-energy graph in $\phi^3_6$ theory,
a scalar field theory with a cubic interaction in six dimensions of space-time.
We have
\be \Phi(\Gamma)(P)=\int_{>\epsilon} d\!A_\Gamma\iota(\Gamma)(P)=\int_{>\epsilon} d\!A_\Gamma\frac{e^{-\frac{\varphi(\Gamma)}{\psi(\Gamma)}}}{\psi(\Gamma)^3}\equiv \int_\epsilon^\infty d\!A_1 d\!A_2  \frac{e^{-\frac{m^2(A_1+A_2)^2+q^2A_1A_2}{(A_1+A_2)}}}{(A_1+A_2)^3}.\ee
We will renormalize by suitable subtractions at chosen values of masses and momenta in the $\varphi$-polynomial.
We hence (with subdivergences taken care of by suitable bar-operations $\iota\to\bar{\iota}$ in the general case) replace $\iota(\Gamma)(P)$ by $\iota(\Gamma)(P)-\iota(\Gamma)(0)$, as this leaves $\iota_-(\Gamma)(P,R)$ invariant.

Then the above can  be written, with this subtraction, and by the familiar change of variables $A_i=t a_i$, and by one partial integration in $t$,
\bea \Phi(\Gamma)(P) & = & \int_\sigma d\!\Omega \int_\epsilon^\infty \frac{d\!t}{t} \frac{[m^2(a_1+a_2)^2+q^2a_1a_2]e^{-t\frac{m^2(a_1+a_2)^2+q^2a_1a_2}{(a_1+a_2)}}}{(a_1+a_2)^4}\nonumber\\
 & & - \int_\sigma d\!\Omega  \frac{[m^2(a_1+a_2)^2+q^2a_1a_2]}{(a_1+a_2)^4},\label{exquad}\eea
 where we expanded the boundary term up to terms constant in $\epsilon$, which gave the term in the second line. We discarded already the pure pole term $\sim 1/\epsilon$ from $\Phi(\Gamma)(P=0)=\int_{>\epsilon}\frac{d\!A}{(A_1+A_2)^3}$ $=$ $\int_\epsilon^\infty d\!t/t^2 \int_0^\infty db_2 1/(1+b_2)^3$.

 Note that graphs $\Gamma$ with $\textrm{sdd}>0$ have $\textrm{res}(\Gamma)=2$. They hence depend on a single kinematical invariant $q^2$ say,
 $\phi(\Gamma)=\phi(\Gamma)(q^2)$, for which we write $\phi(\Gamma)_{q^2}$.

 The result in (\ref{exquad}) leads us to define two top-degree forms. (Here $\Omega=a_1da_2-a_2da_1$ and we still write $\phi,\psi$ for the Kirchhoff--Symanzik polynomials regarded as dependent on either $a_i$ or $A_i$ variables below).
 \be \omega_\Box=\omega_\Box(\Gamma)=\Omega\frac{\phi_1(\Gamma)}{\psi(\Gamma)^4}=\Omega\frac{a_1a_2}{(a_1+a_2)^4},\ee
and
\be \omega_{m^2}=\omega_{m^2}(\Gamma)=\Omega\frac{(a_1+a_2)^2}{\psi(\Gamma)^4}=\Omega\frac{1}{(a_1+a_2)^2},\ee
so that
\bea \Phi(\Gamma)(P) & = & -m^2\int_\sigma [\omega_\Box+\omega_{m^2}]-(q^2-m^2)\int_\sigma \omega_\Box\\
 & & +m^2 \int_\sigma [\omega_\Box+\omega_{m^2}] \int_\epsilon^\infty \frac{d\!t}{t}e^{-t\frac{\varphi(\Gamma)(P)}{\psi(\Gamma)}}\nonumber\\
 & & +(q^2-m^2)\int_\sigma \omega_\Box \int_\epsilon^\infty \frac{d\!t}{t}e^{-t\frac{\varphi(\Gamma)(P)}{\psi(\Gamma)}}.\nonumber
\eea
There are corresponding affine integrands
\bea \iota_\Box(\Gamma) & = & \frac{\phi_1(\Gamma)}{\psi(\Gamma)^4}e^{-\frac{\varphi(\Gamma)(P)}{\psi(\Gamma)}},\\
\iota_{m^2}(\Gamma) & = & \frac{(a_1+a_2)^2}{\psi(\Gamma)^4}e^{-\frac{\varphi(\Gamma)(P)}{\psi(\Gamma)}}.
\eea

The graph $\Gamma$ is renormalized by a choice of a renormalization condition $R_\Box$ for the coefficient of $q^2-m^2$ (wave function renormalization),
and by the choice of a condition $R_{m^2}$ for the mass renormalization. $R$ is often still used to denote the pair of those.
\be \Phi(\Gamma)(P)+m^2\delta_{m^2}+q^2z_\Box=\Phi_{R_\Box,R_{m^2}}(\Gamma)(P).\ee

The mass counterterm is then
\be m^2\delta_{m^2}=-m^2\int_\sigma  [\omega_\Box+\omega_{m^2}]\left(1-\int_\epsilon^\infty \frac{d\!t}{t}e^{-t\frac{\varphi(\Gamma)(R_{m^2})}{\psi(\Gamma)}}\right),\label{massct}\ee
and the wave-function renormalization $q^2z_\Box$ is
\be q^2z_\Box=-q^2\int_\sigma  \omega_\Box\left(1-\int_\epsilon^\infty \frac{d\!t}{t}e^{-t\frac{\varphi(\Gamma)(R_{\Box})}{\psi(\Gamma)}}\right).\label{kinct}\ee
Note the term 1 in the $()$ brackets does not involve exponentials.

The corresponding renormalized contribution is
\be \Phi_{R}(\Gamma)(P)=(q^2-m^2)\int_\sigma \omega_\Box \ln{\frac{\varphi(P)}{\varphi(R_\Box)}}+m^2\int_\sigma[\omega_\Box+\omega_{m^2}]\ln{\frac{\varphi(\Gamma)(P)}{\varphi(\Gamma)(R_{m^2})}}.\ee

The transition from the unrenormalized contribution to the renormalized one is particularly simple upon defining Feynman rules in accordance with external leg structures:
\bea
\Phi((\Gamma,\sigma_\Box)) & = & (q^2-m^2)\int_{>\epsilon}d\!A \frac{\phi_1(\Gamma)e^{-\frac{\varphi(\Gamma)(P)}{\psi(\Gamma)}}}{\psi(\Gamma)^{D/2+1}},\\
\Phi((\Gamma,\sigma_{m^2})) & = & m^2\int_{>\epsilon}d\!A \frac{[\phi_1(\Gamma)+\psi(\Gamma)\sum_e  A_e] e^{-\frac{\varphi(\Gamma)(P)}{\psi(\Gamma)}}}{\psi(\Gamma)^{D/2+1}},\\
\eea
so that renormalization proceeds as before on log-divergent integrands.

This example extends straightforwardly to the case of $\Gamma$ having divergent subgraphs. Let us return to $\phi_4^4$ theory and define for a core
graph $\Gamma$ with $\textrm{sdd}(\Gamma)=2$, (so that it is a self-energy graph and hence has only two external legs, and thus a single kinematical invariant $q^2$), and graph-polynomials $\psi(\Gamma)$, $\phi(\Gamma)=\phi_{q^2}(\Gamma)$, $\varphi(\Gamma)=\varphi(\Gamma)(P)=\phi_{q^2}(\Gamma)+\psi(\Gamma)\sum_eA_em_e^2,$
the forms
\be \omega_\Box(\Gamma)=\Omega_\Gamma \frac{\phi_1(\Gamma)}{\psi^3(\Gamma)},\ee
\be \omega_{m^2}(\Gamma)=\Omega_\Gamma  \frac{\phi_1(\Gamma)+\psi(\Gamma)\sum_e A_e}{\psi^3(\Gamma)}.\ee
The corresponding complete affine integrands $\iota_\Box,\iota_{m^2}$ are immediate replacing $a_i$ by $A_i$ variables, and multiplying by exponentials $\exp{-\varphi(\Gamma)(P)/\psi(\Gamma)}$, with $P\to R$
for counterterms.

One finds by a straightforward computation
\bea \Phi_{R_\Box}((\Gamma,\sigma_\Box))(P) & = & \sum_\gamma S^\Phi_{R;\epsilon}(\gamma)\int_{{\epsilon}} \omega_\Box(\Gamma/\!/\gamma)\ln{\frac{\varphi(\Gamma/\!/\gamma)(P)}{\varphi(\Gamma/\!/\gamma)(R_\Box)}}\\
 & = & \int\Omega_\Gamma  \sum_{\textrm{[for]}}(-1)^{\textrm{[for]}}\omega_\Box(\Gamma/\!/\textrm{[for]})\ln{\frac{\varphi(\Gamma/\!/\textrm{[for]})(P)}{\varphi(\Gamma/\!/
 \textrm{[for]})(R_\Box)}},
\eea
and
\bea \Phi_{R_{m^2}}((\Gamma,\sigma_{m^2}))(P) & = & \sum_\gamma S^\Phi_{R;\epsilon}(\gamma)\int_{{\epsilon}} \omega_{m^2}(\Gamma/\!/\gamma)\ln{\frac{\varphi(\Gamma/\!/\gamma)(P)}{\varphi(\Gamma/\!/\gamma)(R_{m^2})}}\\
 & = & \int\Omega_\Gamma   \sum_{\textrm{[for]}}(-1)^{\textrm{[for]}}\omega_{m^2}(\Gamma/\!/\textrm{[for]})\ln{\frac{\varphi(\Gamma/\!/\textrm{[for]})(P)}{\varphi(\Gamma/\!/\textrm{[for]})(R_{m^2})}}.
\eea
We set \be\Phi_R(\Gamma)(P)\equiv\Phi_R((\Gamma,\One))(P)=\phi_{R_{\Box}}((\Gamma,\sigma_\Box))(P)+\Phi_{R_{m^2}}((\Gamma,\sigma_{m^2}))(P),\ee
in the external leg structure notation of section (\ref{extleg}).
    We can combine the results for graphs $\Gamma$ for all degrees of divergence  $\textrm{sdd}(\Gamma)\geq 0$ by defining $\omega(\Gamma)=\Omega/\psi^2(\Gamma)$ for a log divergent graph with the results above.
And that's that.
Well, we have to hasten and say a word about the Feynman rules when the subgraphs $\gamma$ have $\textrm{sdd}(\gamma)>0$, and hence also about $S^\Phi_{R;\epsilon}(\gamma)$ in that case.

We use, with $P$ the projection  into the augmentation ideal, the notation
\be \bar{\Gamma}=\Gamma+m(S\circ P\otimes P)\Delta=:\Gamma+(\Gamma^\prime)^{-1}\Gamma^{\prime\prime}.\ee
Let us consider the quotient Hopf algebra given by quadratically divergent graphs: $\Delta_2(\Gamma)=\sum_{\gamma,\textrm{sdd}(\gamma)=2}\gamma\otimes
\Gamma/\!/\gamma$. We write
\be \Delta_2(\Gamma)=:\Gamma\otimes \One+\One\otimes\Gamma+\Gamma^\prime_2\otimes \Gamma^{\prime\prime}.\ee

We add $0=+{\Gamma_2^\prime}^{-1}\Gamma^{\prime\prime}-{\Gamma_2^\prime}^{-1}\Gamma^{\prime\prime}$, so
\bea \bar{\Gamma} & = & \Gamma+{\Gamma_2^\prime}^{-1}\Gamma^{\prime\prime}-{\Gamma_2^\prime}^{-1}\Gamma^{\prime\prime}+{\Gamma^\prime}^{-1}\Gamma^{\prime\prime}\\
 & = & \left(\Gamma+{\Gamma_2^\prime}^{-1}\Gamma^{\prime\prime}\right)+\left({\Gamma^\prime}^{-1}-{\Gamma_2^\prime}^{-1}\right)\Gamma^{\prime\prime}.
\eea
Here the sum is over all terms of the coproduct with the $\Gamma_2^\prime$ terms being present whenever $\Gamma^\prime$ is quadratically divergent.

Evaluating the terms $\Gamma_2^\prime$ by $1/\psi^2(\gamma_2^\prime)=\iota(\gamma_2^\prime)(P=0)$ decomposes the bar-operation on the level of integrands as follows.
\be \bar{\iota}(\Gamma)(P)=\overbrace{\left(\iota(\Gamma)(P)+{\iota(\Gamma_2^\prime}^{-1})(P=0)\iota(\Gamma^{\prime\prime})(P)\right)}^{I}
+\iota({\Gamma^\prime}^{-1})(R)\iota(\Gamma^{\prime\prime})(P),\ee
where $\iota({\Gamma^\prime}^{-1})(R)\equiv S^{\iota}_{R;\epsilon}(\Gamma^\prime)$ appears because a subtraction of a $P=0$ term,
from a quadratically divergent term,
precisely delivers those counterterms by our previous analysis. Note that they contain terms which do not have an exponential, as in the example
(\ref{kinct},\ref{massct}). Often, as a two-point vertex of mass type improves the powercounting of the co-graph, we might keep self-energy subgraphs massless, in which case only terms involving $R_\Box$ contribute.

We are left to decompose the terms denoted $I$.
We find by direct computation
\bea I & = & \overbrace{\left[ \omega(\Gamma)+\omega({\Gamma_2^{\prime}}^{-1})\omega(\Gamma^{\prime\prime})\right]e^{-\frac{\varphi(\Gamma)(P)}{\psi(\Gamma)}}}^{II}\\
 & & -\underbrace{\omega({\Gamma_2^\prime}^{-1})\omega(\Gamma^{\prime\prime})\left[e^{-\frac{\varphi(\Gamma)(P)}{\psi(\Gamma)}}
 -e^{-\frac{\varphi(\Gamma/\!/\Gamma^\prime_2)(P)}{\psi(\Gamma/\!/\Gamma^\prime_2)}}\right]}_{III}.
\eea
The terms denoted $II$ gives us the final integrand $\iota(\Gamma)(P)$ with a corresponding form $\omega_{II}(\Gamma)$. $\omega_{II}(\Gamma)=\omega(\Gamma)$ if there are no subgraphs with $\textrm{sdd}=2$.     Note that $II$  has the full $\Gamma$ as an argument in the common exponential, \be II=\omega_{II}(\Gamma)\exp(-\varphi(\Gamma)/\psi(\Gamma)),\ee which defines $\omega_{II}$.
The rational coefficient $\omega_{II}$ has log-poles only for all subgraphs including the ones with $\textrm{sdd}=2$.

The terms $III$ is considered in $t,a_i$ variables. We can integrate $t$ as before. As the rational part of the integrand factorizes in
$\Gamma^\prime_2$ and $\Gamma^{\prime\prime}$ variables, we similarly decompose the former into $s,b_i$, $i\in {\Gamma_2^\prime}^{[1]}$, variables.
We note $s$ only appears in the
log (after the $t$ integration) as a coefficient of $\phi_{\Gamma,\Gamma^\prime_2}$, using Lemma (\ref{lem4}). Partial integration in $s$ eliminates the log
and delivers a top-degree form for the $b_i$ integration. These terms precisely compensate against the constant terms mentioned above,
as $\phi_{\Gamma,\Gamma_2^\prime}=\phi_1(\Gamma_2^\prime)\phi_1(\Gamma-\Gamma_2^\prime)$, using that $\textrm{res}(\Gamma_2^\prime)=2$.

We hence summarize

\begin{thm}\label{cructhm}
\be \Phi_R(\Gamma)(P)=\lim_{\epsilon\to 0}\sum_\gamma S^\Phi_{R;\epsilon}(\gamma)\int_{\epsilon} \omega_{II}(\Gamma/\!/\gamma)\ln{\frac{\varphi(\Gamma/\!/\gamma)(P)}{\varphi(\Gamma/\!/\gamma)(R)}}.\ee
It is understood that each counterterm is computed with a subtraction $R$ as befits its argument $\gamma$, and forms $\Gamma$ are chosen in accordance with the previous derivations. Here, $\omega_{II}$ is constructed to have log-poles only.
As a projective integral this reads
\bea \Phi_R(\Gamma)(P) & = &  \int \Omega_\Gamma \sum_{\rm{[for]}}(-1)^{\rm{[for]}}\times\nonumber\\ & &
\times\ln{\left(
\frac{\varphi(\Gamma/\!/\rm{[for]})(P)\prod_j\psi(\gamma_j)+\sum_j\varphi(\gamma_j)(R)\psi(\Gamma/\!/\rm{[for]})
\prod_{h\not= j}\psi(\gamma_h)}{
\varphi(\Gamma/\!/\rm{[for]})(R)\prod_j\psi(\gamma_j)+\sum_j\varphi(\gamma_j)(R)\psi(\Gamma/\!/\rm{[for]})
\prod_{h\not= j}\psi(\gamma_h)}\right)}\nonumber\\ & & \times\omega(\Gamma/\!/\rm{[for]})\prod_j\omega(\gamma_j).\eea
\end{thm}

\begin{rem}Similar formulas can be obtained for the bar-operations and counterterms, with the same rational functions in the integrands, and exponentials
$\exp(-\varphi(\Gamma/\!/\gamma)(X)/\psi(\Gamma/\!/\gamma))$, with $X=P$ or $X=R$ as needed.
\end{rem}

\begin{rem}
We have worked with choices of renormalizations for mass and wave functions, $R\to R_\Box, R_{m^2}$. One can actually also define $P\to P_\Box,P_{m^2}$, and for example set masses to zero in all exponentials ($\varphi(\cdot)(P)\to\phi_{q^2}(\cdot)$), that's essentially the Weinberg scheme if one then subtracts at $q^2=\mu^2$.
\end{rem}
\begin{rem}
This all is nicely reflected in properties of analytic regulators. For example in dimensional regularization the identity $\int d^Dk [k^2]^\rho=0$, $\forall \rho$, leads to $\Phi(\Gamma)(P=0)=0$ immediately, where $\Phi$ now indicates unrenormalized Feynman rules using that regulator.
\end{rem}
\begin{rem}
We are working so far with constant lower boundaries. The chains introduced in previous sections have moving lower boundaries which respect the hierarchy in each flag. We will study that difference in section (\ref{seclmhsvsren}).
\end{rem}

\subsection{Specifics of the MOM-scheme}
We define the MOM-scheme by setting all masses to zero in radiative corrections and keeping a single kinematical invariant $q^2$ in the $\phi$-polynomial,
$P=\{0\},\{q_i\cdot q_j\sim q^2\}$,
\be \phi(\Gamma)=q^2R_{q^2}(\Gamma).\ee
Such a situation arises if we set masses to zero (possibly after factorization of a polynomial part from the amplitude as in the Weinberg scheme), and for vertices if we consider
the case of zero momentum transfers, or evaluate at a symmetric point $q_i^2=q^2$,
     where $i$ denotes the external half-edges of $\Gamma$.
If we want to emphasize the $q^2$ dependence we write $\phi_{q^2}$. Trivially, $\phi_{q^2}=q^2\phi_1$.
In the MOM-scheme, subtractions are done at $q^2=\mu^2$, which defines $R$ for all graphs. Counter-terms in the MOM-scheme become very simple when expressed in parametric integrals thanks to the homogeneity of the $\phi$-polynomial. Note that we hence have $\varphi(\Gamma)=\phi(\Gamma)$ as we have set all masses to zero.

In a MOM-scheme, renormalized diagrams are polynomials in $\ln q^2/\mu^2$:
\begin{thm}
For all $\Gamma$, \be\Phi_{\rm{MOM}}(\Gamma)(q^2/\mu^2)=\sum_{j=1}^{\rm{aug}(\Gamma)}c_j(\Gamma)\ln^j{q^2/\mu^2}.\ee
\end{thm}
Here, $\textrm{aug}(\Gamma)=\max_{\textrm{[for]}}|\textrm{[for]}|$.\\
Proof: Consider a sequence $\gamma_1\subsetneq\gamma_2\cdots\gamma_{\textrm{aug}(\Gamma)}\subsetneq\Gamma$.
This is in one-to-one correspondence with some decorated rooted tree appearing in $\rho_\sR(\Gamma)$ \eqref{2.35b}. Choose one edge $e_j\in \gamma_j/\gamma_{j-1}$ in each decoration and de-homogenize with respect to that edge. We get a sequence of lower boundaries $\epsilon,\epsilon/A_2,\epsilon/A_2/A_3,\cdots$.
Use the affine representation and integrate to obtain the result.\hfill$\Box$\\
\subsubsection{MOM scheme results from residues}
In such a scheme, it is particularly useful to take a derivative with respect to $\ln q^2$.
We consider
\be p_1(\Gamma):=q^2\partial_{q^2}\Phi_{\textrm{MOM}}(\Gamma)(q^2/\mu^2)_{|_{q^2=\mu^2}},\label{pone}\ee where we evaluate at $q^2=\mu^2$
after taking the derivative. This number, which for a primitive element of the renormalization Hopf algebra is the residue of that graph in the sense
of \cite{BEK}, is our main concern for a general graph. It will be obtained in the limit of the limiting mixed Hodge structure we construct.

\begin{rem}
It is not that this limit would not exist for general schemes. But the limit would be a complicated function of ratios of masses and kinematical invariants, which has a constant term given by the number $p_1(\Gamma)$ and beyond that a dependence on these    ratios which demands a much finer Hodge theoretic study than we can offer here.
\end{rem}

But first we need
to remind ourselves how coefficients of higher powers of logarithms of complicated graphs related to coefficients of lower powers of sub- and co-graphs thanks to the
renormalization group.
\subsubsection{The counterterm $S_{\rm MOM}^\Phi$}
For $S_{\textrm{MOM}}^\Phi(\Gamma)=:\sum_{j=1}^{\textrm{aug}(\Gamma)}s_j(\Gamma)\ln^j\mu^2$, we simply use the renormalization group or the scattering type formula.
In particular, we have
\be S_{\textrm MOM}^\Phi(\Gamma)=\sum_{j=1}^{\textrm{aug}(\Gamma)}\frac{1}{j!}(-1)^j \underbrace{[p_1\otimes\cdots\otimes p_1]}_{{j \textrm{factors}}}\Delta^{j-1}(\Gamma).\ee
This is easily derived \cite{RHII,KY1} upon noting that $p_1(\Gamma) =\Phi(S\star Y(\Gamma))$.

Note that this determines counter-terms by iteration: for a $k$-loop graph, knowledge of all the lower order counterterms suffices to determine
all contributions to the $k$-loop counterterm but the lowest order coefficient of $\ln \mu^2$. But then, that coefficient is given by the formula
\be s_1(\Gamma)=p_1(\Gamma)\ln \mu^2,\ee
which itself only involves counter-terms of less than $k$ loops,     by the structure of the bar operation.
\subsubsection{$p_1(\Gamma)$ from co-graphs}

We can now summarize the consequences of the renormalization group and our projective representations for parametric representations of Feynman integrals.
The interesting question is about the logs which we had in numerators. Thm.(\ref{cructhm}) becomes
\begin{thm}\label{cruc}
\be
p_1(\Gamma)=\lim_{\epsilon\to 0}\sum_\gamma S^\Phi_{\rm{MOM};\epsilon}(\gamma)q^2\partial_{q^2}\int_{\epsilon} \omega_{II}(\Gamma/\!/\gamma)\ln{\phi_{q^2/\mu^2}(\Gamma/\!/\gamma)}.\label{p1g}\ee
This limit is
\bea p_1(\Gamma) & = &  \int \Omega_\Gamma \sum_{\rm{[for]}}(-1)^{\rm{[for]}}\times\nonumber\\ & &
\times q^2\partial_{q^2} \ln{\left(\phi_{q^2/\mu^2}(\Gamma/\!/\rm{[for]})\prod_j\psi(\gamma_j)+\sum_j\phi_1(\gamma_j)\psi(\Gamma/\!/\rm{[for]})
\prod_{h\not= j}\psi(\gamma_h)\right)}\nonumber\\ & & \times\omega(\Gamma/\!/\rm{[for]})\prod_j\omega(\gamma_j).\eea
The derivative with respect to $\ln q^2$ can be taken inside the integral in (\ref{p1g}) if and only if all edges carrying external momentum are in the complement $C(\Gamma)$
of all edges belonging to divergent subgraphs. In that case, $q^2\partial_{q^2}\ln{\phi_{q^2}(\Gamma/\!/\gamma)}=1$ and no logs in the numerator appear.
\end{thm}
\begin{rem}
Note that overlapping divergent graphs can force all edges to belong to divergent subgraphs, cf.\ Fig.(\ref{overl}).
\end{rem}
\noindent Proof: If all edges carrying external momentum are in the complement to divergent subgraphs, we bring the counter-terms under the integrand using the bar-operation.
We can take the limit $\epsilon\to 0$ in the integrand for all edge variables belonging to subgraphs, and this limit commutes with the derivative with respect to $\ln q^2$ by assumption: each $\phi_q^2(\Gamma/\!/\gamma)$ is a linear combination of terms $A_e\psi_e(\Gamma/\!/\gamma)$, where $e$ is in that
complement $C(\Gamma)$ of subgraph edges, and $\psi_e(\Gamma/\!/\gamma)=\psi(\Gamma/\!/\gamma/e)$.  Applying then the Chen-Wu theorem \cite{Smirnov} with respect to the
elements of $C(\Gamma)$ disentangles the $q^2$ dependence from the limit in $\epsilon$.\hfill$\Box$\\
\begin{rem}
Note that the discussion below with respect to the limiting Hodge structure assumes that we have this situation of disentanglement of divergent
subgraphs and edges carrying external momentum.
We hence have no logarithms in the numerator. But note that the general case does no harm to the ensuing discussion: by Lemma (\ref{lem4}), any logarithms in the numerator are congruent to one along any exceptional divisor of $X_{\Gamma/\!/\rm{[for]}}$.
Furthermore, when external momentum interferes with subgraphs, all logs can be turned to rational functions by a partial integration.
The fact that the second Kirchhoff--Symanzik polynomial is a linear combination of  $\psi$-polynomials, applied to graphs with an extra shrunken edge,
in the MOM-case establishes these rational functions to have poles coming from our analysis of this $\psi(\Gamma)$ polynomial.
A full mathematical
discussion of this "$\int\omega\ln f$" situation should be subject to future work.
\end{rem}

\subsubsection{Examples}
>From now on we measure $q^2$ in units of $\mu^2$ so that subtractions are done at $1$. This simplifies notation.
Let us first consider the Dunce's cap in detail, (\ref{duncedetail}).
\begin{figure}[t]
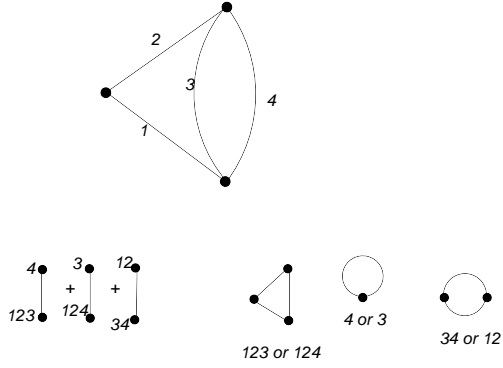
\dunce\caption{The Dunce's cap, again. We label the edges $1,2,3,4$. Resolved in trees, we find three trees in the core Hopf algebra.
We label the vertices by edge labels of the graph. The sets $123$ and $124$ correspond to a triangle graph as indicated, the sets $12$ and $34$ are one-loop vertex graphs, and tadpoles appear in the coproduct on the rhs for edges $3$ or $4$. The coproduct in the core Hopf algebra  is, expressed in edge labels,
$\Delta^\prime(1234)=123\otimes 4+124 \otimes 3 + 34 \otimes 12$.  Only the last term appears in the renormalization Hopf algebra.}\label{duncedetail}\end{figure}
We have the following data ($\textrm{path}_q(\Gamma)$ refers to the momentum path through the graph):
\bea
\psi(\Gamma) & = & (A_1+A_2)(A_3+A_4)+A_3A_4,\\
\psi(\gamma) & = & A_3+A_4, \psi(\Gamma/\!/\gamma)=A_1+A_2,\\
\textrm{path}_q(\Gamma) & = & e_1,\\
\phi_1(\Gamma) & = & A_1(A_2A_3+A_3A_4+A_4A_2)=A_1\psi(\Gamma/\!/e_1)=A_1\psi^1(\Gamma),\\
\phi_1(\gamma) & = & A_3A_4, \phi_1(\Gamma/\!/\gamma)=A_1A_2, \\
\{\textrm{[for]}\} & = & \{\emptyset, (34)\}.
\eea
\be \Phi(\Gamma)_\epsilon(q^2)=\int_\epsilon^\infty \prod_{i=1}^4d\!A_i\frac{\exp{-q^2\frac{\phi_1(\Gamma)}{\psi(\Gamma)}}}{\psi^2(\Gamma)}.
\ee
Hence we choose a function $\tau(\epsilon)$ which goes to zero rapidly enough so that $\lim_{\epsilon\to 0}\tau(\epsilon)/\epsilon=0$ and compute
\bea
 \bar{\Phi}_\epsilon(\Gamma)(q^2,\mu^2) & = &
 \int_\epsilon^\infty d\!A_1d\!A_2\int_{\tau(\epsilon)}^\infty d\!A_3d\!A_4
 \left\{
 \frac{\exp{-q^2\frac{\phi_1(\Gamma)}{\psi(\Gamma)}}}{\psi^2(\Gamma)}\right.\nonumber\\ & &
- \left.\frac{\exp{\left[-q^2\frac{\phi_1(\Gamma/\!/\gamma)}{\psi(\Gamma/\!/\gamma)}\right]}}{\psi^2(\Gamma/\!/\gamma)}
\frac{\exp{\left[-\frac{\phi_1(\gamma)}{\psi(\gamma)}\right]}}{\psi^2(\gamma)}
\right\}\\
  & = &
 \int_{q^2\epsilon}^\infty d\!A_1d\!A_2\int_{\tau(\epsilon)q^2}^\infty d\!A_3d\!A_4
  \left\{
  \frac{\exp{-\frac{\phi_1(\Gamma)}{\psi(\Gamma)}}}{\psi^2(\Gamma)}\right\}\nonumber\\ & &
-\left\{  \int_{q^2\epsilon}^\infty d\!A_1d\!A_2\int_{\tau(\epsilon)}^\infty d\!A_3d\!A_4
\frac{\exp{\left[-\frac{\phi_1(\Gamma/\!/\gamma)}{\psi(\Gamma/\!/\gamma)}\right]}}{\psi^2(\Gamma/\!/\gamma)}
\frac{\exp{\left[-\frac{\phi_1(\gamma)}{\psi(\gamma)}\right]}}{\psi^2(\gamma)}
\right\}.
\eea
Let us now re-scale to variables $A_i\to A_1B_i$ for all variables $i\in 2,3,4$.
We get
\bea
 \bar{\Phi}_\epsilon(\Gamma)(q^2,\mu^2) & = &
 \int_{q^2\epsilon}^\infty \frac{d\!A_1}{A_1} \int_{q^2\epsilon}^\infty d\!B_2\int_{q^2\tau(\epsilon)/A_1}^\infty d\!B_3d\!B_4
  \left\{
  \frac{\exp{-A_1\frac{(B_2B_3+B_3B_4+B_4B_2)}{(1+B_2)(B_3+B_4)+B_3B_4}}}{[(1+B_2)(B_3+B_4)+B_3B_4]^2}\right\}\nonumber\\ & &
-\left\{   \int_{q^2\epsilon}^\infty \frac{d\!A_1}{A_1} \int_{q^2\epsilon}^\infty d\!B_2\int_{\tau(\epsilon)/A_1}^\infty d\!B_3d\!B_4
 \frac{\exp{-A_1\frac{B_2}{1+B_2}}}{(1+B_2)^2}
\frac{\exp{-A_1\frac{B_3B_4}{B_3+B_4}}}{(B_3+B_4)^2}
\right\}.
\eea
We re-scale once more $B_4=B_3C_4$. Also, we set the lower boundaries in the $B_2$ and $C_4$ integrations to zero. This is justified as $A_1$ and $B_3$ remain positive.
\bea
 \bar{\Phi}_\epsilon(\Gamma)(q^2,\mu^2) & = &
 \int_{q^2\epsilon}^\infty \frac{d\!A_1}{A_1} \int_{0}^\infty d\!B_2\int_{q^2\tau(\epsilon)/A_1}^\infty \frac{d\!B_3}{B_3}\int_0^\infty d\!C_4
  \left\{
  \frac{\exp{-A_1\frac{(B_2+B_3C_4+C_4B_2)}{(1+B_2)(1+C_4)+B_3C_4}}}{[(1+B_2)(1+C_4)+B_3C_4]^2}\right\}\nonumber\\ & &
-
\int_{q^2\epsilon}^\infty \frac{d\!A_1}{A_1} \int_{0}^\infty d\!B_2\int_{\tau(\epsilon)/A_1}^\infty \frac{d\!B_3}{B_3}\int_0^\infty d\!C_4
\left\{ \frac{\exp{-A_1\frac{B_2}{1+B_2}}}{(1+B_2)^2}
\frac{\exp{-A_1B_3\frac{C_4}{1+C_4}}}{(1+C_4)^2}
\right\}.
\eea
Taking a derivative wrt $\ln q^2$ and using that $\lim_{\epsilon\to 0}\tau(\epsilon)/\epsilon=0$, delivers three remaining terms
\bea
 \partial_{\ln q^2}\bar{\Phi}_\epsilon(\Gamma)_{q^2=1} & = &
 \int_{0}^\infty d\!B_2\int_{\tau(\epsilon)/\epsilon}^\infty \frac{d\!B_3}{B_3}\int_0^\infty d\!C_4
  \left\{
  \frac{1}{[(1+B_2)(1+C_4)+B_3C_4]^2}\right\}\nonumber\\ & &
- \int_{0}^\infty d\!B_2\int_{\tau(\epsilon)/(q^2\epsilon)}^\infty \frac{d\!B_3}{B_3}\int_0^\infty d\!C_4
\left\{ \frac{1}{(1+B_2)^2}
\frac{e^{-\epsilon q^2\frac{B_3C_4}{1+C_4}}}{(1+C_4)^2}
\right\}\nonumber\\ & &
+ \int_{q^2\epsilon}^\infty \frac{d\!A_1}{A_1}\int_0^\infty d\!B_2\int_0^\infty d\!C_4
\left\{ \frac{e^{-A_1\frac{B_2}{(1+B_2)}}}{[(1+B_2)(1+C_4)]^2}
\right\}
.\eea
Integrating $B_3$ in the second line and $A_1$ in the third,  we find
\bea
 \partial_{\ln q^2}\bar{\Phi}_\epsilon(\Gamma)_{q^2=1} & = &
 \int_{0}^\infty d\!B_2\int_{\tau(\epsilon)/\epsilon}^\infty \frac{d\!B_3}{B_3}\int_0^\infty d\!C_4
    \frac{1}{[(1+B_2)(1+C_4)+B_3C_4]^2}\nonumber\\ & &
+ \ln{\tau(\epsilon)/\epsilon}\int\frac{\Omega_\gamma}{\psi^2(\gamma)}
\int\frac{\Omega_{\Gamma/\!/\gamma}}{\psi^2(\Gamma/\!/\gamma)}.\eea
Using the exponential integral, those $B_3$ and $A_1$ integrations also deliver finite contributions
\be
- \int_{0}^\infty d\!B_2\int_0^\infty d\!C_4
\left\{ \frac{1}{(1+B_2)^2}
\frac{\ln{\frac{C_4}{1+C_4}}}{(1+C_4)^2}
\right\}
+ \int_0^\infty d\!B_2\int_0^\infty d\!C_4
\left\{ \frac{\ln{\frac{B_2}{1+B_2}}}{[(1+B_2)(1+C_4)]^2}
\right\}=0.
\ee
This cancellation of logs is no accident: while in this simple example it looks as if it originates from the fact that the co-graph
and subgraph are identical, actually the cross-ratio
\be \ln \frac{\phi(\Gamma/\!/\gamma)\psi(\gamma)}{\psi(\Gamma/\!/\gamma)\phi(\gamma)}\ee
vanishes identically when integrated against the de-homogenized product measure
\be \int_0 dA_{\Gamma/\!/\gamma}dA_\gamma \frac{1}{\psi^2(\Gamma/\!/\gamma)\psi^2(\gamma)}.\ee

This is precisely because $C(\Gamma)=e_1$ has an empty intersection with $\gamma^{[1]}=e_3,e_4 $.

But then, this cancelation of logs will break down if $\phi(\Gamma)$ is not as nicely disentangled from $\phi(\gamma)$ for all log-poles as it is here,
and will be replaced by logs congruent to 1 along subdivergences in general, in accordance with Thm.(\ref{cruc}).

Let us study this in some detail.
Consider the graph on the upper left in Fig.(\ref{reconstr}), and consider the finite $\ln \phi/\psi$-type contributions of the exponential integral
to in the vicinity of the exceptional divisor for the subspace $A_3=A_4=0$.

Routing an external momentum through edges 1,6, we have the following graph polynomials:
\bea
\phi_1(\Gamma) & = & A_1[A_3A_4(A_5+A_6)+A_5A_6(A_3+A_4)+A_2(A_3+A_4)(A_5+A_6)]\\ & & +A_6A_5[(A_1+A_2)(A_3+A_4)+A_3A_4]\nonumber\\
\phi_1(\Gamma/34) & = & A_1[A_5A_6+A_2(A_5+A_6)]+A_5A_6[(A_1+A_2)]\\
\phi_1(34) & = & A_3A_4\\
\psi(\Gamma/34) & = & (A_1+A_2)(A_5+A_6)+A_5A_6\\
\psi(34) & = & A_3+A_4.
\eea
    We have $C(\Gamma)=e_1,e_6$, and $\cup_{{\gamma\subsetneq\Gamma,\textrm{res}(\gamma)\geq 0}}\gamma^{[1]}=e_3,e_4,e_5,e_6$.
The intersection is $e_6$.
We hence find, with suitable de-homogenization,
\be \frac{\ln{\frac{\overbrace{B_5B_6(1+B_2)}^X+\overbrace{B_5B_6+B_2(B_5+B_6)}^Y}{(1+B_2)(B_5+B_6)+B_5B_6}} -\ln{\frac{C_4}{1+C_4}}}{[(1+B_2)(B_5+B_6)+B_5B_6]^2[1+C_4]^2}d\!B_2d\!C_4d\!B_5d\!B_6.\ee
Here, the term $X$ denotes a term which would be absent if the momenta would only go through edge 1 and hence the above intersection would be empty, while $Y$ indicates the terms from the momentum flow through edge 1.

This is of the form $\ln(f_{\Gamma/\!/\gamma}/f_\gamma)[\omega_{\Gamma/\!/\gamma}\wedge \omega_\gamma]$.
If the term $X$ would be absent, a partial integration
\be \int_\epsilon^\infty \frac{\ln{\frac{xu+v}{xu+w}}}{(xu+w)^2}\sim \int_\epsilon^\infty\frac{1}{(xu+w)^2}\ee
would show the vanishing of this expression as above. The presence of $X$ leaves us with a contribution which can be written,
replacing $\ln{C_4/(1+C_4)}$ by $ \ln Y/\psi(\Gamma/34)$,
\be \frac{\ln{\frac{\overbrace{B_5B_6(1+B_2)}^X+\overbrace{B_5B_6+B_2(B_5+B_6)}^Y}{\underbrace{B_5B_6+B_2(B_5+B_6)}_Y}}}{[(1+B_2)(B_5+B_6)+B_5B_6]^2[1+C_4]^2}.\ee
As promised, it is congruent to one along the remaining log-pole at $A_5=A_6=0$. It has to be: the forest where the subgraph $56$ shrinks to a point looses the momentum flow through edge 6 and could not contribute any counterterm for a pole remaining in the terms discussed above.

Note that in general higher powers of logarithms can appear in the numerator as subgraphs can have substructure. Lacking a handle to notate all the log-poles which do not cancel due to partial integration identities known beyond mankind we consider it understood that all terms from the asymptotic expansion of the exponential integral up to constant terms (higher order terms in $\epsilon$ are not needed as all poles are logarithmic only) are kept without being shown explicitly in further notation. We emphasize though that all those logarithm terms in the numerator are congruent to one along log-poles -and deserve study in their own right elsewhere-, and hence thanks to Lemma (\ref{lem4}) which guarantees indeed all necessary cancelations, we have in all cases:
\be
p_1(\Gamma)=\lim_{\epsilon\to 0} \partial_{\ln q^2}\bar{\Phi}_\epsilon(\Gamma)_{q^2=1}.\ee
\begin{rem}There is freedom in the choice of $\tau$, a natural choice comes from the rooted tree representation $\rho(\Gamma)$ of the forest.
Each forest is part of a legal tree $t$ and any subgraph $\gamma$ corresponds to a vertex $v$ in that tree. If $d_v$ is the distance of $v$ to the root of $t$, $\tau(\epsilon)=\epsilon^{d_v+1}$ is a natural choice.
\end{rem}

\section{$N$}\label{secn}
\subsection{ for physicists: The antipode as monodromy}
Let us now come back to the core Hopf algebra and prepare for an analysis in terms of limiting mixed Hodge structures. This will be achieved in two steps: an analysis of the structure of the antipode of the renormalization Hopf algebra, which will then allow to define a matrix $N$ for the monodromy in question such that $S(\Gamma)$ can be expressed in a particularly nice way.      In fact, because of orientations, the $N$ which arises in the monodromy calculation is the negative of the $N$ computed in this section. We omit the minus sign to simplify the notation.

Let us consider the antipode first.
Thanks to the above lemma we can  write for the antipode $S(\Gamma)$
\be S(\Gamma)=-\sum_{j=0}^{|\Gamma|}(-1)^j\sum_{|C|=j}\sum_t P^C(t)R^C(t).\ee
    Here, we abuse notation in an obvious manner identifying $\Gamma$ and $\rho_T(\Gamma)$, the latter being the indicated sum over trees,
in accordance with Eq.(\ref{homomorph}).

We also define $\overline{R}(\Gamma)=-S(\Gamma)$.
Let us now label the edges of each $t(\Gamma)$ once and for all by $1,2,\cdots,|\Gamma|-1$.
Then, we have $|\Gamma|-1$ cuts $C$ with $|C|=1$, and \be \left({\genfrac{}{}{0pt}{}{|\Gamma|-1}{j}}\right)\ee
cuts of cardinality $|C|=j$.
We hence can define a vector $v(\Gamma)$ with $2^{|\Gamma|-1}$ entries in $H$, ordered
 according to a never decreasing cardinality of cuts: \be v(\Gamma)=(\Gamma,\underbrace{\sum_t P^C(t)R^C(t)}_{\textrm{$\left({\genfrac{}{}{0pt}{}{|\Gamma|-1}{1}}\right)$ entries of cardinality $1$}},\cdots,\underbrace{\sum_t P^C(t)R^C(t)}_{\textrm{$\left({\genfrac{}{}{0pt}{}{|\Gamma|-1}{j}}\right)$ entries of cardinality $j$}} ,\cdots)^T.\ee
 Example: Dunce's cap with edges $1,2,3,4$ and divergent subgraph $3,4$, comare Fig.(\ref{duncedetail}).
The core coproduct is
\be \Delta_c^\prime=123\otimes4+124\otimes3+34\otimes12.\ee
The vector $v$ is then
\be v=\left( {\genfrac{}{}{0pt}{}{1234}{(123)(4)+(124)(3)+(12)(34)}}\right).\ee
Let $N^{(2)}$ be the to-by-two matrix \be N^{(2)}=\left(
                                                    \begin{array}{cc}
                                                      0 & 1 \\
                                                      0 & 0 \\
                                                    \end{array}
                                                  \right)
.\ee
Note that
\be \left[\left(
      \begin{array}{cc}
        1 & 0 \\
        0 & 1 \\
      \end{array}
    \right)-N^{(2)}\right]\left(
                            \begin{array}{c}
                              1234 \\
                              (123)(4)+(124)(3)+(12)(34) \\
                            \end{array}
                          \right)
=\left(
   \begin{array}{c}
     \overline{R}(1234) \\
      (123)(4)+(124)(3)+(12)(34)\\
   \end{array}
 \right)
,\ee
with
\be \overline{R}(1234)=1234-(123)(4)-(124)(3)-(34)(12).\ee
In fact, it is our first task to find  a nilpotent matrix $N$, $N^{|\Gamma|}=0$, such that
\be \sum_{j=0}^{|\Gamma|-1}(-1)^jN^j/j!=(\overline{R}(\Gamma),\underbrace{\sum_t \overline{R}(P^C(t))\overline{R}(R^C(t))}_{\textrm{$\left({\genfrac{}{}{0pt}{}{|\Gamma|-1}{1}}\right)$ entries of cardinality $1$}},\cdots,\underbrace{\sum_t \overline{R}(P^C(t))\overline{R}(R^C(t))}_{\textrm{$\left({\genfrac{}{}{0pt}{}{|\Gamma|-1}{j}}\right)$ entries of cardinality $j$}} ,\cdots),)^T.\ee
For $P^C(t)=\prod_i t_i$ we here have abbreviated $\overline{R}(P^C(t))$ for $\prod_i \overline{R}(t_i)$.
\subsection{The matrix $N$}
Let $M(0,1)$ be the space of matrices with entries in the two point
set $\{0,1\}$.

Let now $m+1$ be the number of loops $m=|\Gamma|-1$  in the graph and let us construct a
nilpotent $2^{m}\times 2^{m}$ square matrix $N\equiv N^{(m)}$, $N^{m+1}=0$, in $M(0,1)$ as follows.

Consider first the $m+1$-th row of the Pascal triangle, for example for $m=3$ it reads $1,3,3,1$.
For this example, we will then construct blocks of sizes $1\times 1$, $1\times 3$, $3\times 3$, $3 \times 1$ and $1 \times 1$, all with entries either $0$ or $1$.

So this gives us in general  $m+2$ blocks $M^{(m)}_j$, $0\leq j\leq m+1$, of matrices of size $M^{(m)}_0:1\times 1$,
$M^{(m)}_1:1\times m$, $M^{(m)}_2:m\times m(m-1)/2!$, $\cdots$, $M^{(m)}_m:m\times 1$, $M^{(m)}_{m+1}:1\times 1$.

In the block $M^{(m)}_j$, $0\leq j<(m+2)/2$,
fill the columns, from left to right, by never increasing sequences of binary numbers (read from top to bottom)
where each such number contains $j$ entries $1$ for the block $M^{(m)}_j$. Put $M^{(m)}_0=(0)$ in the left upper corner and $M^{(m)}_1$ to the left of it. For $j\geq 2$, put the block $M^{(m)}_j$ below and to the right of the block $M^{(m)}_{j-1}$,
in $N$. All entries in $N$ outside these blocks are zero. Determine the entries of the blocks $M^{(m)}_j$, $m+1\geq j\geq (m+2)/2$, by the requirement that
$N^\bot=N$, where $N^\bot$ is obtained from $N$ by reflection along the diagonal which goes from the lower left to the upper right.
We write ${M^{(m)}_i}^\bot=M^{(m)}_{m+1-i}$. For odd integer $m$, we have ${M^{(m)}_{(m+1)/2}}^\bot=M^{(m)}_{(m+1)/2+1}$, by construction.
Here are $M^{(3)}_j$ and  $N,N^2,N^3$ for $m=3$:
\be M^{(3)}_0 =(0), M^{(3)}_1=(1,1,1), M^{(3)}_2=\left(
                                        \begin{array}{ccc}
                                          1 & 1 & 0 \\
                                          1 & 0 & 1 \\
                                          0 & 1 & 1 \\
                                        \end{array}
                                      \right), M^{(3)}_3=\left(
                                                       \begin{array}{c}
                                                         1 \\
                                                         1 \\
                                                         1 \\
                                                       \end{array}
                                                     \right), M^{(3)}_4=(0).
\ee
\bea \label{8.10} N^{(3)}& = & \left(
  \begin{array}{cccccccc}
    0 & |\underline{{\textbf 1}} & {\textbf 1} & \underline{{\textbf 1}}| & 0 & 0 & 0 & 0 \\
    0 & 0 & 0 & 0 & |\overline{{\textbf 1}} & {\textbf 1} & \overline{{\textbf 0}}| & 0 \\
    0 & 0 & 0 & 0 & |{\textbf 1} & {\textbf 0} & {\textbf 1}| & 0 \\
    0 & 0 & 0 & 0 & |\underline{{\textbf 0}} & {\textbf 1} & \underline{{\textbf 1}}| & 0 \\
    0 & 0 & 0 & 0 & 0 & 0 & 0 & |\overline{{\textbf 1}} \\
    0 & 0 & 0 & 0 & 0 & 0 & 0 & |{\textbf 1} \\
    0 & 0 & 0 & 0 & 0 & 0 & 0 & |\underline{{\textbf 1}} \\
    0 & 0 & 0 & 0 & 0 & 0 & 0 & 0 \\
  \end{array}
\right),\\  {N^{(3)}}^2 & = & \left(
                       \begin{array}{cccccccc}
                         0 & 0 & 0 & 0 & 2 & 2 & 2 & 0 \\
                         0 & 0 & 0 & 0 & 0 & 0 & 0 & 2 \\
                         0 & 0 & 0 & 0 & 0 & 0 & 0 & 2 \\
                         0 & 0 & 0 & 0 & 0 & 0 & 0 & 2 \\
                         0 & 0 & 0 & 0 & 0 & 0 & 0 & 0 \\
                         0 & 0 & 0 & 0 & 0 & 0 & 0 & 0 \\
                         0 & 0 & 0 & 0 & 0 & 0 & 0 & 0 \\
                         0 & 0 & 0 & 0 & 0 & 0 & 0 & 0 \\
                       \end{array}
                     \right),\\
                     {N^{(3)}}^3 & = & \left(
                       \begin{array}{cccccccc}
                         0 & 0 & 0 & 0 & 0 & 0 & 0 & 6 \\
                         0 & 0 & 0 & 0 & 0 & 0 & 0 & 0 \\
                         0 & 0 & 0 & 0 & 0 & 0 & 0 & 0 \\
                         0 & 0 & 0 & 0 & 0 & 0 & 0 & 0 \\
                         0 & 0 & 0 & 0 & 0 & 0 & 0 & 0 \\
                         0 & 0 & 0 & 0 & 0 & 0 & 0 & 0 \\
                         0 & 0 & 0 & 0 & 0 & 0 & 0 & 0 \\
                         0 & 0 & 0 & 0 & 0 & 0 & 0 & 0 \\
                       \end{array}
                     \right).
\eea
We can now write, for $1\leq j\leq m$,
\be N^j=j! n^{(m)}_j,\ee
where the matrix $n^{(m)}_j\in M(0,1)$, by construction. Hence
\be \exp{\left\{-LN^{(m)}\right\}}=\sum_{j=0}^{m} \frac{(-L)^j}{j!}{N^{(m)}}^j=\sum_{j=0}^{m}(-L)^jn^{(m)}_j.\ee
This is obvious from the set-up above.
Furthermore, directly from construction, $n^{(m)}_j$, $j\geq 1$, has a block structure into blocks of size \be
(1\times m),\cdots,\underbrace{\cdots}_{\textrm{ $j-1$ middle blocks missing}},\cdots,(m\times 1),\ee
located in the uppermost right corner of size $2^{m-j+1}\times 2^{m-j+1}$ as in the above example.

\subsection{Math:The Matrix $N$}
In this section we compute the matrix $N$ which gives the log of the monodromy. Because of orientations, the answer we get is the negative of the physical $N$ computed in the previous section.

Our basic result gives the monodromy
\eq{6.1}{m(\sigma_1) = \sum_I (-1)^{p}\tau_I = \sigma_1 + \sum_{I,\ p\ge 1} (-1)^{p}\tau_I .
}
Here we have changed notation. $I=\{i_1,\dotsc,i_p\}$ refers to a flag $\Gamma_{i_1}\subsetneq\cdots\subsetneq \Gamma_{i_p}\subsetneq \Gamma$ of core subgraphs. More generally
\eq{6.2}{m(\tau_I) = \sum_{J\supset I}(-1)^{q-p}\tau_J.
}
Here $J=\{j_1,\dotsc,j_q\}\supset I$. to verify \eqref{6.2}, consider e.g. the case corresponding to $\Gamma_1 \subsetneq \Gamma$. We have seen (lemma \ref{lem1.5}) that the blowup of $\P(\Gamma)$ along the linear space defined by the edge variables associated to edges of $\Gamma_1$ yields as exceptional divisor $E_1 \cong \P(\Gamma_1) \times \P(\Gamma/\!/\Gamma_1)$. In fact, the strict transform of $E_1$ in the full blowup $P(\Gamma)$ can be identified with $P(\Gamma_1)\times P(\Gamma/\!/\Gamma_1)$. To see this, note that by proposition \ref{prop1.6}, the intersection in $P(\Gamma)$ of distinct exceptional divisors $E_1\cap\cdots\cap E_p$ is non-empty if and only if after reordering, the corresponding core subgraphs of $\Gamma$ form a flag. This means, for example, that $E_1\cap E_I \neq \emptyset$ if and only if the flag corresponding to $I$ has a subflag of core subgraphs contained in $\Gamma_1$, and the remaining core subgraphs form a flag containing $\Gamma_1$. In this way, we blow up appropriate linear spaces in $\P(\Gamma_1)$ and in $\P(\Gamma\!/\Gamma_1)$. the result is $P(\Gamma_1)\times P(\Gamma/\!/\Gamma_1) \subset P(\Gamma)$. The chain $\tau_1$ is an $S^1$-bundle over the chain $\sigma_{\P(\Gamma_1)}\times \sigma_{\P(\Gamma/\!/\Gamma_1)}$ (slightly modified along the boundaries as above), and the monodromy map is the product of the    monodromies on each factor. (The monodromy takes place on $P(\Gamma_1)\times P(\Gamma/\!/\Gamma_1)$. In the end, one takes the $S^1$-bundle over $m(\sigma_{\P(\Gamma_1)}\times \sigma_{\P(\Gamma/\!/\Gamma_1)})$.)
But this yields exactly \eqref{6.2}. The result for a general $m(\tau_I)$ is precisely analogous.
To compute $N$, suppose $\Gamma$ has exactly $k$ core subgraphs $\Gamma'\subsetneq \Gamma$. (This means that $P(\Gamma)$ will have $k$ exceptional divisors $E_i$.)  Consider the commutative ring
\eq{6.3}{R:= \Q[x_1,\dotsc,x_k]/(x_1^2,\dotsc,x_k^2,M_1,\dotsc,M_r),
}
where we think of the $x_i$ as corresponding to exceptional divisors $E_i$ on $P(\Gamma)$, and the $M_j$ are monomials corresponding to empty intersections of the $E_i$. The notation means that we factor the polynomial ring in the $x_i$ by the ideal generated by the indicated elements.  We may if we like drop the $M_j$ from the ideal. This will simply mean the column vector on which $N$ acts will have many entries equal to $0$.  As a vector space, we can identify $R$ with the free vector space on $\sigma_1$ and the $\tau_I$ by mapping $\sigma_1 \mapsto 1$ and $\tau_I \mapsto \prod_{i\in I} x_i$. With this identification, the monodromy map $m$ is given    (compare \eqref{6.2}) by multiplication by $(1-x_1)(1-x_2)\cdots (1-x_k)$. But the map $R \to \text{End}_{\text{vec. sp.}}(R)$ given by multiplication is a homomorphism of rings, so $\log(m)$ is given by (note $x_i^2=0$)
\eq{}{\log\Big((1-x_1)\cdots (1-x_k)\Big) = -\sum x_i.
}
Thus $N$ is the matrix for the map given by multiplication by $-\sum
x_i$. If we ignore the relations $M_j$ and just write the matrix for
the action on $\Q[x_1,\dotsc,x_k]/(x_1^2,\dotsc,x_k^2)$, it has size
$2^k\times 2^k$ and is strictly upper triangular. For $k=3$, the
matrix is $-N^{(3)}$ \eqref{8.10}.

\section{Renormalization: the removal of log-poles}\label{secld}
Recall we have defined $\textrm{sdd}(\Gamma)$, the degree of
superficial divergence of a graph with respect to a given physical
theory, \eqref{sdd}. The choice of the theory determines a differential form
$\omega_\Gamma$ associated to $\Gamma$. We will be interested in the
{\it logarithmic divergent} case, when $\textrm{sdd}(\Gamma)\geq 0$, but $\omega_\Gamma$ has been chosen such that it only has log-poles,
see in particular section \ref{logpolereduct}.  The affine integral in this case will be overall
logarithmically divergent, but this overall divergence can be
eliminated by passing to the associated projective integral. If, for
all core subgraphs $\Gamma' \subset \Gamma$, we have
$\textrm{sdd}(\Gamma')<0$, then the projective integral actually converges
and we are done. If $\Gamma'>0$ for some subgraph, then one is obliged
to manipulate the differential form as described in section \ref{parrep} above.
To simplify notation, from now on we assume that all graphs and subgraphs have $\textrm{sdd}\leq 0$,
while all following lemmas hold similarly for higher degrees of divergence with the appropriate choice of $\omega_{II}$.
Below, we spell all results out for the case $\omega_{II}=\Omega_{2n-1}/\psi_\Gamma^2$, and we set $\psi_\Gamma\equiv \psi(\Gamma)$.
\begin{lem}\label{lem5.1} Let $\Gamma'\subsetneq \Gamma$ be core graphs and assume
  $\textrm{sdd}(\Gamma)=0$. Let $L \subset X_\Gamma \subset
  \P(\Gamma)$ be the coordinate linear space defined by the edges
  occurring in $\Gamma'$. Let $\pi: P_L \to \P(\Gamma)$ be the blowup of
  $L$. Then $\pi^*\omega_\Gamma$ has a logarithmic pole on $E$ if and only if
  $\textrm{sdd}(\Gamma')=0$. Similarly, the pullback of
  $\omega_\Gamma$ to the full core blowup $P(\Gamma)$ (cf. formula \eqref{3.3b}) has a
  log pole of order along the exceptional divisor
  $E_{\Gamma'}$ associated to $\Gamma'$ if and only if  $\textrm{sdd}(\Gamma')=0$.
\end{lem}
\begin{proof}We give the proof for $\phi^4$-theory. Let the loop
  number $|\Gamma|=m$ so the graph has $2m$ edges \eqref{sdd}. Let
\eq{5.1}{\Omega_{2m-1}
  = \sum (-1)^i
  A_idA_1\wedge\cdots\wedge\widehat{dA_i}\wedge\cdots\wedge
 dA_{2m} =
A_{2m}^{2m}d(A_1/A_{2m})\wedge\cdots\wedge
 d(A_{2m-1}/A_{2m}).
}
Then
\eq{5.2}{\omega_\Gamma = \frac{\Omega_{2m-1}}{\psi_\Gamma^2}.
}
Suppose $L:A_1=\cdots=A_p=0$. We can write the graph polynomial
(\cite{BEK}, prop. 3.5)
\eq{}{\psi_\Gamma =
  \psi_{\Gamma'}(A_1,\dotsc,A_p)\psi_{\Gamma/\!/\Gamma'}(A_{p+1},\dotsc,A_{2m})+R
}
where the degree of $R$ in $A_1,\dotsc,A_p$ is strictly greater than
$\deg \psi_{\Gamma'} = |\Gamma'|$. Let $a_i = A_i/A_{2m}$, and let $b_i=a_i/a_p,\
i<p$. Locally on $P$ we can take
$b_1,\dotsc,b_{p-1},a_p,a_{p+1},\dotsc, A_{2m}$ as local coordinates
and write
\eq{5.4}{\omega_\Gamma = \pm
  a_p^{p-2|\Gamma'|}\frac{da_p}{a_p}\wedge\frac{db_1\wedge
    \cdots\wedge da_{2m-1}}{F^2} .
}
Here $F$ is some polynomial in the $a_i$'s and the $b_j$'s which is not divisible by $a_p$. The assertion for the blowup of $L$ follows immediately. The assertion for $P(\Gamma)$ is also clear because we can find a non-empty open set on $\P(\Gamma)$ meeting $L$ such that the inverse images in $P(\Gamma)$ and in $P_L$ are isomorphic.
\end{proof} We want to state the basic renormalization result coming out of our monodromy method. For this, we restrict to the case
\eq{5.5}{\textrm{sdd}(\Gamma') \le 2,\ \forall \Gamma' \subseteq \Gamma,
}
with an understanding that appropriate forms $\omega_{II}(\Gamma')$ have been chosen so that the differential forms has log-poles only.
The following lemma applies then to $\phi^4$-theory. A physicist wishing to apply our results to another theory    needs only check the lemma holds with $\omega_\Gamma$ replaced by the integrand given by Feynman rules.
\begin{lem}\label{lem5.2} Let $\tau^\ve_V$ be the chains on $\P(\Gamma)$ constructed above (section \ref{sectopch}) (including the case $\tau^\ve_{P(\Gamma)}=\sigma_\ve$). Then, assuming \eqref{5.5}, we will have
\eq{}{\Big|\int_{\tau^\ve_V} \omega_\Gamma\Big| = O(|\log |\ve||^k), \ |\ve|\to 0 }
for some $k\ge 0$.
\end{lem} \begin{proof}We first consider the integral for the chain $\sigma_\ve = \tau_{P(\Gamma)}^\ve$. Locally on the blowup $P(\Gamma)$ the integrand will look like \eqref{5.4} but there may be more than one log form; i.e.  $\widetilde \omega da_{p_1}/a_{p_1}\wedge\cdots\wedge da_{p_k}/a_{p_k}$. An easy estimate for such an integral over a compact chain satisfying $a_j \ge \ve$ gives $C(|\log\ve|)^k$.
The integrals over $\tau_V^\ve,\ V\subsetneq P(\Gamma)$ involve first integrating over one or more circles. Locally the chain is an $(S^1)^p$-bundle over an intersection $x_1=\cdots=x_p=0$ in local coordinates. We may compute the integral by first taking residues. $V$ will be the closure of a torus orbit in $P(\Gamma)$ associated to a flag $\Gamma_p \subsetneq\cdots\subsetneq\Gamma_1\subsetneq \Gamma$ (proposition \ref{prop1.6}). We may assume $x_i$ is a local equation for the exceptional divisor in $P(\Gamma)$ associated to $\Gamma_i\subset \Gamma$. By lemma \ref{lem5.1}, our integrand will have a pole on $x_i=0$ if and only if $\textrm{sdd}(\Gamma_i)=0$. (Note that the integrand has no singularities on $\tau_V^\ve$, so we may integrate in any order.) The situation is confusing because $\textrm{sdd}(\Gamma_i)<0\Rightarrow \textrm{sdd}(\Gamma/\!/\Gamma_i)>0$ so one might expect non-log growth in this case. The problem does not arise, because the residue will vanish.
Assuming $\textrm{sdd}(\Gamma_i)=0,\ \forall i$, the residue integral is
\eq{5.7c}{\int_{\prod_j \tau_{P(\Gamma_j/\!/\Gamma_{j+1})}^\ve}\omega_{\Gamma_p}\wedge\cdots \wedge \omega_{\Gamma/\!/\Gamma_1}. } Since $\textrm{sdd}(\Gamma_i/\!/\Gamma_{i+1})=0$, we may simply write \eqref{5.7c} as a product of integrals and argue as above.
\end{proof}
We want now to apply the argument sketched in the introduction to our situation. There is one mathematical point which must be dealt with first. We want to consider $\int_{\sigma_t}\omega_\Gamma$ as a function of $t$. Here we must be a bit careful. For $t=\ve e^{i\theta}$ and $|\theta|<<1$ we are ok, but as $\theta$ grows, our chain may meet $X_\Gamma$. Topologically, we have (proposition \ref{prop6.3}) the chains $\tilde c^{\eta,\ve,\theta}$ which miss $X_\Gamma$ and which represent the correct homology class in $H_*(\P(\Gamma)-X_\Gamma,\Delta_t-X_\Gamma\cap \Delta_t)$, but one must show our integral depends only on the class in homology relative to $\Delta_t$, i.e. $\omega_\Gamma$ integrates to zero over any chain on $\Delta_t-X_\Gamma\cap \Delta_t$. Intuitively, this is because $\omega_\Gamma|\Delta_t=0$, but, because $\Delta_t$ has singularities it is best to be more precise. Quite generally, assume $U$ is a smooth variety of dimension $r$, and $D\subset U$ is a normal crossings divisor (i.e. for any point $u\in U$ there exist local coordinates $x_1,\dotsc,x_r$ near $u$, and $p\le r$ such that $D:x_1x_2\cdots x_p=0$ near $u$). One has sheaves \eq{10.8}{\Omega^q_U(\log D)(-D) \subset \Omega^q_U \subset \Omega^q_U(\log D)
}
where $\Omega^q_U$ is the sheaf of algebraic (or complex analytic; in fact, either will work here) $q$-forms on $X$, and  $\Omega^q_U(\log D)$ is obtained by adjoining locally wedges of differential forms $dx_i/x_i,\ 1\le i\le p$. Locally, $\Omega^q_U(\log D)(-D) := x_1x_2\cdots x_p\Omega^q_U(\log D)$. All three sheaves are easily seen to be stable under exterior differential (for varying $q$). The resulting complexes calculate the de Rham cohomology for $(U,D), U, (U-D)$ respectively, \cite{D}. Note that in the top degree $r=\dim U$ we have
\eq{10.9}{\Omega^r_U(\log D)(-D) =
\sO_U\cdot x_1x_2\cdots x_p\frac{dx_1\wedge\cdots\wedge dx_p}{x_1x_2\cdots x_p}dx_{p+1}\wedge\cdots \wedge dx_r = \Omega^r_U.
}
It follows that we get a maps \eq{10.10}{\Omega^r_U[-r] \to \Omega^*_U(\log D)(-D);\quad \Gamma(U,\Omega^r_U) \to H^r_{DR}(U,D).
}
In particular, taking $U = \P(\Gamma)-X_\Gamma$, we see that integrals $\int_{\text{ch.rel.} \Delta_t}\omega_\Gamma$ are well-defined.
\begin{thm}We suppose given a graph $\Gamma$ such that all core subgraphs $\Gamma'\subseteq \Gamma$ have superficial divergence $\textrm{sdd}(\Gamma')\le 0$ for a given physical theory. Let $\omega_\Gamma$ be the form associated to the given theory. Let $N$ be the upper-triangular matrix of size $K\times K$ described in the previous section, where $K$ is the number of chains of core subgraphs $$\Gamma_p \subsetneq \cdots \subsetneq \Gamma. $$ Then the lefthand side of the expression below is single-valued and analytic for $t$ in a disk about $0$  so the limit \eq{9.11}{\lim_{|t|\to 0} \exp(-N\frac{\log t}{2\pi i})\begin{pmatrix} \int_{\tau^t_{P(\Gamma)}}\omega_\Gamma \\
\vdots \\
 \int_{\tau^t_{V}}\omega_\Gamma \\ \vdots \end{pmatrix} = \begin{pmatrix} a_1 \\ \vdots \\
 a_k \end{pmatrix}
}
exists. \end{thm} \begin{proof} The proof proceeds as outlined in section \ref{ssecmi}. $N$ is chosen to be nilpotent and such that the lefthand side has no monodromy. The lemma \ref{lem5.2} assures that terms have at worst log growth. Since they are single-valued on $D^*$, they extend to the origin. \end{proof}

\begin{rem}
It is time to compare what we are calculating here with what a physicist computes according to Thm.(\ref{cructhm}).
The transition is understood upon noticing that in our constructions of chains, we pick up the residue from each exceptional divisor by computing the
monodromy.
In physics we iterate those residues as iterated integrals. Below the top entry $a_1$ this gives different rational weights to them in according with the scattering type formula of \cite{RHII}. We discuss this below in section (\ref{seclmhsvsren}).
\end{rem}

\begin{defn}\label{topentry} With notation as above, the renormalized value $\int_\sigma \omega_\Gamma$ is the top entry in the column vector $\exp(+N\frac{\log t}{2\pi i})\begin{pmatrix}a_1 \\
\vdots \end{pmatrix}$. \end{defn}

\begin{rmk} Note that the terms $ \int_{\tau^t_{V}}\omega_\Gamma$ on the lefthand side of \eqref{9.11} may be calculated recursively.  As in lemma \ref{lem5.2} above, $V$ corresponds to a flag of core subgraphs of $\Gamma$. As in formula \eqref{5.7c}, the integral dies unless all the $\Gamma_i/\!/\Gamma_{i+1}$ are log divergent. In this case, one gets
\eq{9.12}{(2\pi i)^{p-1}\prod \int_{\tau^t_{P(\Gamma_i/\!/\Gamma_{i+1})}}\omega_{\Gamma_i/\!/\Gamma_{i+1}}.
}

If, in addition, the subquotients $\Gamma_i/\!/\Gamma_{i+1}$ are {\it primitive}, i.e. they are log divergent but have no divergent subgraphs, then the integrals in \eqref{9.12} will converge as $|t| \to 0$. Upto a term which is $O(t)$ and can be ignored in the limit, they may be replaced  by their limits as $t\to 0$. These entries in \eqref{9.11} may then be taken to be constant.
\end{rmk}
\begin{ex}Consider the dunce's cap fig.(\ref{figb}). It has $3$ core subgraphs, but only the $2$-edged graph $\gamma$ with edges $1, 2$ is log divergent. Thus, the column vector in \eqref{9.11} has $4$ entries, but only $2$ are non-zero. Dropping unnecessary rows and columns, the matrix $N = \begin{pmatrix} 0 & -1 \\ 0 & 0\end{pmatrix}$. The constant entry in the column vector is
\eq{}{2\pi i\int_{\sigma_\gamma} \frac{\Omega_1}{\psi_\gamma^2}\int_{\sigma_{\Gamma/\!/\gamma}}\frac{\Omega_1}{\psi_{\Gamma/\!/\gamma}^2} = 2\pi i\Big(\int_0^\infty\frac{da}{(a+1)^2}\Big)^2 = 2\pi i.
}
\end{ex}
    It remains to connect $\int_\sigma \omega_\Gamma$ to the physicists computation.
\subsection{lMHS vs $\Phi_R$}\label{seclmhsvsren}
Let us understand how the period matrix $p^T=(a_1,a_2,\cdots,a_r)$ which we have constructed connects to the coefficients $c_j$                                                                                                                                            \be \Phi_{\textrm{MOM}}(\Gamma)(q^2/\mu^2)=\sum_{j=1}^r c_j(\Gamma) \ln^j q^2/\mu^2.\ee
Going to variables \bea
& & t_\Gamma,a_1,\ldots, a_{|\Gamma^{[1]}|},\sum a_i=1,\nonumber\\ & & t_1, b_1,\ldots, b_{|{\Gamma_1}^{[1]}|},\sum b_i=1,\nonumber\\
 & & \ldots,\nonumber\\ & &  t_p, z_1,\ldots, b_{|{\Gamma_p}^{[1]}|},\sum z_i=1,\nonumber \eea for a chain of core graphs $\Gamma_p\subsetneq\cdots\subsetneq\Gamma_1\subsetneq\Gamma$ gives, for each such  flag and constant lower boundaries $\epsilon$, an iterated integral over
\be \int_\epsilon^\infty dt\int_{\epsilon/t}^\infty dt_1\cdots\int_{\epsilon/t/t_1\cdots/t_{p-1}}^\infty dt_p.\ee
As the integral has a logarithmic pole along any $t_i$ integration, the difference between integrating against the chains, which only collect the coefficients of $\ln \epsilon$ for each such integral, and the iteration above is a factorial for each flag. A summation over all flags established the desired relation using tree factorials \cite{Chen}:\\
As the entries in the vector $(a_1,\cdots)^T$ are in one-to-one correspondence with forests of $\Gamma$, identifying $a_1$ with the empty forest, we can write the top-entry defined in Defn.(\ref{topentry}) as
\be \sum_{\textrm{[for]}}\left(\frac{\ln t}{2\pi i}\right)^{|\textrm{[for]}|}a_{\textrm{[for]}},\label{topentryc}\ee
where
\be a_{\textrm{[for]}}=p_1(\Gamma/\!/\textrm{[for]})\prod_j p_1(\gamma_j),\ee
using the notation of Eqs.(\ref{fornot},\ref{pone}).
Then,
\be \partial_{\ln t}\Phi_{\textrm{MOM}}(\Gamma)(t)=\sum_{\textrm{[for]}}\textrm{aug}(\Gamma)\left(\frac{\ln t}{\textrm{[for]}^*!}\right)^{|\textrm{[for]}|}a_{\textrm{[for]}}.\ee
Here, $\textrm{[for]}^*!$ is a forest factorial defined as follows. Any forest $\textrm{[for]}$ defines a tree $T$ and a collection of edges $C$ such that $P^C(T)$ and $R^C(T)$ denote the core sub- and co-graphs in question. The complement set $T^{[1]}/C$ defines a forest $\cup_i t_i$ say.
We set $\textrm{[for]}^*!=\prod_i t_i!$, for standard tree factorials $t_i!$ \cite{Chen}.
For example, comparing the two graphs
\be \Gamma_1=\grapha,\;\Gamma_2=\graphb,\ee
we have the two vectors
\be\left(\begin{array}{l}
  p_1\left(\grapha\right) \\
  p_1\left(\graphc\right)p_1\left(\graphd\right) \\
  p_1\left(\graphd\right)p_1\left(\graphc\right) \\
  p_1\left(\graphc\right)p_1\left(\graphc\right)p_1\left(\graphc\right)
\end{array}\right)\ee
and
\be\left(\begin{array}{l}
  p_1\left(\graphb\right) \\
  p_1\left(\graphc\right)p_1\left(\graphd\right) \\
  p_1\left(\graphc\right)p_1\left(\graphd\right) \\
  p_1\left(\graphc\right)p_1\left(\graphc\right)p_1\left(\graphc\right)
\end{array}\right).\ee
Hence, we find the same $\ln^2 t$ term upon computing Eq.(\ref{topentryc}) for the monodromy.

On the other hand, the tree factorials deliver 1/2 for that term in the case of $\Gamma_1$, and $1$ for $\Gamma_2$, while we get $2$ in both cases for the
term $\sim \ln t$. Indeed, the flag \be \graphc\subsetneq\graphd\subsetneq\grapha\ee corresponds to a tree with two edges. The term $\sim \ln^2 t$
comes from the cut $C$ which corresponds to both of these edges.  The complement is the empty cut, whose tree factorial is $3!$ simply.
As we took a derivative with respect to $\ln t$, we get a factor of $\textrm{aug}(\Gamma)=3$, which leaves us with a factor $3/3!=1/2$.

For $\graphb$, we note that the tree factorial is $3$ instead of $3!$ (we have two flags instead of one), which leaves us with a factor $1$.

\subsection{Limiting Mixed Hodge Structures}\label{seclmhs} In this final paragraph at the suggestion of the referee we outline the structure of a limiting mixed Hodge structure associated to a variation of mixed Hodge structure and how it might apply to the Feynman graph amplitudes.

Let $\Gamma$ be a log divergent graph with $n$ loops and $2n$ edges. The graph hypersurface $X_\Gamma:\psi_\Gamma=0$ is a hypersurface in $\P^{2n-1}$, and the Feynman integrand represents a cohomology class
\eq{9.23}{\Big [\frac{\Omega}{\psi_\Gamma^2}\Big ]\in H^{2n-1}(\P^{2n-1}-X_\Gamma,\C) = H^{2n-1}(\P^{2n-1}-X_\Gamma,\Q)\otimes\C = H_\C = H_\Q\otimes \C.
}
The cohomology group has a {\it mixed Hodge structure}, which means there are defined two filtrations: \newline\noindent (i) The {\it weight} filtration $W_*H_\Q$ which is defined over $\Q$ and increasing. It looks like
\eq{}{0 \subset W_{2n}H_\Q \subset W_{2n+1}H_\Q \subset \cdots \subset W_{4n-2}H_\Q = H_\Q.
}
Blowing up on $X_\Gamma$ so it becomes a normal crossings divisor $D_*$, there is a spectral sequence relating the graded pieces $W_{2n-1+i}/W_{2n-1+i-1}$ to the Tate twist by $-i$ of the cohomology in degree $2n-1-i$ of the codimension $i-1$ strata of $D$. (So, for example, $gr^W_{2n}$ is related to $\oplus_j H^{2n-2}(D_j)(-1)$ where $D=\bigcup D_j$.) \newline\noindent (ii) The {\it Hodge} filtration $F^*H_\C$ which is defined over $\C$ and decreasing:
\eq{}{(0) \subset F^{2n-1} \subset F^{2n-2} \subset \cdots \subset F^{1} = H_\C.
}

The filtrations are subject to the compatibility condition that the filtration
\eq{}{F^p(gr^W_q\otimes \C) := F^pH_\C\cap W_q\otimes \C\Big/F^pH_\C\cap W_{q-1}\otimes \C
}
is the Hodge filtration of a pure Hodge structure of weight $q$. (This is simply the condition that $F^*gr^W_q\otimes \C$ be $q$-opposite to its complex conjugate, i.e. that $gr^W_q\otimes \C = F^p\oplus \overline F^{q-p+1}$ for any $p$.)

Let us say that a class $\omega \in H_\C$ has Hodge level $p$ if $\omega \in F^pH_\C- F^{p+1}H_\C$. An important problem is to determine the Hodge level of the Feynman form \eqref{9.23}. One may speculate that the Hodge level of the Feynman form equals the {\it transcendental weight} of the period. (The transcendental weight of a multizeta number $\zeta(n_1,\dotsc,n_p)$ is the sum of the $n_i$.) For example, in \cite{BK} one finds many examples of Feynman amplitudes of the form $*\zeta(N)$ where * is rational. In all known cases $N=2n-3$. To estimate the Hodge level, one may use the pole order filtration \cite{D}, 3.12. One blows up on $X_\Gamma \subset \P^{2n-1}$ to replace $X$ by a normal crossings divisor $D= \bigcup_{i=1}^r D_i$. Let  $\omega$ on $\P^{2n-1}-X_\Gamma$ be a $(2n-1)$-form and let $I \subset \{1,\dotsc,r\}$ be the indices $i$ such that $\omega$ has a pole along $D_i$. Write $p_i+1$ for the order of this pole, with $p_i \ge 0$.  Then the Hodge level of $\omega$ is $\ge 2n-1-\sum p_i$.  (For a more precise statement, see op. cit.) For example, if $X_\Gamma$ is smooth (this happens only when $n=1$) one would get $p_1=1$ so the Hodge level would be $\ge 2n-2$.
\begin{prop} For the Feynman form, at least $2$ of the $p_i\ge 1$. The pole order calculation thus suggests the Hodge level of the Feynman form above is $\le 2n-3$.
\end{prop}
\begin{proof}The situation for $n=1$ is trivial, so we assume $n\ge 2$. The space of symmetric $n\times n$-matrices has dimension $d:=\frac{n(n+1)}{2}$. Let $\P^{d-1}$ be viewed as the projectivized space of such matrices, so a point corresponds to a matrix upto scale. The determinant of the universal matrix defines a hypersurface $\sX \subset \P^{d-1}$. More generally, we define $\sX_p \subset \P^{d-1}$ to be the locus where the rank of the corresponding symmetric matrix is $\le n-p$. We have $\sX = \sX_1$, and it is easy to see that $\sX_p$ has codimension $\frac{p(p+1)}{2}$ in $\P^{d-1}$. Points in $\sX_p$ will have multiplicity $\ge p$ on $\sX$.

There is an inclusion $\rho: \P^{2n-1} \inj \P^{d-1}$ such that $X_\Gamma = \sX\cap \P^{2n-1}$. Points of $\sX_2\cap \P^{2n-1}$ will have multiplicity $\ge 2$ in $X_\Gamma$ and codimension $\le 3$ in $\P^{2n-1}$ . This means that in the local ring on $\P^{2n-1}$ at a general point of $\sX_2\cap \P^{2n-1}$, there will be functions $x_1, x_2, x_3$ which form part of a system of coordinates on $\P^{2n-1}$ such that a local defining equation $\psi$ for $X_\Gamma$ lies in $(x_1,x_2,x_3)^2$. We may construct our normal crossings divisor $D$ as above by first blowing up $\sX_2\cap \P^{2n-1}$ in   $\P^{2n-1}$. Subsequent blowups will not affect the pole order, which may be computed at the generic point of the exceptional divisor $E$. We have
\eq{}{\frac{dx_1dx_2dx_3\cdots}{\psi^2} = \frac{x_1^2 dx_1d(x_2/x_1)d(x_3/x_1)\cdots}{x_1^4\phi(x_1,x_2/x_1,x_3/x_1,\ldots)}.
}
It follows that the Feynman form has a double pole on $E$ as well as a double pole on the strict transform of $X_\Gamma$ in the blowup. \end{proof}
\begin{rmk} (i) To give a complete proof that the Hodge level is $\le 2n-3$ one would have to show the double order pole was not killed by an exact form. \newline\noindent
(ii) It would be exciting to be able to say something about the weight filtration on $H^{2n-1}(\P^{2n-1}-X_\Gamma)$. \newline\noindent
(iii) The data in \cite{BK} suggests that double zetas which occur will have transcendental weight $2n-4$. For example, the bipartite graph $\Gamma$ consisting of the $12$ edges joining sets of $3$ and $4$ vertices has Feynman amplitude a rational multiple of $\zeta(3,5)$. In general, a calculation as above shows $\sX_3\cap \P^{2n-1}$ has multiplicity $\ge 3$ and codimension $\le 6$. If one could show that for the bipartite $\Gamma$ that this codimension drops to $5$, then the same argument as above would yield $3$ poles with $p_i\ge 1$, suggesting a Hodge level $2n-4$.
\end{rmk}

Next we should consider the mixed Hodge structure necessary for the relative period calculation. Recall \eqref{3.3b} we work in a toric blowup $P=P(\Gamma) \to \P^{2n-1}$. Let $B \subset P$ be the complement of the big toric orbit in $P$. It is the  union of the strict transform of the coordinate divisor $\Delta \subset \P^{2n-1}$ and the exceptional divisors. Let $Y \subset P$ be the strict transform of $X_\Gamma$. The relevant cohomology group is the middle group in the sequence
\eq{9.28}{H^{2n-2}(B-Y\cap B; \Q) \to H^{2n-1}(P-Y,B-Y\cap B; \Q) \to H^{2n-1}(P-Y, \Q).
}
If all the subgraphs $\Gamma' \subsetneq \Gamma$ have $\text{sdd}(\Gamma')<0$, then renormalization is unnecessary. The Feynman amplitude as we have defined it is simply a period of the mixed Hodge structure \eqref{9.28}. The weight filtration for the group on the left involves the cohomology of the strata of the normal crossings divisor $B$. For example, we have an exact sequence
\eq{}{H^0(B_{(1)}-Y\cap B_{(1)},\Q) \to H^0(B_{(0)}-Y\cap B_{(0)},\Q) \to W_0H^{2n-2}(B-Y\cap B_{(0)},\Q).
}
Here we write $B_{(i)}$ for the disjoint union of the components of the strata of dimension $i$.
We know from corollary \ref{cor3.3} that $Y\cap B_{(0)}=\emptyset$, and a bit of thought about the combinatorics of $B_{(i)}, i=0,1$ reveals that $W_0H^{2n-2}(B-Y\cap B,\Q) = \Q(0)$. This gives a map of the trivial Hodge structure $\Q(0)$ to our period motive:
\eq{}{\Q(0) \to H^{2n-1}(P-Y,B-Y\cap B; \Q).
}
When the period is a rational multiple of $\zeta(2n-3)$ we expect that there is a map of Hodge structures $\Q(3-2n) \to H^{2n-1}(\P^{2n-1}-X_\Gamma,\Q)$
and that the extension of $\Q(3-2n)$ by $\Q(0)$ associated to $\zeta(2n-3)$ is a subquotient of \eqref{9.28}.

Finally the main focus of this paper has been the renormalization case when one or more proper subgraphs of $\Gamma$ has $\text{sdd}=0$. In this case, the Feynman form will have a pole along one or more divisor in $B$, so \eqref{9.28} is no longer the relevant Hodge structure. In this case, we work with the limiting mixed Hodge structure $H_{lim}$ associated to $H_t:=H^{2n-1}(\P^{2n-1}-X_\Gamma,\Delta_t-\Delta_t\cap X_\Gamma)$. Let $D$ be a small disk around $t=0$, and let $D^*=D-\{0\}$. Then $H_{D^*} = \bigcup_{t\neq 0} H_t$ becomes a local system on $D^*$. Let $\sH_{D^*} = H_{D^*} \otimes \sO_{D^*}$ be the corresponding analytic bundle. If we untwist by the monodromy, we get a trivial local system ($h=\dim H_t$)
\eq{9.31}{\C^h_{D^*} \cong \exp(-N\log t)H_{D^*} \subset \sH_{D^*}.
}
Since this local system is trivial, it extends (trivially) across $t=0$. It also has a canonical $\Q$-structure defined from the $\Q$-structure at any point $t_0 \neq 0$. The analytic bundle $\sH_{D^*}$ has a Hodge filtration $F^*\sH_{D^*}$ coming from the Hodge filtrations on the $H_t$. (Note the Hodge filtration is not horizontal, so there is no Hodge filtration on the local system $H_{D^*}$.) From \eqref{9.31} we get a canonical trivialization of the analytic bundle $\sH_{D^*} \cong \sO_{D^*}^h$ and hence a canonical extension across $t=0$. One can show \cite{CK}, 2.1(i) that the Hodge filtration extends across $t=0$ as well.

Thus, on the fibre $H_0$ we have a Hodge filtration and a $\Q$-structure. If you think in terms of periods, i.e. using the pairing $H_{0,\Q}^\vee \times H_0 \to \C$, the above description of the Hodge filtration as a limit across $t=0$ coincides with the computation \eqref{9.11}. What we have not given is the weight filtration. This monodromy weight or limiting weight filtration is more subtle, essentially being determined by the endomorphism $N$ together with the given weight filtrations on the fibres $H_t$. We hope that the computation of $N$ in this paper will help to understand this structure, but at the moment the weight  structures on the $H_t$ are not well enough understood to say more. For the general theory, the interested reader is referred to \cite{CK} and the references cited there.

\newpage \bibliographystyle{plain} \renewcommand\refname{References}

\end{document}